\documentclass[aps,pra,twocolumn,superscriptaddress,superscriptaddress]{revtex4-2}
\pdfoutput=1
\usepackage{comment}
\usepackage{amsthm}
\usepackage{amssymb}
\usepackage{amsmath}
\usepackage{amsfonts}
\usepackage{mathtools}
\usepackage{physics}
\usepackage{float}
\usepackage{xcolor}
\usepackage{balance}
\usepackage{graphicx}
\usepackage{dcolumn}
\usepackage{bm}
\usepackage[hidelinks]{hyperref}
\usepackage[capitalize]{cleveref}

\newcommand{\be}{\begin{equation}}
\newcommand{\ee}{\end{equation}}
\newcommand{\ba}{\begin{aligned}}
\newcommand{\ea}{\end{aligned}}

\newcommand{\R}{\mathbb{R}}
\newcommand{\bc}{\begin{center}}
\newcommand{\ec}{\end{center}}
\newcommand{\beq}{\begin{equation}}
\newcommand{\eeq}{\end{equation}}
\newcommand{\beqq}{\begin{equation*}}
\newcommand{\eeqq}{\end{equation*}}
\newcommand{\beqa}{\begin{align}}
\newcommand{\eeqa}{\end{align}}
\newcommand{\barr}{\begin{array}}
\newcommand{\earr}{\end{array}}
\newcommand{\bi}{\begin{itemize}}
\newcommand{\ei}{\end{itemize}}

\newcommand{\E}{\mathbb{E}}
\newcommand{\C}{\mathbb{C}}

\newtheorem{lem}{Lemma}
\newtheorem{theo}{Theorem}

\newtheorem{defi}{Definition}
\newtheorem{protocol}{Protocol}
\newtheorem{res}{Result}
\DeclareMathOperator{\N}{\mathbb{N}}
\DeclareMathOperator{\I}{\mathbb{I}}

\begin{document}

\title{Symplectic coherence: a measure of position-momentum correlations in quantum states}
\author{Varun Upreti}
 \email{varun.upreti@inria.fr}
\author{Ulysse Chabaud}
\email{ulysse.chabaud@inria.fr}
\affiliation{%
 DIENS, \'Ecole Normale Sup\'erieure, PSL University, CNRS, INRIA, 45 rue d’Ulm, Paris, 75005, France
}%

\date{\today}

\begin{abstract}
The interdependence of position and momentum, as highlighted by the Heisenberg uncertainty principle, is a cornerstone of quantum physics. Yet, position-momentum correlations have received little systematic attention. Motivated by recent developments in bosonic quantum physics that underscore their relevance in quantum thermodynamics, metrology, and computing, we establish a general framework to study and quantify position-momentum correlations in quantum states. We introduce \textit{symplectic coherence}, a faithful and easily computable measure defined as the Frobenius norm of the block of the covariance matrix encoding position-momentum correlations, and demonstrate that symplectic coherence is monotone under relevant operations and robust under small perturbations. Furthermore, using a recent mapping by Barthe et al.~\cite{Barthe2025} which relates the covariance matrix of a bosonic state to the density matrix of a finite-dimensional system, we show that position-momentum correlations correspond to beyond-classical correlations in a virtual finite-dimensional quantum state, with symplectic coherence mapping naturally to geometric quantum discord. Taking energy constraints into account, we determine the maximal position-momentum correlations achievable at fixed energy, revealing structural insights about the corresponding optimal states. Finally, we illustrate the operational relevance of symplectic coherence through several examples in quantum information tasks and quantum thermodynamics. In the process, we establish new technical results on matrix norms and quantum covariance matrices, and demonstrate the conceptual significance of viewing covariance matrices as density matrices of virtual quantum states.
\end{abstract}
\maketitle

\section{Introduction}
Despite its century-long history and a recent surge of interest driven by computational applications, quantum physics continues to pose profound conceptual challenges. A central question concerns the features that fundamentally distinguish quantum systems from classical ones. Within the framework of quantum information theory, these distinguishing features—commonly referred to as quantum resources \cite{Howard2017,Chitambar2018,albarelli2018resource,Takagi2019,Amaral2019,thomas2024}—are what enable the long-sought quantum advantage \cite{shor1994algorithms,Aaronson2013,Ronnow2014,Harrow2017,Boixo2018} in computation, communication, metrology, and simulation. From a foundational physics perspective, elucidating these features is equally crucial for advancing our understanding of the fundamental principles that govern our world.

The Heisenberg uncertainty principle \cite{Heisenberg1927} lies at the foundation of quantum mechanics, imposing a bound on the product of position and momentum variances, and thus capturing the inherent trade-off in how “spread out” a wavefunction can be along these two conjugate variables. As such, it highlights the interdependence of position and momentum of a quantum system.

Position–momentum correlations also play a significant role in the physics of bosonic quantum systems. For instance, in \cite{Serafini2007entanglement}, the authors gave numerical evidence that micro-canonical ensembles of quantum states with non-zero position–momentum correlations are systematically more entangled than those with no position–momentum correlations. Recent theoretical and experimental work in bosonic systems has also highlighted the importance of such correlations: they can serve as a potent resource in quantum metrology, leading to precision surpassing what is achievable with uncorrelated Gaussian probes \cite{Porto2025}; moreover, these correlations markedly influence the open-system dynamics of bosonic quantum states, affecting decoherence rates and the purity of these states \cite{dasilva2024}.

From a computational perspective, bosonic systems have recently gained attention as promising platforms for quantum computing, supported by experimental advances such as the deterministic generation of large-scale entangled states \cite{yokoyama2013ultra}, robust quantum error correction schemes \cite{sivak2023}, and the potential of these bosonic quantum computers to possibly outperform discrete-variable quantum computing architectures \cite{chabaud2024complexity,brenner2024factoring,upreti2025bounding}. Recent findings also indicate that universal bosonic computations are classically simulable when the quantum gates involved do not produce position–momentum correlations, regardless of the circuit’s other physical features \cite{upreti2025interplay}.

In the context of quantum information processing, the uncertainty principle alone does not provide a complete picture. While it reveals the interdependence between position and momentum in a quantum system, it leaves open the question of ``how much'' position–momentum correlations are present in the system, and the quantification of this correlation between conjugate observables—specifically, position \(q\) and momentum \(p\)—has received relatively little systematic attention.
 
In light of these developments, this work is the first attempt, to the best of our knowledge, to quantify position-momentum correlations in a quantum state. While for a certain class of bosonic quantum states (Gaussian states), the covariance matrix fully captures position–momentum correlations, the characterization of some other bosonic quantum states (non-Gaussian states) also require higher-order position-momentum correlations. Here we focus on correlations described by the covariance matrix, as these are the easiest to measure in experiments. More concretely, we ask the question: Looking at the covariance matrix of a bosonic quantum system, how to rigorously quantify its position-momentum correlations?

To answer this question, we provide a new measure termed \textit{symplectic coherence}, which possesses a clear operational interpretation in terms of distance to the set of states with no position-momentum correlations. Moreover, using a recent mapping introduced in \cite{Barthe2025}, which allows one to map a covariance matrix to the density matrix of a (virtual) finite-dimensional quantum state, we show that position-momentum correlations in a bosonic quantum state can be understood as a form of beyond-classical correlations known as quantum discord \cite{Ollivier2001} in this virtual state. Through this mapping, symplectic coherence is found to naturally correspond to geometric quantum discord \cite{Brukner2010}.

Further, we perform a resource-theoretic treatment of position-momentum correlations, establishing the class of free states, free and non-free operations, and proving the faithfulness of symplectic coherence, and its monotonicity under free operations.
Taking into account the experimental constraint of finite available energy, we determine the maximum symplectic coherence that a quantum state can achieve under a fixed energy budget. This reveals that, to maximise position–momentum correlations, all available energy should be concentrated into a single mode, with the remaining modes in the vacuum state, followed by appropriate passive linear unitaries. As an application of this result, we further analyse the robustness of symplectic coherence, demonstrating its stability under small perturbations.

We conclude by demonstrating how symplectic coherence provides a measure for assessing diverse quantum information processing tasks and physical phenomena: it enhances the quantum Fisher information in displacement amplitude estimation, governs efficiency in channel discrimination protocols, and influences the behavior of entanglement within thermodynamic ensembles. These examples collectively highlight the role of position–momentum correlations, as captured by symplectic coherence, in various facets of quantum physics.

The rest of the paper is organized as follows. In Section \ref{sec:prelims}, we introduce the required preliminaries, definitions, and notations. Section \ref{sec:pm_measure} introduces \textit{symplectic coherence}, a measure of position-momentum correlations in qunatum states. In Section \ref{sec:discord_and_sc}, we provide an operational interpretation of symplectic coherence and reveal its link to geometric quantum discord. 
In Section \ref{sec:msc}, we determine the maximum symplectic coherence attainable by finite‑energy states, characterize the optimal states that achieve it and demonstrate how it allows us to analyze the robustness of symplectic coherence under perturbations. Section \ref{sec:sc_qi} illustrates several quantum information processing tasks where symplectic coherence plays a role. We conclude in Section \ref{sec:conclusion} with a summary of our results and directions for future research.


\section{Preliminaries}\label{sec:prelims}

We refer the reader to \cite{NielsenChuang} for background on quantum information theory and to \cite{Braunstein2005,ferraro2005,Weedbrook2012} for continuous variable (CV) quantum information material. Hereafter, the sets $\N, \R$ and $\C$ are the set of natural, real, and complex numbers respectively, with a * superscript when $0$ is removed from the set.

A density matrix is a positive semi-definite, trace-one operator acting on a Hilbert space $\mathcal H$, and it represents the state of a quantum system. In discrete-variable (DV) quantum information, the Hilbert space has finite dimension $n$, which we denote by $\mathcal H_n$. Throughout this paper, we use $\varrho$ and $\varsigma$ to denote DV density matrices in finite-dimensional spaces. In contrast, CV quantum information theory, which is used to describe bosons, deals with density matrices in infinite-dimensional Hilbert spaces, and we denote CV quantum states by $\rho$ and $\sigma$ throughout.


\subsection{CV quantum information}\label{subsec:prelims_CVQI}

In CV quantum information, a qumode or simply mode refers to a degree of freedom associated with a specific quantum field of a CV quantum system, such as a single spatial or frequency mode of light, and is the equivalent of a qubit in the CV setting. In this paper, $m \in \N^*$ denotes the number of modes in the system, $\ket0$ refers to the single-mode vacuum state, and $\hat{a}$ and $\hat{a}^\dagger$ refer to the single-mode annihilation and creation operator respectively, satisfying $[\hat{a},\hat{a}^\dagger] = \I$. These are related to the position and momentum quadrature operators as
\begin{equation}
    \hat{q} = \hat{a} + \hat{a}^\dagger, \hspace{5mm}\hat{p} = -i(\hat{a} - \hat{a}^\dagger),
\end{equation} 
with the convention $\hbar = 2$. Furthermore, $\hat{q}$ and $\hat{p}$ satisfy the canonical commutation relation $[\hat{q},\hat{p}] = 2i\I$. The multi-mode particle number operator is given by
\begin{equation}
    \hat N = \sum_{i=1}^m \hat a_i^\dagger \hat a_i = \sum_{i=1}^m \left(\frac{\hat q_i^2 + \hat p_i^2}{4} -\frac12\right),
\end{equation}
where the subscript refers to the mode, while the energy operator is defined as
\begin{equation}
    \hat E = \hat N + \frac{m}{2}\I = \sum_{i=1}^m \left(\frac{\hat q_i^2 + \hat p_i^2}{4}\right).
\end{equation}

Product of unitary operations generated by Hamiltonians that are quadratic in the quadrature operators of the modes are called Gaussian unitary operations, and states produced by applying a Gaussian unitary operation to the vacuum state are Gaussian states. The action of an $m$-mode Gaussian unitary operation $\hat{G}$ on the vector of quadratures 
\begin{equation}\label{eqn:vec_quadratueres}
    \boldsymbol{\Gamma}= (\hat{q}_1,\dots,\hat{q}_m,\hat{p}_1,\dots,\hat{p}_m)^T
\end{equation}
is given by
\begin{equation}\label{eqn:symplectic_transform_quadrature}
    \hat{G}^\dagger \boldsymbol{\Gamma} \hat{G} = S \boldsymbol{\Gamma} + \boldsymbol{d},
\end{equation}
where $S$ is a $2m \times 2m$ symplectic matrix and $\boldsymbol{d} \in \R^{2m}$ is a displacement vector. Single-mode {displacement} operators and single-mode {squeezing} operators are defined as $\hat{D}(\alpha) = e^{\alpha \hat{a}^\dagger - \alpha^* \hat{a}}$ and $\hat{S} (\xi) = e^{\frac12 (\xi \hat{a}^2 - \xi^* \hat{a}^{\dagger 2})}$ respectively, where $\alpha, \xi \in \C$. In this paper, we take $\xi = r e^{i\theta}$ to be real ($\theta = 0$) with $r>0$, such that $\ket r = \hat S (r) \ket 0$ is the single-mode squeezed state squeezed along $p$ direction, and the symplectic matrix associated with $\hat S(r)$ is given by
\begin{equation}
    \begin{bmatrix}
        e^r && 0 \\
        0 && e^{-r}
    \end{bmatrix}.
\end{equation}

{Block-diagonal orthogonal gates} $\hat O$ are $m$-mode Gaussian gates with the associated symplectic matrix in the block-diagonal form
\begin{equation}
    S_O =\begin{bmatrix}
        O && 0 \\
        0 && O
    \end{bmatrix},
\end{equation}
where $O$ is a $m \times m$ orthogonal matrix, and zero displacement vector. {Phase shifters} are single-mode Gaussian gates with associated symplectic matrix 
\begin{equation}
    S_\theta = \begin{bmatrix}
        \cos \theta && \sin \theta \\
        -\sin\theta && \cos \theta
    \end{bmatrix},
\end{equation}
and zero displacement vector. Passive linear unitary gates $\hat U$ are $m$-mode Gaussian gates with associated symplectic matrix
\begin{equation}
    S_U = \begin{bmatrix}
        X && Y \\
        -Y && X
    \end{bmatrix},
\end{equation}
such that $S_U$ is orthogonal, and zero displacement vector. Through the Cartan (KAK) decomposition, any passive linear unitary gate can be decomposed into two block-diagonal orthogonal gates and $m$ single-mode phase shifter gates \cite[Section V.A]{edelman2022cartan}
\begin{equation}
    \hat U = \hat O_2 (\hat{R}(\theta_1)\otimes\dots\otimes\hat{R}(\theta_m)) \hat O_1.
\end{equation}
The Bloch--Messiah decomposition \cite{ferraro2005} states that any symplectic matrix $S$ can be decomposed in terms of symplectic matrix of two passive linear unitary gates $S_{U_1}$ and $S_{U_2}$ and the symplectic matrix of $m$ single-mode squeezing gates
\begin{equation}
    S = S_{U_1} \begin{bmatrix}
        Z && 0 \\
        0 && Z^{-1}
    \end{bmatrix} S_{U_2},
\end{equation}
where $Z = \mathrm{diag}(e^{r_1},\dots,e^{r_m})$, with $r_i \geq 0$, $\forall i$.

For a quantum state $\rho$, its covariance matrix $V^\rho$ is given by
\begin{equation}
    V^\rho = \frac{1}{2}\Tr[\{\Delta \bm \Gamma_\rho, \Delta \bm \Gamma_\rho^T\}\rho],
\end{equation}
where 
\begin{equation}
    \Delta \bm \Gamma_\rho \coloneqq \bm \Gamma - \Tr[\bm \Gamma \rho],
\end{equation}
and $\{.,.\}$ represents the anti-commutator between two operators $\hat A$ and $\hat B$:
\begin{equation}
    \{\hat A, \hat B\} = \hat A \hat B + \hat B \hat A.
\end{equation}
We represent $V^\rho$ in a block-matrix form 
\begin{equation}
    V^\rho = \begin{bmatrix}
        V_x^\rho && V_{xp}^\rho \\
        (V_{xp}^\rho)^T && V_p^\rho
    \end{bmatrix},
\end{equation}
where $V_x^\rho$ and $V_p^\rho$ are $m \times m$ matrices representing position-position and momentum-momentum correlations, respectively, while $V_{xp}^\rho$ is the $m \times m$ matrix representing the cross-correlations between position and momentum quadratures. The covariance matrix of a quantum state satisfies the following properties:
\begin{eqnarray}\label{eq:covar_inequality}
 V^\rho + i\Omega &\succeq& 0, \nonumber \\
 V^\rho &\succ& 0, \nonumber \\
    V_x^\rho \succ 0 &,& V_p^\rho \succ 0, \nonumber \\
    \Tr[V^\rho] &\geq& 2m,
\end{eqnarray}
where 
\begin{equation}
    \Omega = \begin{bmatrix}
        \bm 0_m && \I_m \\
        -\I_m && \bm 0_m
    \end{bmatrix},
\end{equation}
where $\bm 0_m$ and $\I_m$ are $m \times m$ zero and identity matrices respectively. The first two inequalities are standard in quantum information, and can be found for example in \cite{Wilde2019}. $V_x^\rho \succ 0$ and $V_p^\rho \succ 0$ follows because the principal submatrices of a positive definite matrix are always positive definite. Finally, we have
\begin{equation}
    \Tr[\rho \hat N] = \frac{\Tr[V^\rho] + \bm d^T \bm d}{4} - \frac{m}{2},
\end{equation}
where $\bm d$ is the vector of expectation values of single degree quadratures:
\begin{equation}\label{eqn:displ_vector}
    \bm d = (\Tr[\rho \hat q_1],\dots,\Tr[\rho \hat q_m], \Tr[\rho \hat p_1],\dots,\Tr[\rho \hat p_m])^T.
\end{equation}
This directly implies (every covariance matrix $V^\rho$ can be associated to a state with zero first moment vectors by applying displacement gates)
\begin{equation}\label{eq:lowerboundtracecov}
    \Tr[V^\rho]\ge2m,
\end{equation}
as well as:
\begin{equation}\label{eqn:energy_covariance matrix}
    \Tr[\rho \hat E] = \frac{\Tr[V^\rho] + \bm d^T \bm d}{4}.
\end{equation}
 
Under the action of a Gaussian unitary gate $\hat G$ with associated symplectic matrix $S$, the covariance matrix transforms as
\begin{equation}
    V^\rho \rightarrow S V^\rho S^T.
\end{equation}

The Williamson decomposition of a covariance matrix satisfying Eq.~\ref{eq:covar_inequality} gives
\begin{equation}
    V^\rho = S \nu S^T,
\end{equation}
where $S$ is a $2m \times 2m$ symplectic matrix, where
\begin{equation}
    \nu = \mathrm{diag}(\nu_1,\dots,\nu_m,\nu_1,\dots,\nu_m)
\end{equation}
where $\nu_i \geq 1, \forall i$. For a pure Gaussian state $\nu_i =1, \forall i$ and with the Bloch--Messiah decomposition, its covariance matrix can be written as
\begin{equation}\label{eqn:covar_pure_gaussian}
    V^\rho = S_U \begin{bmatrix}
        Z^2 && 0 \\
        0 && Z^{-2} 
    \end{bmatrix} S_U^T,
\end{equation}
where $S_U$ is the symplectic matrix associated to some passive linear unitary $\hat U$, and $Z = \mathrm{diag}(e^{r_1},\dots,e^{r_m})$.

Finally, given any non-Gaussian quantum state, there exists a mixed Gaussian state with the same covariance matrix. This is called \textit{Gaussification}.


\subsection{Quantum discord}
\label{sec:discord}

For a bipartite quantum state $\varrho_{AB}$, quantum discord \cite{Ollivier2001} is a measure of the non-classical correlations between the subsystems $A$ and $B$, and is defined as
\begin{equation}\label{eq:discord}
    D(\varrho_{AB})_{\Pi_A} = \min_{\Pi_A} (\mathcal I(B:A) - \mathcal J(B:A)_{\{\Pi_A\}}),
\end{equation}
where $\mathcal I(B:A)$ is the mutual information between subsystems $A$ and $B$, whereas $\mathcal J(B:A)_{\{\Pi_A\}}$ can be informally defined as a measure of information gained about $B$ by making measurements on $A$ using the set of POVMs $\{\Pi_A\}$. In classical information theory, the notions of $\mathcal I(B:A)$ and $\mathcal J(B:A)_{\{\Pi_A\}}$ are equivalent, but this is not the case in quantum information theory, with even some non-entangled states showing non-zero discord \cite{Ollivier2001}. Quantum discord has since found various applications 
\cite{Streltsov_2015}.

The evaluation of quantum discord from Eq.~\ref{eq:discord} is highly non-trivial except for the simplest cases. This has led to the introduction of geometric quantum discord, defined as \cite{Roga_2016}
\begin{equation}\label{eq:gqd}
    D_G (\varrho_{AB}) = \min_{\varsigma_{AB} \in \mathcal{CQ}} d^2(\varrho_{AB},\varsigma_{AB}),
\end{equation}
where $\mathcal{CQ}$ is the set of ``classical-quantum'' states (states with zero quantum discord with respect to measurements in the subsystem $A$), and $d$ is a distance measure. The common distance measures used are the Hilbert--Schmidt distance, the trace distance and the Hellinger distance.

Note that, while the standard definition of quantum discord includes minimization over all possible measurement bases on $A$, we use hereafter a restricted version of quantum discord where the measurement basis is fixed to be in the computational basis. In appendix \ref{app:rdiscord}, we formally define this restricted version of quantum discord, and derive some additional properties of geometric quantum discord for this definition.


\subsection{Matrix norms}\label{subsec:prelims_subsec_matrix_norms}

A matrix norm $||.||$ is a function that assigns a non-negative real number to a given matrix $A$. Common properties followed by matrix norms are that they
\begin{itemize}
    \item are unitarily invariant, i.e.
    \begin{equation}
        ||UA|| = ||A||,
    \end{equation}
    for a unitary matrix $U$;
    \item follow the triangle inequality, i.e.
    \begin{equation}
        ||A+B|| \leq ||A|| + ||B||.
    \end{equation}
    \item are zero if and only if the matrix $A$ is a zero matrix, i.e.
    \begin{equation}
        ||A|| = 0 \text{ if and only if } A = \bm 0.
    \end{equation}
\end{itemize}
Three matrix norms which we use in this paper are the Frobenius norm ($||.||_F$), the trace norm ($||.||_1$) and the operator norm ($||.||_\infty$), all satisfying these properties.

The \textit{Frobenius norm} of an $m \times m$ matrix $A$ is defined as
\begin{equation}\label{eq:FM}
    ||A||_F = \sqrt{\sum_{i,j=1}^m |A_{ij}|^2} = \sqrt{\Tr[A^\dagger A]} = \sqrt{\sum_{i=1}^m s_i^2 (A)},
\end{equation}
where $s_i$ are singular values of $A$ and can be expressed in terms of the eigenvalues $\lambda_i$ of $A^\dagger A$ as $s_i (A) = \sqrt{\lambda_i(A^\dagger A)}, \forall i$. The Hilbert--Schmidt distance between two matrices $A$ and $B$ is defined as
\begin{equation}
    d_{HS}(A,B) = ||A - B||_F.
\end{equation}
The \textit{trace norm} of $A$ is defined as
\begin{equation}
    ||A||_1 = \Tr\left[\sqrt{A^\dagger A}\right] = \sum_{i=1}^m s_i (A),
\end{equation}
and is related to the Frobenius norm as
\begin{equation}\label{eq:FM_TM}
    ||A||_F \leq||A||_1 \leq \sqrt{m}||A||_F.
\end{equation}
The trace distance between two quantum states with density matrices $\rho$ and $\sigma$ is given in terms of the trace norm as
\begin{equation}
    \frac12 ||\rho - \sigma||_1.
\end{equation}
Finally, the \textit{operator norm} of $A$ is defined as
\begin{equation}
    ||A||_\infty = \max_i s_i(A),
\end{equation}
and is related to the Frobenius norm as
\begin{eqnarray}\label{eq:FM_OM}
    ||A||_\infty&\leq& ||A||_F \leq \sqrt m ||A||_\infty.
\end{eqnarray}


\subsection{Quantum Resource Theories}

Quantum resources are physical properties that govern the performance of specific quantum information tasks and are thus regarded as valuable within these contexts. Quantum resource theories \cite{Howard2017,Chitambar2018,albarelli2018resource,Takagi2019,Amaral2019,thomas2024} offer a rigorous mathematical framework for the characterization, quantification, and manipulation of these resourceful properties.

A central idea in quantum resource theories is the separation of quantum states into two categories: \textit{free states}, which lack the physical feature (resource) of interest, and \textit{resourceful states}, which possess it. When defining a resource measure—a non-negative function quantifying this property—an important requirement is that the measure should be zero if and only if the state is free. This property is known as the \textit{faithfulness} of the measure. Another desirable property is (sub-)additivity under tensor product \cite{Chitambar2018,albarelli2018resource}.

In the study of resource manipulation, quantum operations are also classified as \textit{free operations} and \textit{resourceful operations}. Free operations are those that do not increase the resource content of a quantum state, mapping free states to free states, while resourceful operations may potentially enhance it. In this context, \textit{monotonicity} under free operations is another important criterion for a resource measure: it states that the measure’s value should not increase when free operations are applied to a quantum state.

With the necessary preliminaries detailed, we introduce a measure of position-momentum correlations in the next section.


\section{Measure of position-momentum correlations}\label{sec:pm_measure}

To define a measure of position-momentum correlations, we first define the set of free states, denoted by $\mathcal{C}$:
\begin{defi}[Free states for position-momentum correlations]\label{defi:free_states}
   The set of free states with respect to position-momentum correlations, $\mathcal C$ is defined as
   \begin{equation}
    \mathcal C \coloneqq \{\sigma : V_{xp}^\sigma = 0\}.
    \end{equation}
\end{defi}

\noindent This set consists in the states without position-momentum correlations as captured by their covariance matrix.

\subsection{Symplectic coherence}

Given the covariance matrix of a quantum state $\rho$,
\begin{equation}
    V^\rho = \begin{bmatrix}
        V_x^\rho && V_{xp}^\rho \\
        (V_{xp}^\rho)^T && V_p^\rho
    \end{bmatrix},
\end{equation}
since the position-momentum correlations are encoded in $V_{xp}^\rho$, the measure of position-momentum correlations should be a function $f: \R^{m\times m} \mapsto \R$ of $V_{xp}^\rho$, where $\R^{m\times m}$ is the set of real $m \times m$ matrices. Further $f$ should be faithful, that is,
\begin{equation}
    f(A) = 0 \text{  if and only if } A = \bm 0_m,
\end{equation}
where $\bm 0_m$ is the $m \times m$ zero matrix. One natural way to construct faithful measures of correlation is via matrix norms (see Section~\ref{sec:prelims}). Among these, the Frobenius norm stands out for its simplicity: it is simple to compute (Eq.~\ref{eq:FM}), since it only requires element‐wise squaring and summation, whereas other norms would require singular value decomposition routines to find the singular values. Its simple definition also makes derivations of its properties more straightforward. Finally, since one can directly measure the entries of a covariance matrix in the lab, $||V_{xp}||_F$ is an immediately computable quantity.

Motivated by these advantages, we define a measure of position–momentum correlations—termed \textit{symplectic coherence}—by the squared Frobenius norm of $V_{xp}^\rho$:

\begin{defi}[Symplectic coherence] \label{defi:sc_expr} Given a quantum state $\rho$ with covariance matrix
\begin{equation}
    V^\rho = \begin{bmatrix}
        V_x^\rho && V_{xp}^\rho \\ (V_{xp}^\rho)^T && V_p^\rho
    \end{bmatrix},
\end{equation}
the symplectic coherence of $\rho$, denoted by $\mathfrak c_\rho$, is defined as:
    \begin{equation}\label{eq:sc_expr_org}
  \mathfrak c_\rho \coloneqq || V_{xp}^\rho ||_F^2.
\end{equation}
\end{defi}

Note that symplectic coherence of quantum states defined here is the different from the notion of symplectic coherence of quantum gates introduced in \cite{upreti2025interplay}. Whereas symplectic coherence of quantum gates is defined as their ability to mix position and momentum quadratures, symplectic coherence of quantum states, defined in Eq.~\ref{eq:sc_expr_org} is a quantifier of the amount of position-momentum correlations in them.

Moreover, although we focus on the Frobenius norm, standard inequalities between matrix norms (Section~\ref{subsec:prelims_subsec_matrix_norms}) and properties followed by the commonly studied matrix norms in the literature (unitary invariance and triangle inequality) allow us to make statements about alternative matrix norm choices.

\subsection{Operational interpretation of symplectic coherence}

The following result gives an operation interpretation of symplectic coherence:
\begin{theo}[Operational interpretation of symplectic coherence]\label{theo:sc_expr}
    Given a quantum state $\rho$ with covariance matrix
    \begin{equation}
        V^\rho = \begin{bmatrix}
            V_x && V_{xp} \\
            V_{xp}^T && V_p
        \end{bmatrix},
    \end{equation}
    its symplectic coherence is equal to the minimal Hilbert--Schmidt distance from the set of states with the set of free states $\mathcal C$, namely
    \begin{equation}
       \mathfrak c_\rho = \min_{\sigma \in \mathcal C} \frac12 ||V^\rho - V^\sigma||_F^2 = \min_{\sigma \in \mathcal C} \frac12 \Tr[(V^\rho - V^\sigma)^2].
    \end{equation}
\end{theo}
\noindent The proof of Theorem \ref{theo:sc_expr} is given in appendix \ref{app:sc_exp} and consists of two steps: first, we find $\sigma \in \mathcal C$ whose covariance matrix is closest in Hilbert--Schmidt distance to $V^\rho$ and then we calculate the distance of $V^\rho$ to the said state to obtain the final expression. In this proof step, we also prove the rather non-trivial result:
\begin{lem}
    Given the covariance matrix of a quantum state
    \begin{equation}
        V^\rho = \begin{bmatrix}
            V_x^\rho && V_{xp}^\rho \\
            (V_{xp}^\rho)^T && V_p^\rho
        \end{bmatrix},
    \end{equation}
    the matrix obtained by removing the position-momentum correlations,
    \begin{equation}
        \tilde V = \begin{bmatrix}
            V_x^\rho && 0\\
            0 && V_p^\rho
        \end{bmatrix},
    \end{equation}
    represents a valid covariance matrix of a quantum state.
\end{lem}

\subsection{Properties of symplectic coherence}\label{sec:fmr_sc}

As explained in the preliminaries, desirable properties of a measure of a physical resource are: additivity, faithfulness and monotonicity. In this section, we demonstrate that symplectic coherence satisfies these properties.

From \cref{defi:sc_expr} the symplectic coherence is additive under tensor product, since the covariance matrix of a tensor product of quantum states is the direct sum of the covariance matrices of the individual quantum states, i.e.\
\begin{equation}
    \mathfrak c_{\rho\otimes\sigma}=\mathfrak c_{\rho}+\mathfrak c_{\sigma}.
\end{equation}
Moreover, we obtain the following result:
   \begin{theo}[Faithfulness and monotonicity of symplectic coherence]\label{theo:monotonicity_sc}
        Given a quantum state $\rho$, its symplectic coherence $\mathfrak c_\rho = 0$ if and only if $\rho \in \mathcal C$ (Definition \ref{defi:free_states}). Further, the symplectic coherence is non-increasing under the following operations:
        \begin{itemize}
            \item Block-diagonal orthogonal unitary gates $\hat O$.
            \item Displacement unitary gate.
            \item Tensor product with free states.
            \item Partial traces.
            \item Classical mixing of zero first moment states.
        \end{itemize}
    \end{theo}
\noindent While the faithfulness of symplectic coherence directly follows from the property of matrix norms, the proof of monotinicity is given in appendix \ref{app:monotonicity_sc}, and utilizes properties of covariance matrices and Frobenius norms. It is natural to expect that these operations should be free for any measure of position-momentum correlations, as we will explain in what follows: block-diagonal orthogonal transformations and displacement operations should qualify as free operations because matrix norms are invariant under unitary (and hence orthogonal) transformations. Specifically, block-diagonal orthogonal gates transform $ V_{xp}^\rho $ as $O V_{xp}^\rho O^T$ for some orthogonal matrix $O$, while displacement operations leave $V_{xp}^\rho$ unchanged. Given that any quantifier of position-momentum correlations is expected to be a matrix norm of $V_{xp}^\rho$, it follows directly that block-diagonal orthogonal gates and displacement operations do not alter this quantifier, and hence are indeed free operations.

Since making a tensor product of states does not generate additional position-momentum correlations than the ones already present in the individual states, a measure of position-momentum correlations should simply be additive under tensor product, which is what we obtain with symplectic coherence. 

Taking partial traces removes some of the elements of the covariance matrix, so the value of any quantifier of position-momentum correlations should decrease, which is what we observe. Finally, covariance matrix obtained by classical mixing of states with zero first moment of quadratures is simply a linear combination of the covariance matrices. The triangle inequality of matrix norms indicates that gives that no position-momentum correlation quantifier should increase under classical mixing of such zero first moment states.

Since phase shifters can mix position and momenturm quadratures, they can possibly increase position-momentum correlations. In appendix \ref{app:sc_non_free}, we also give examples of cases where post-selecting on heterodyne measurements, homodyne measurements and classical mixing of states with non-zero first moments can increase the symplectic coherence.

An important class of non-free operations that we would like to highlight here are the block-diagonal active Gaussian unitary gates, whose symplectic matrices are of the form
\begin{equation}
    \begin{bmatrix}
        A && 0 \\ 
        0 && (A^T)^{-1}
    \end{bmatrix},
\end{equation}
with $A$ being any real, invertible and non-orthogonal $m \times m$ matrix. These gates do not mix position-momentum quadratures and hence from the point of view of gate-based symplectic coherence defined in \cite{upreti2025interplay}, they do not possess symplectic coherence. However, they can, in fact, increase the symplectic coherence of a quantum state, as formalized in the following Theorem:
\begin{theo}[Position-momentum correlation monotonicity no-go]\label{theo:active_increase_sc}
    There is no faithful measure of position-momentum correlations which is non-increasing under all block-diagonal Gaussian symplectic operations.
\end{theo}
\noindent The proof of Theorem \ref{theo:active_increase_sc} is given in appendix \ref{app:active_increase_sc}, and follows from looking at how the block-diagonal Gaussians change the covariance matrices, and using properties of matrices. Therefore, even though these gates do not mix position and mometum quadratures, they can in fact increase the amount of position-momentum correlations in the state.

Inspired by the class of free operations defined by Theorem \ref{theo:monotonicity_sc}, we define orthogonal Stinespring dilations:
\begin{defi}[Orthogonal Stinespring dilation]\label{defi:ortho_Stinespring}
    Given a quantum state $\rho$, an orthogonal Stinespring dilation of $\rho$ is defined as
    \begin{equation}
        \Tr_B[\hat D \hat{O} (\rho \otimes \sigma^{(E)}) \hat O^\dagger \hat D^\dagger],
    \end{equation}
    where $\hat O$ are block-diagonal orthogonal gates, $\hat D$ is the displacement operator, and $\sigma \in \mathcal C$.
\end{defi}
\noindent Examples of orthogonal Stinespring dilation include pure photon loss and thermal loss channels. The following is then a standard resource-theoretic statement \cite{albarelli2018resource}:
\begin{lem}[Impossibility of deterministic conversion with orthogonal Stinespring dilations]
    Given two quantum states $\rho$ and $\sigma$ with $\mathfrak c_\rho < \mathfrak c_\sigma$ it is impossible to deterministically convert $\rho$ to $\sigma$ through orthogonal Stinespring dilations.
\end{lem}

\noindent In the following section, we describe a mapping originally proposed in \cite{Barthe2025}, which represents covariance matrices as density matrices of virtual qubit–qudit systems. We then show how this mapping enables us to establish connections between symplectic coherence and quantum discord.
\section{Symplectic coherence and quantum discord}\label{sec:discord_and_sc}
In this section, we establish a connection between symplectic coherence and quantum discord \cite{Ollivier2001}. We start by connecting the set of free states $\mathcal C$ with the set of classical-quantum states (with zero quantum discord \cite{Piani2008}) of a qubit-qudit system with respect to computational basis measurements of the qubit. To do so, we rely on a mapping introduced in \cite{Barthe2025}, which allows us to identify a covariance matrix of size $2m \times 2m$ with the density matrix of a $2m$-dimensional qubit-qudit state.

\begin{defi}[Covariance matrix to density matrix mapping \cite{Barthe2025}]\label{defi:Covariance_to_density} 
    Given a covariance matrix
    \begin{equation}
        V = \begin{bmatrix}
            V_x && V_{xp} \\
            V_{xp}^T && V_p
        \end{bmatrix},
    \end{equation}
    we define a mapping $\mathcal M: V \mapsto \varrho$ which maps $V$ onto the density matrix $\varrho$ in the $2m$-dimensional qubit-qudit Hilbert space $\mathcal H_2 \otimes \mathcal H_m$ where the state of the qubit $\ket{0}$ or $\ket{1}$ denotes the position or momentum quadratures and the state of the qudit corresponds to the mode number (Figure \ref{fig:covar_to_density}). Including the normalization of the obtained virtual state, this mapping gives
    \begin{equation}\label{eqn:covariance_to_density}
        \mathcal{M}\left(V = \begin{bmatrix}
            V_x && V_{xp} \\
            V_{xp}^T && V_p
        \end{bmatrix}\right) = \varrho = \frac{1}{\Tr[V]} \begin{bmatrix}
            V_x && V_{xp} \\
            V_{xp}^T && V_p
        \end{bmatrix},
    \end{equation}
    such that $V_x$ denotes the block for $\ket{0}\bra{0} \otimes \ket{\bm i}\bra{\bm j}$, $V_{xp}$ for $\ket{0}\bra{1} \otimes \ket{\bm i}\bra{\bm j}$, etc., where $\ket{\bm i}$ are the qudit computational basis states.
\end{defi}
\begin{figure}
    \centering
    \includegraphics[width=\linewidth]{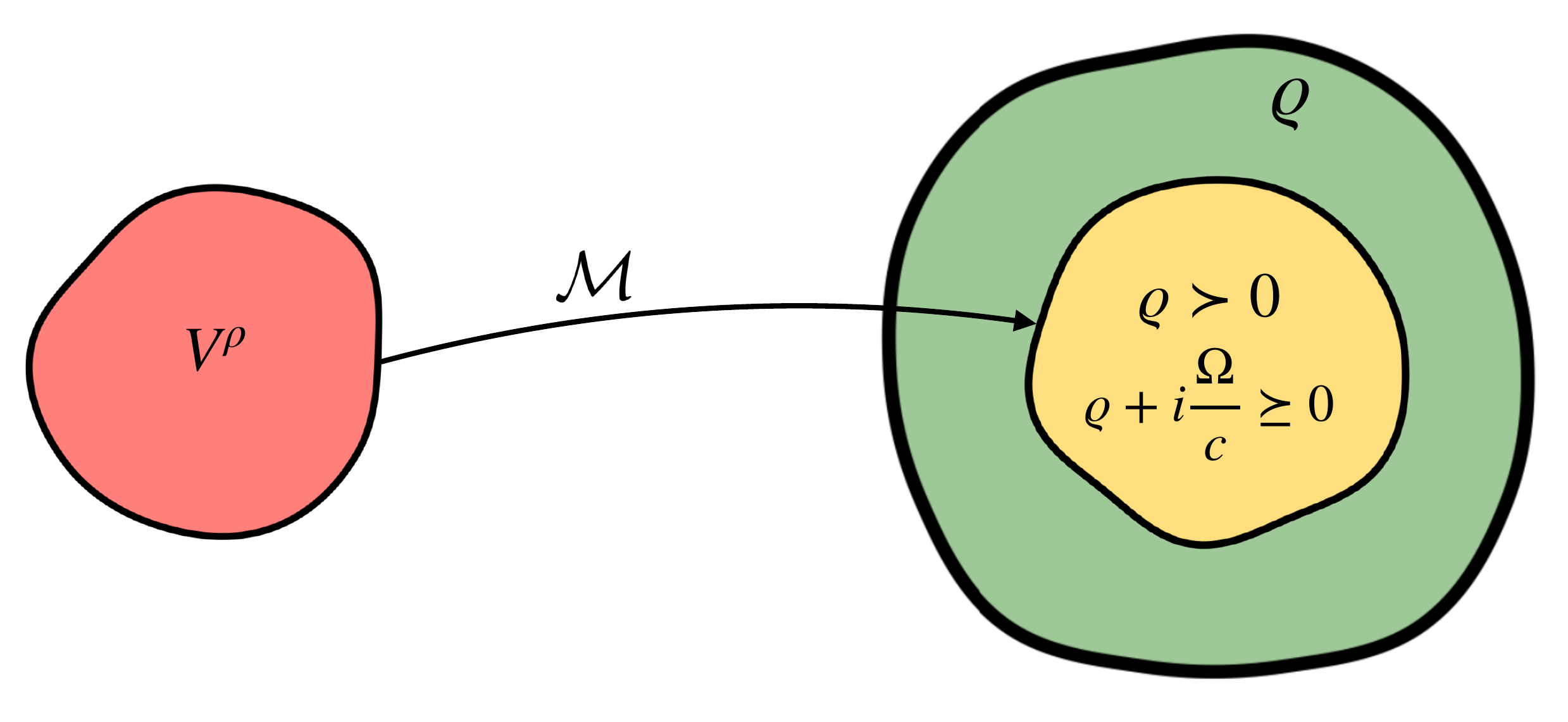}
    \caption{Visual illustration of the covariance matrix to density matrix mapping $\mathcal{M}$ (Definition \ref{defi:Covariance_to_density}) which maps a covariance matrix $V^\rho$ to the density matrix of a virtual qubit-qudit state $\varrho$ such that $\varrho \succ 0$ and $\varrho + i \Omega/c \succeq 0$ for some constant $c \geq 2m$.}
    \label{fig:covar_to_density}
\end{figure}
\noindent Under this mapping, we obtain the following result:
\begin{theo}[Position-momentum correlations and quantum discord]\label{theo:zero_sc_zero_discord}
    Under the mapping $\mathcal{M}$ given in \cref{defi:Covariance_to_density}, a covariance matrix $V^\rho$ with $\rho \in \mathcal C$ maps to a classical-quantum virtual qubit-qudit state $\varrho = \mathcal M(V^\rho)$ with zero quantum discord with respect to computational basis measurements of the qubit register. Furthermore, the symplectic coherence of $\rho$ can be related to geometric quantum discord (\cref{eq:gqd}),with Hilbert-Schmidt distance as the distance measure, with respect to computational basis measurements in the qubit subsystem of $\varrho$ as:
\begin{equation}\label{eqn:sc_gd}
    \mathfrak c_\rho = \frac{\Tr[V^\rho]^2}2 D_G^{HS}(\varrho) = 8\left(\Tr[\rho \hat E] - \bm d^T \bm d\right)^2 D_G^{HS}(\varrho).
\end{equation}
\end{theo}
The proof of Theorem \ref{theo:zero_sc_zero_discord} is given in appendix \ref{app:zero_sc_zero_discord}, and consists in proving that the virtual qubit-qudit states obtained from covariance matrices with zero position-momentum correlations through the mapping given in \cref{defi:Covariance_to_density} have zero quantum discord with respect to computational basis measurements of the qubit \cite[Section I]{Streltsov2017}. The connection between symplectic coherence $\mathfrak c_\rho$ and geometric quantum discord $D_G^{HS}(\varrho)$ can further established by expressing $D_G^{HS}(\varrho)$ in terms of the off-block-diagonal matrix of $\varrho$, and using the mapping $\mathcal{M}$ to identify the corresponding block of $V^\rho$ that maps to this off-block-diagonal part.

It must be emphasized, however, that \cref{theo:zero_sc_zero_discord} does not imply that position-momentum correlations are equivalent to quantum discord in CV systems, unlike Gaussian quantum discord \cite{Giorda2010}. Indeed, while Gaussian quantum discord aims to quantify the non-classical correlation between subsystems in Gaussian states, we are considering correlations between the position and momentum quadratures of a single quantum system. Theorem \ref{theo:zero_sc_zero_discord} instead connects position-momentum correlations in CV quantum states to quantum discord in a virtual (qubit-qudit) state over two subsystems.

Importantly, in defining our set $\mathcal{CQ}$ we restrict to classical–quantum states with respect to measurements in the computational basis of the qubit, rather than allowing all possible measurement bases.  This deliberate choice ensures a one‐to‐one correspondence between non-zero off‐diagonal entries $\varrho_{01}$ (which map to $V_{xp}^\rho$ under $\mathcal{M}$) and non-zero quantum discord.  Concretely, if $\varrho_{01}\neq0$, then position–momentum correlations are present and our restricted discord measure will detect them.

This choice also explains why symplectic coherence avoids the problem one encounters with geometric quantum discord with Hilbert--Schmidt distance as the distance measure. 
Indeed, as pointed out in \cite{Piani2012}, geometric quantum discord with Hilbert--Schmidt distance as the distance measure suffers from logical inconsistencies. Specifically, while the quantum discord of a DV quantum state $\varrho$ with respect to measurements in the subsystem of $\varrho$ should not change when taking its tensor product with another state $\varsigma$, the authors exhibit in \cite{Piani2012} an example where the geometric quantum discord with respect to Hilbert--Schmidt distance actually increases. However, Theorem \ref{theo:monotonicity_sc} demonstrates that we do not encounter any such logical inconsistencies with symplectic coherence, as symplectic coherence is additive under tensor product.

Finally, we would like to elucidate the differences between symplectic coherence in CV quantum states and quantum discord in virtual qubit-qudit states, as this will be important later when we discuss maximal symplectic coherence in the following section. While the density matrix of a quantum state may have complex entries and has to satisfy two conditions
\begin{eqnarray}
    \rho &=& \rho^\dagger \nonumber \\
    \rho &\succeq& 0,
\end{eqnarray}
the virtual density matrix obtained from the inverse of the mapping $\mathcal M$ (Definition \ref{defi:Covariance_to_density}) has only real matrix elements and has to satisfy an additional condition:
\begin{eqnarray}
    \rho &=& \rho^T \nonumber \\
    \rho &\succ& 0 \nonumber \\
    \rho + i\Omega/c &\succeq& 0 \text{ for some constant } c \geq 2m.
\end{eqnarray}
The last condition comes from the uncertainty principle. This difference manifests itself in the value of maximal symplectic coherence (see \cref{theo:max_sc}), which is different from what one would expect when looking purely from a quantum discord point of view. Another difference is that not all Gaussian unitary gates can be mapped to valid unitary gates acting on the virtual qubit-qudit system, a point that was also mentioned in the original paper defining the mapping \cite{Barthe2025}. However, note that block diagonal orthogonal operations map to local unitaries acting on the qudit register. For more details, see the discussion following Eq.~\ref{app_eq:block_diagonal_orthogonal} of appendix \ref{app:monotonicity_sc}.


\section{Maximal symplectic coherence}\label{sec:msc}

In this section, we characterise the maximum amount of position-momentum correlations in states of bounded energy. According to Eq.~\ref{eqn:energy_covariance matrix},
\begin{equation}
    \Tr[\rho \hat E] = \frac{\Tr[V^\rho] + \bm d^T \bm d}{4},
\end{equation}
where $\bm d$ is the first  moment vector given by Eq.~\ref{eqn:displ_vector}. Therefore, sticking to states with zero first moments, i.e. states with $\bm d$ being a zero vector, fixing the trace of the covariance matrix is equivalent to fixing the average energy of the state. In what follows, we consider the set of states with covariance matrices $V^\rho$ such that $\Tr[V^\rho] = E$, and investigate the maximal symplectic coherence obtainable by such states. Note that for states with non-zero first moments, fixing the trace of $V^\rho$ still puts an upper bound on the average energy of the system, since $\bm d^ T \bm d \geq 0$ in Eq.~\ref{eqn:energy_covariance matrix}.
We first define a class of pure Gaussian states $\mathcal S^{\max}(E)$:
\begin{defi}\label{defi:max_SC_states}
     A pure Gaussian state $\ket\psi_G$ belongs to the set $\mathcal S^{\max}(E)$ if and only if it can be written in the form
    \begin{equation}
        \ket{\psi}_G = \hat O_2\hat R(\vec \theta)\hat{O}_1\ket{r}\otimes\ket{0}^{m-1},
    \end{equation}
    up to arbitrary displacements at the end, where $\ket r$ is a squeezed state with squeezing parameter $r$ such that
    \begin{equation}\label{eq:focus_squeeze}
        e^{2r} + e^{-2r} = E - 2(m-1),
    \end{equation}
    and where the phase shift angles $\vec \theta$ and the block-diagonal orthogonal unitary gate $\hat O_1$ with the associated symplectic matrix 
    \begin{equation}
        \begin{bmatrix}
            O && 0 \\
            0 && O
        \end{bmatrix}
\end{equation}
with orthogonal matrix $O$ are defined such that for all $j=1\dots m$,
\begin{equation}
    2 |A_{1j}| = \delta_{1j},
\end{equation}
for $A\in\{ O^T C O,O^T S O,O^T(CS)O\}$, where $C\coloneqq \mathrm{diag}(\cos \theta_1,\dots,\cos \theta_m)$ and $S \coloneqq \mathrm{diag}(\sin \theta_1,\dots,   \sin \theta_m)$. The block-diagonal orthogonal unitary gate $\hat O_2$ can be arbitrary.
\end{defi}
\noindent We then have the following result for maximal symplectic coherence:
\begin{theo}[Maximal symplectic coherence of quantum states]\label{theo:max_sc}
    Given an $m$-mode quantum state $\rho$ with covariance matrix $V^\rho$ such that $\Tr[V^\rho] = E$, the symplectic coherence of $\rho$ is upper bounded as
    \begin{eqnarray}\label{eq:max_sc}
        \mathfrak{c}(\rho) &\leq& \frac{(E-2m)^2}{4} + (E-2m) \nonumber \\ &\equiv& \mathfrak c_\mathrm{max}(E,m).
    \end{eqnarray}
    Moreover, the set of pure Gaussian states of maximal symplectic coherence is given by $\mathcal{S}^{\max}(E)$ (Definition \ref{defi:max_SC_states}), while the covariance matrix of set of all states with maximal symplectic coherence is contained in the convex hull of covariance matrices of states in $\mathcal{S}^{\max}(E)$.
\end{theo}
The proof of Theorem \ref{theo:max_sc} is given in appendix \ref{app:max_sc}, and consists of two steps: firstly, we prove that the maximal symplectic coherence of pure Gaussian states is given by $\mathfrak c_\mathrm{max}(E,m)$ and find the structure of maximally symplectic coherent pure Gaussian states. This step of the proof is rather technical and relies on mathematical results regarding matrix decompositions and matrix norms. Secondly, we prove that the maximal symplectic coherence of mixed state is less than equal to $\mathfrak c_\mathrm{max}(E,m)$ and find the conditions on the the covariance matrix of such maximally symplectic coherent mixed states, by using Jensen's inequality and a result from \cite[Theorem 3]{parthasarathy2011} stating that any mixed Gaussian state can decomposed as an equally weighted mixture of two pure Gaussian states. 

Note that the maximal value achievable by the geometric quantum discord $D_G$ of a qubit-qudit density matrix with respect to computational basis measurements of the qubit is equal to $1/2$ (see Appendix \ref{app:rdiscord}). With Eq.~\ref{eqn:sc_gd} we have
\begin{equation}
    D_G^{HS}(\varrho) = \frac{2}{E^2} \mathfrak c_\rho \leq \frac12,
\end{equation}
leading to $\mathfrak c_\rho \leq E^2/4$. Theorem \ref{theo:max_sc} proves that this is not a tight bound and underlines the lack of complete equivalence between symplectic coherence and geometric quantum discord: not all qubit-qudit density matrices are valid covariance matrices, and in particular those with maximal geometric discord are not.

The structure of maximally \cref{theo:max_sc} not only allows one to engineer states of maximal symplectic coherence in the lab, such as
\begin{equation}
    \ket{MSC} = \hat R_1(\pi/4)\ket{r}\otimes\ket{0}^{m-1},
\end{equation}
with $r$ satisfying Eq.~\ref{eq:focus_squeeze}, it also gives an insight on the structure of the states maximizing position-momentum correlations. The composition of block‑diagonal orthogonal unitary gate $\hat O_1$ with phase‑shift angles behaves as expected: the phase shifters, after conjugation with the block-diagonal unitary gate $\hat O_1$ superposes the state’s information--encoded in the first mode, where all the energy resides--equally in the position and momentum quadrature. What appears to be novel, however, is the realization that one can only achieve maximal symplectic coherence by first concentrating all the energy into the first mode (with the remaining modes in the vacuum state) and then applying an appropriate passive linear unitary. To our knowledge, no previous work has explicitly identified this construction.

Note that when using other matrix norms in the definition of the measure of position-momentum correlations, the inequalities given in Eq.~\ref{eq:FM_TM} and Eq.~\ref{eq:FM_OM} allow us to bound their maximal values as well. However, whether the bound is tight and whether the set of states classified by $\mathcal S^{\max}$ also maximizes these other measures remains an open question.

Finally, a remarkable consequence of \cref{theo:max_sc} is the robustness of symplectic coherence under small perturbations:
\begin{theo}[Symplectic coherence under perturbations]\label{theo:SC_perturbations}
    Given two quantum states $\rho$ and $\sigma$ with symplectic coherences $\mathfrak c_\rho$ and $\mathfrak c_\sigma$ such that $\Tr[\hat E^2 \rho],\Tr[\hat E^2 \sigma] \leq E^2$, and $||\rho - \sigma||_1 \leq \epsilon/m$, then 
    \begin{eqnarray}
  |\mathfrak c_\rho - \mathfrak c_\sigma| &\leq& \mathcal O(E^2\sqrt\epsilon).
    \end{eqnarray}
\end{theo}
The proof of Theorem \ref{theo:SC_perturbations} is given in appendix \ref{app:SC_perturbations} together with an explicit bound, and follows from bounding the Hilbert--Schmidt distance between two covariance matrices in terms of their trace distance, combining results from \cite[Theorem 62]{annamele2025symplecticrank} with standard matrix norm inequalities given in \ref{subsec:prelims_subsec_matrix_norms}. 

Theorem \ref{theo:SC_perturbations} implies that symplectic coherence is robust under small perturbations in the quantum state.
Additionally, we show that under the action of a photon loss channel with parameter $\eta$, the covariance matrix of a quantum state $V^\rho$ changes as (Appendix \ref{app:sc_under_loss})
\begin{equation}
    V^\rho \rightarrow \eta V^\rho + (1-\eta)\I_{2m},
\end{equation}
and therefore the symplectic coherence of a quantum state $\rho$ under that photon loss channel evolves as
\begin{equation}
    \mathfrak c_{\Lambda_\eta(\rho)} = \eta^2 \mathfrak c_\rho.
\end{equation}


\section{Symplectic coherence in quantum information processing}\label{sec:sc_qi}

Having developed a framework for quantifying position-momentum correlations, we now consider a variety of quantum information processing tasks where position-momentum correlations, as quantified by symplectic coherence, plays a role. Note that we consider specific examples in which position-momentum correlations are expected to be resourceful, and demonstrate that symplectic coherence appropriately captures these correlations. Once again, we restrict to quantum states with zero first moments for brevity---analogous results may be obtained when including non-zero first moments.

\subsection{Symplectic coherence and entanglement}

We first highlight a result that was numerically observed by the authors in \cite[Figure 2]{Serafini2007entanglement}, i.e.\ that, while looking at micro-canonical ensembles of pure Gaussian states \cite[Section 3.1]{Serafini2007}, states with zero position-momentum correlations are typically less entangled that states with non-zero position-momentum correlations, at fixed energy. We formalise this observation by the following Theorem (Informal version of Theorem \ref{apptheo:typical_entangle}):
\begin{theo}[Position-momentum correlations enhance entanglement, informal version of Theorem \ref{apptheo:typical_entangle}]\label{theo:sc_typical_entangle}
At fixed energy, the entanglement of pure Gaussian states, as measured by the symplectic eigenvalue squared of the reduced covariance matrix of the first  mode, is higher on average for pure Gaussian states with non-zero position-momentum correlations.
\end{theo}
\noindent A formal version along with the detailed proof of Theorem \ref{theo:sc_typical_entangle} is given in appendix \ref{app:sc_typical_entangle}, and relies on proving that on average, the symplectic eigenvalue squared of the reduced covariance matrix of the first mode, is higher for pure Gaussian states with non-zero symplectic coherence using properties of Haar-random unitary and orthogonal matrices.

Since the covariance matrix of a pure Gaussian state is given by (Eq.~\ref{eqn:covar_pure_gaussian}):
\begin{equation}
    V^\rho = S_U \begin{bmatrix}
        Z^2 && 0 \\
        0 && Z^{-2} 
    \end{bmatrix} S_U^T,
\end{equation}
fixing the energy of pure Gaussian states is equivalent to fixing $\Tr[V^\rho]$ through Eq.~\ref{eqn:energy_covariance matrix} (since we are restricting our study to zero first moment states) or equivalently fixing $Z^2$ ($Z^{-2}$), since $S_U$ is orthogonal. Then, we are Haar-randomly sampling the $2m\times2m$ symplectic orthogonal matrix $S_U$ (isomorphic to an $m\times m$ unitary), as it was done in \cite{Serafini2007entanglement} and taking the average of symplectic eigenvalue squared of the reduced covariance matrix of the first mode over these samples. Note that while in \cite[Figure 2]{Serafini2007entanglement}, the entanglement is quantified by the entanglement entropy $h$, the symplectic eigenvalue squared of the reduced covariance matrix of the first mode $\nu^2$ is also a valid measure of entanglement for pure Gaussian states \cite[Section 5]{Serafini2007}, with the two being related by
    \begin{equation}
    h= \frac{(\nu+1)}{2} \ln\left(\frac{\nu+1}{2}\right) - \frac{(\nu-1)}{2} \ln\left(\frac{\nu-1}{2}\right).
\end{equation}

This result may seem counterintuitive, since at first sight the (squared) symplectic eigenvalue of the reduced covariance matrix
\begin{equation}
    \nu^2 = (V_x)_1 (V_p)_1 - (V_{xp})_1^2,
\end{equation}
should decrease for non-zero $(V_{xp})_1$ (non-zero symplectic coherence), where $(V_x)_1, (V_p)_1$ and $(V_{xp})_1$ are elements of the reduced covariance matrix of the first mode. However, the effect of gates generating non-zero symplectic coherence in pure Gaussian states is such that they also focus the energy on the diagonal elements $(V_x)_1 (V_p)_1$ in such a way that on average, the symplectic eigenvalue squared of these states, and as a consequence the entanglement, increases.

Now, we discuss the role of symplectic coherence in two quantum information processing tasks: quantum metrology and quantum channel discrimination.


\subsection{Symplectic coherence and quantum metrology}\label{sec:sc_metro}

Quantum metrology aims to utilize quantum states to probe quantum systems and measure physical parameters with a precision beyond the limits of classical measurement techniques, and is one of the central research themes in quantum information \cite{Giovannetti2011}. The ultimate precision of estimation of a given parameter $r$ with an input probe quantum state is given by the single-shot quantum Cram\'er--Rao bound:
\begin{equation}
    \Delta r^2 \geq \frac{1}{\mathcal F_Q},
\end{equation}
where $\Delta r$ is the variance of $r$, which is distributed according to some probability distribution that depends on the estimation process, and $\mathcal F_Q$ is the quantum Fisher information---the higher the quantum Fisher information, the higher the ultimate precision of estimation.

Here, we discuss the quantum metrology task of single-mode displacement amplitude single-shot estimation: Consider a scenario where a quantum system is acted upon by the displacement operator
\begin{equation}
    \hat D(r) = \exp\left(ir\left(\hat q/\sqrt 2 + \hat p/\sqrt 2\right)\right),
\end{equation}
and we wish to probe the amplitude parameter $r$ (Figure \ref{fig_illustration_metro_cd}). We show the following result:
\begin{figure}
    \centering
    \includegraphics[width=0.8\linewidth]{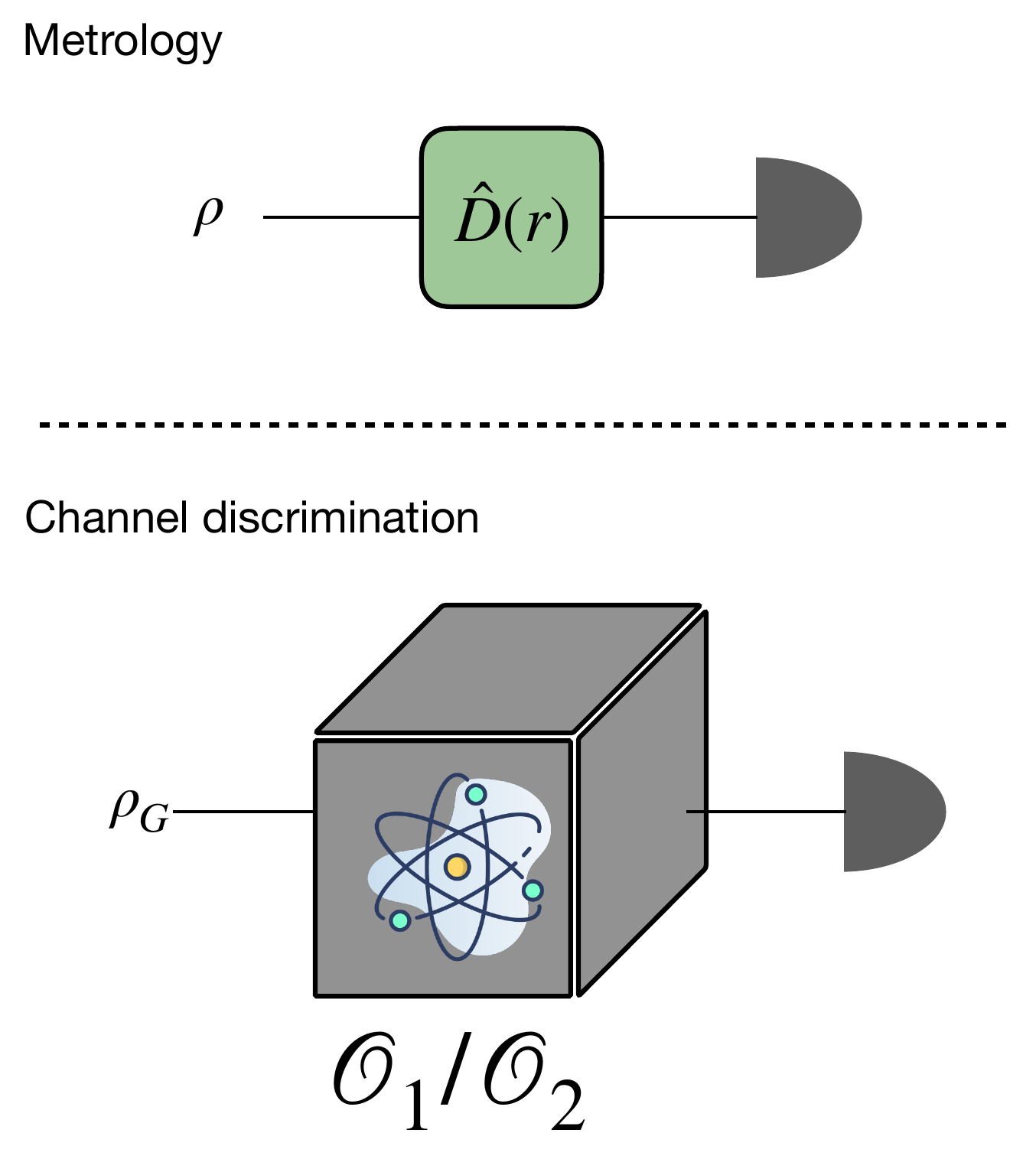}
    \caption{(Top) We are considering the metrological scenario where the input probe state $\rho$ goes through a physical medium implementing a displacment operator $\hat D(r)$ and we wish to estimate $r$ by making measurements at the output. (Bottom) We consider a channel discimination scenario where using an input probe Gaussian state $\rho_G$ and making appropriate measurements at the output, we want to determine which of the give two orthgonal Stinespring dilation (Definition \ref{defi:ortho_Stinespring}) $\mathcal O_1$ or $\mathcal{O}_2$ is the black box acting on the input state implementing.}
    \label{fig_illustration_metro_cd}
\end{figure}
\begin{theo}[Position-momentum correlations enhance metrology]\label{theo:sc_metro}
     Given a pure single-mode input probe state $\rho$ with covariance matrix $V^\rho$ and zero first moments going through a single-mode displacement operator $\hat D(r)$, the Fisher quantum information with respect to $r$ is given by
    \begin{equation}
        \mathcal F_Q \leq 2E + 4V_{xp}^\rho = 2E + 4\sqrt{\mathfrak c_\rho},
    \end{equation}
    where $E = V_x^\rho + V_p^\rho = 4 \Tr[\rho \hat E]$. Further, the inequality becomes an equality for pure input probe states.
\end{theo}
\noindent The proof of Theorem \ref{theo:sc_metro} is given in appendix \ref{app:sc_metro} and is based on the definition of Fisher quantum information of quantum states in terms of the generating operator $\hat H$ of the parameter $r$ \cite[Eqs.~60,61]{Toth_2014} (here, $\hat H = \hat q/\sqrt 2 + \hat p/\sqrt 2$).

Therefore, for states with fixed energy $E$, increasing the symplectic coherence $V_{xp}^\rho$ of the input probe state always increases the quantum Fisher information and as a consequence, the ultimate precision of estimation. Note that $V_{xp}^\rho = \sqrt{\mathfrak c_\rho}$ holds only as long as $V_{xp}^\rho \geq 0$. This is always possible without loss of generality, since given a single-mode quantum state $\rho$ with covariance matrix
\begin{equation}
    V^\rho = \begin{bmatrix}
        V_x^\rho && V_{xp}^\rho \\
        V_{xp}^\rho && V_p^\rho
    \end{bmatrix},
\end{equation}
with $V_{xp}^\rho < 0$, we can always apply the Fourier gate $\hat F = e^{i\frac{\pi}{2}\hat a^\dagger \hat a}$ to $\rho$ to obtain
\begin{equation}
    V^{\hat F \rho \hat F^\dagger} = \begin{bmatrix}
        V_p^\rho && -V_{xp}^\rho \\
        -V_{xp}^\rho && V_x^\rho
    \end{bmatrix},
\end{equation}
whose off-diagonal elements are now positive, without changing the symplectic coherence.

Note that this discussion also reveals that Fourier gate is another example of a gate which mixes position and momentum quadratures and thus has gate-based symplectic coherence \cite{upreti2025interplay}, but does not change the symplectic coherence of a quantum state.

\subsection{Symplectic coherence and quantum channel discrimination}
The task of quantum channel discrimination—distinguishing between a given set of quantum channels—is a fundamental problem in quantum information theory \cite{Childs2000,Weedbrook2012}, with wide-ranging applications such as quantum reading \cite{Pirandola2011} and quantum state detection via quantum illumination \cite{Lloyd2008,Tan2008}.

To illustrate the use of symplectic coherence in quantum channel discrimination, we consider a scenario where the aim is to find out which of two orthogonal Stinespring dilations (Definition \ref{defi:ortho_Stinespring}) is being implemented on an input probe Gaussian state with zero first moments by making appropriate measurements at the output (Figure \ref{fig_illustration_metro_cd}). We are provided either single (single-shot discrimination) or multiple (multi-shot discrimination) copies of the probe state. 

In this section, we focus on zero-mean Gaussian probe states to isolate the contribution of symplectic coherence to channel distinguishability. Non-zero displacements add linear terms but do not affect symplectic coherence, as displacement operations are free. Including them is simple but would obscure the clear link to the discrimination bounds, so we restrict to the mean-zero case for clarity.

\subsubsection{Single-shot discrimination} 

In this section, we take one of the two channels to be a pure photon loss channel with parameter $\eta$ and the other channel to be any orthogonal Stinespring dilation channel (Definition \ref{defi:ortho_Stinespring}). An important example of this channel discrimination task would be a scenario in which we want to discern whether quantum states are passing through a given physical medium with or without experiencing photon loss. For such a channel discrimination task, we have the following result:
\begin{lem}[Symplectic coherence lower bounds optimal channel discrimination success probability]\label{lem:sc_lb_td_new}
    The optimal success probability of successful channel discrimination between a multi-mode photon loss channel $\Lambda_\eta$ and an arbitrary orthogonal Stinespring dilation $\mathcal O_0$ for a given input probe state $\rho$ is lower bounded in terms of its symplectic coherence $\mathfrak c_\rho$:
\begin{equation}
        p_\mathrm{opt} \geq \frac12 + \frac{1}{400}\min\left(1,\frac{ \sqrt{\frac{(1-\eta)^2 (E-2m)^2}{2m} + 2\mathfrak c_\rho(1-\eta)^2}}{E + 1} \right),
    \end{equation}
where $E/4 = \Tr[\rho \hat E]$.
\end{lem}
\noindent The proof of Lemma \ref{lem:sc_lb_td_new} is given in \ref{app:sc_lb_td}, and combines the Helstrom bound \cite{Helstrom1969} with a technical result which lower bounds the trace distance between two quantum states $\rho$ and $\sigma$ in terms of the difference between their symplectic coherence:
\begin{lem}[Symplectic coherence lower bounds trace distance between two states]\label{lem:sc_lb_td}
   Let $\rho$ and $\sigma$ be $m$-mode states with zero first moments, covariance matrices $V^\rho, V^\sigma$ with $\Tr[V^\rho] = 4 \Tr[\rho \hat E] = E_1$, $\Tr[V^\sigma] = 4\Tr[\sigma \hat E] = E_2$ for positive constants $E_1$ and $E_2$, and such that $\Tr[\rho \hat E], \Tr[\sigma \hat E] \leq E/4$. If $\rho$ and $\sigma$ are Gaussian states, then we have
    \begin{equation}
       \frac12 ||\rho - \sigma||_1 \geq \!\frac{1}{200} \min \!\!\left(\!1,\frac{\sqrt{\frac{(E_1 - E_2)^2}{2m} + 2(\sqrt{\mathfrak c_\rho} - \sqrt{\mathfrak c_\sigma})^2}}{E+1}\right)\!.
    \end{equation}
    If $\rho$ and $\sigma$ are general (Gaussian or non-Gaussian) states, we have
    \begin{equation}
        \frac12 ||\rho-\sigma||_1 \geq \frac{\frac{(E_1 - E_2)^2}{2m} + 2(\sqrt{\mathfrak c_\rho} - \sqrt{\mathfrak c_\sigma})^2}{3200\tilde E^2m},
    \end{equation}
    where $\tilde E^2$ is such that $\Tr[\hat{E^2}\rho] = \Tr[\hat{E^2}\sigma] \leq \tilde E^2$.
\end{lem}
The proof of Lemma \ref{lem:sc_lb_td} is given in appendix \ref{app:sc_lb_td}, and combines results from \cite[Theorem 11]{mele2024learning} and \cite[Theorem 62]{annamele2025symplecticrank} with a bound on the Hilbert--Schmidt distance between two covariance matrices in terms of their trace and their symplectic coherence. 

This result shows that higher symplectic coherence guarantees higher single-shot discrimination success between a photon loss channel (parameterized by $\eta$) and any block-diagonal orthogonal Stinespring dilation.

\subsubsection{Multi-shot discrimination} 

In this section, we consider a different scenario in which we want to discriminate between two block-diagonal orthogonal Stinespring dilations $\mathcal O_1$ and $\mathcal O_2$, using measurement of position-momentum correlations in the first mode, given by the observable $\hat M = \{\hat q_1, \hat p_1\}/2$, and using $N$ copies of an input probe state $\rho$. The discrimination protocol that we use is detailed in appendix \ref{app:eff_ortho_discrimination} and gives the following result:

\begin{lem}[Symplectic coherence underlines the efficiency of multi-shot discrimination protocols]\label{lem:eff_ortho_discrimination}
    Given a Gaussian input probe state $\rho$ with $\Tr[V^\rho] = E$ and symplectic coherence $\mathfrak c_\rho$, two orthogonal Stinespring dilation channels $\mathcal O_1$ and $\mathcal O_2$ satisfying 
$\Tr[\mathcal O_i(\rho) \{\hat q_1, \hat p_1\}/2] = \mu_i$ for $i=1,2$ can be discriminated with success probability $1 - \delta$ by measuring $\hat M = \{\hat q_1, \hat p_1\}/2$ when the number of samples satisfies $N \geq N_{\mathrm{thres}}$, where
\begin{equation}
    N_{\mathrm{thres}} \leq 
    272 \log\!\left(\frac2\delta\right)\frac{
\max (\mu_1^2,\mu_2^2)+f(m,E)}{(\mu_2 - \mu_1)^2},
\end{equation}
where $f(m,E)\coloneqq1 + (E/2-(m-1))^2$. Moreover, since for given energy $E$ and symplectic coherence $\mathfrak c_\rho$, $|\mu_1|,\, |\mu_2| \leq \sqrt{\mathfrak c_\rho}$, the optimal upper bound (its lowest possible value) is given by
\begin{equation}
68 \log\!\left(\frac2\delta\right) 
    \left(1+\frac{f(m,E)}{\mathfrak c_\rho}\right).
\end{equation}
\end{lem}
\noindent The proof of Lemma \ref{lem:eff_ortho_discrimination} is given in the \cref{app:eff_ortho_discrimination}, and involves combining median of means estimator \cite{JERRUM1986} with the Wick's theorem \cite{Wick1950} for zero mean Gaussians. Intuitively, the result follows from the fact that a larger difference in expectation values makes the measurement outcomes more distinguishable, thus simplifying channel discrimination. However, it does not imply that simply increasing the probe state's symplectic coherence always reduces the required number of samples—the optimal probe depends on the specific channels. Lemma~\ref{lem:eff_ortho_discrimination} asserts that, for a given probe state $\rho$, its symplectic coherence bounds the number of samples of $\rho$ sufficient to distinguish between the two channels.

To give a concrete example of a scenario where the symplectic coherence of the input state is a resource, we further consider a channel discrimination scenario where the two orthogonal Stinespring dilation channels are multi-mode photon loss channels with loss parameters $\eta_1$ and $\eta_2$ respectively. Then we have the following result:
\begin{lem}[Symplectic coherence for loss channels discrimination]\label{lem:eff_loss_discrimination}
Given a Gaussian input probe state $\rho$ with $\Tr[V^\rho] = E$,  $\Tr[\rho \{\hat q_1, \hat p_1\}/2] = \mu$, and with $\nu$ the symplectic eigenvalue and $E_1$ the trace of the reduced covariance matrix of its first mode, two photon loss channels with transmissivities $\eta_1$ and $\eta_2$ can be distinguised with success probability $1-\delta$ by measuring $\hat M = \{\hat q_1, \hat p_1\}/2$, provided the number of samples of $\rho$ is greater than
\begin{equation}
    N_\mathrm{thres} \coloneqq272 \log\!\left(\frac2\delta\right)\frac{ 
    \max\big( g(\mu,\eta_1,E_1), g(\mu,\eta_2,E_2) \big)}{(\eta_2 - \eta_1)^2},
\end{equation}
where $g(\mu,\eta,E)$ is a function in $\mu,\eta$ and $E$ whose exact expression is given in appendix \ref{app:eff_loss_discrimination}. For general $\eta_1$ and $\eta_2$, this is lower bounded by
\begin{equation}
    N_\mathrm{thres} \ge 
    272 \log\!\left(\frac2\delta\right)\frac{ 
    \max\big( \tilde g(m,E,\eta_1),\, \tilde g(m,E,\eta_2) \big)}{(\eta_2 - \eta_1)^2},
\end{equation}
with
\begin{equation}
    \tilde g(m,E,\eta) \coloneqq 
    \frac{
        \eta^2 + (1-\eta)^2}{
        \mathfrak c_{\max}(m,E)} 
    + 2\eta^2,
\end{equation}
where $\mathfrak c_{\max}(m,E)$ is given by Eq.~\ref{eq:max_sc}. 
The lower bound is saturated by a state of maximal symplectic coherence
\begin{equation}
    \ket{MSC} = 
    \hat R(\pi/4) \ket{r} \otimes \ket{0}^{\otimes (m-1)},
\end{equation}
where $\ket{r}$ is a momentum-squeezed vacuum state ($r>0$), with $r$ satisfying Eq.~\ref{eq:focus_squeeze}.
\end{lem}
\noindent The proof of Lemma \ref{lem:eff_loss_discrimination} is given in appendix \ref{app:eff_loss_discrimination} and is similar to the proof of Lemma \ref{lem:eff_ortho_discrimination} but the knowledge that the channels we are trying to discriminate are pure photon loss channels allows us to find the structure of the state which gives a tighter number of sufficient samples, which is a state of maximal symplectic coherence $\ket{MSC}$.

As a final illustration of the use of position-momentum correlations for assessing channel discrimination tasks, we consider a single-mode Gaussian channel discrimination scenario in which the output measurement is the rotated quadrature measurement $\hat q_\theta = \hat q \cos \theta - \hat p \sin \theta$, and the two Gaussian channels are $\textsf{PP-MM}$-equivalent, where $\textsf{PP-MM}$ equivalence is defined as follows:
\begin{defi}[$\textsf{PP-MM}$ equivalence of Gaussian channels]\label{defi:PP-MM_equivalence}
    Two $m$-mode Gaussian channels $\mathcal{G}_1$ and $\mathcal{G}_2$ are said to be $\textsf{PP-MM}$ equivalent if given an arbitrary $m$-mode quantum state $\rho$ with covariance matrix $V^\rho$, they change the diagonal block matrices of $V^\rho$ the same (the off-block diagonal matrices may change differently).
\end{defi} 
\noindent In other words, \textsf{PP-MM} equivalent channels affect the position-position and momentum-momentum correlations in the covariance matrix of a quantum state in the same way, but may change the position-momentum correlations differently. Let us illustrate this definition with an example. The action of any Gaussian channel on a covariance matrix $V$ is parameterized as \cite[Eq.~8]{eisert2005}
\begin{equation}
    \mathcal G(\sigma) = X^T \sigma X + Y,
\end{equation}
where $X$ and $Y$ are real $2m \times 2m$ matrices which satisfy the condition
\begin{equation}
    Y + i \Omega - i X^T \Omega X \succeq 0
\end{equation}
Therefore, defining Gaussian channels $\mathcal G_1$ and $\mathcal G_2$  such that $X = \I$ (as it is the case of classical noise \cite[Section 2.3]{eisert2005}) and $Y_1$ and $Y_2$ are positive semi-definite such that 
\begin{equation}
    Y_2 = Y_1 + \begin{bmatrix}
        0 && Y_\mathrm{rem} \\
        Y_\mathrm{rem}^T && 0
    \end{bmatrix},
\end{equation}
then $\mathcal G_1$ and $\mathcal G_2$ are $\textsf{PP-MM}$ equivalent. 

For such channels, we obtain the following result:
\begin{lem}[Position-momentum correlations determines distinguishability for $\textsf{PP-MM}$ equivalent channels]\label{lem:rotated_quadrature_sc}
     Given an input single-mode probe state $\rho$ with covariance matrix $V^\rho$ and two $\textsf{PP-MM}$ equivalent Gaussian channels $\mathcal G_1$ and $\mathcal G_2$ (Definition \ref{defi:PP-MM_equivalence}), the total variation distance between probability distributions generated by rotated quadrature measurement $\hat q_\theta$
    \begin{equation}
        \hat q_\theta = \hat q \cos \theta + \hat p \sin \theta
    \end{equation}
    is upper bounded by
    \begin{equation}
        \min(1,h_\theta(\mathcal G_1(\rho)),h_\theta(\mathcal G_2(\rho))), \nonumber\\
    \end{equation}
    where
    \begin{equation}
        h_{\theta}(\mathcal G(\rho)) \coloneqq \frac{|\sin(2\theta)||\sigma_{xp}^{\mathcal{G}_1(\rho)} - \sigma_{xp}^{\mathcal{G}_2(\rho)}|}{|\sigma_x \cos^2 \theta + \sigma_p \sin^2 \theta + \sigma_{xp}^{\mathcal G(\rho)} \sin 2\theta|},
    \end{equation}
    and $\sigma_{xp}^{\mathcal{G}_1(\rho)}$ is off-diagonal element of the covariance matrix of $\mathcal G_1(\rho)$, and similarly for $\sigma_{xp}^{\mathcal{G}_2(\rho)}$.
\end{lem}
\noindent The proof of Lemma \ref{lem:rotated_quadrature_sc} is given in appendix \ref{app:rotated_quadrature_sc}, and consists in two steps. Firstly, we parameterize the probability distribution generated by the measurement $\hat q_\theta$ on $\rho$. Combining this with the upper bound on total variational distance between two zero-mean Gaussians \cite[Theorem 1.1]{devroye2023}, we get the desired result.

Theorem \ref{lem:rotated_quadrature_sc} implies that the the distinguishability of two $\textsf{PP-MM}$ equivalent channels is completely determined by the difference in their position-momentum correlations (for $\sigma_{xp}^{\mathcal{G}_1(\rho)} = \sigma_{xp}^{\mathcal{G}_2(\rho)} $, the total variation distance is zero), and since $|\sigma_{xp}^{\mathcal{G}_1(\rho)}| = \sqrt{\mathfrak c_{\mathcal{G}_1(\rho)}}$ and similarly for $\mathcal G_2 (\rho)$, the position-momentum correlations of the two states are captured by their symplectic coherence.


\section{Conclusion}\label{sec:conclusion}

In this work, we have provided a rigorous framework for assessing the amount of position-momentum correlations in quantum states. To do this, we have introduced a geometric measure of position-momentum correlations, which we call \textit{symplectic coherence}, that can be defined informally as a quantifier of how much ``quantumness'' of the state is encoded in its position-momentum correlations, in the sense of geometric quantum discord.

As a measure of position-momentum correlations, symplectic coherence is additive under tensor product, faithful, monotone under a large class of quantum operations, and robust under small perturbations. The search for maximal symplectic coherence reveals an important fact: to obtain states with the highest position-momentum correlations (as measured by symplectic coherence) at a fixed energy, one must start with all the energy concentrated in a single mode. Only then can the appropriate passive linear unitary transformations produce the states with maximal position-momentum correlations. To our knowledge, no prior results indicate that such a state should have this particular structure. Finally, we demonstrate the role of symplectic coherence in various quantum thermodynamic phenomena and quantum information tasks, highlighting the practical significance of our framework in addition to its conceptual relevance.

In deriving the different results of the paper, we also obtained several technical results about covariance matrices of quantum states and matrix norms, which may be of independent interest. In summary, this work establishes symplectic coherence as a rigorous and operationally meaningful quantifier of position-momentum correlations, highlighting their fundamental importance in quantum physics and establishing the framework for their use as a quantum resource in future applications.

Our work paves the way for several future research directions. An immediate next step is to investigate additional quantum information tasks where symplectic coherence serves as a resource quantifier. Another promising avenue is to explore whether mapping covariance matrices onto virtual quantum-state density matrices can yield further physical insights, just as it demonstrated the link between symplectic coherence and quantum discord. Additionally, developing a physical understanding of why concentrating all energy into a single mode maximizes position–momentum correlations of quantum states is an intriguing problem. Finally, extending the current framework to include higher‑order position–momentum correlations beyond covariance matrices--essential for non‑Gaussian states--is a compelling research direction.
\section{Acknowledgements}
We acknowledge inspiring discussions with H.\ Ollivier.
We acknowledge funding from the European Union’s Horizon Europe Framework Programme (EIC Pathfinder Challenge project Veriqub) under Grant Agreement No.~101114899.
\bibliography{apssamp.bib}
\onecolumngrid
\newpage
\appendix
\tableofcontents

\newpage

\section{Quantum discord for computational basis measurements}\label{app:rdiscord}
In this paper, we consider quantum discord for bipartite quantum states $\varrho_{AB}$, where:
\begin{itemize}
    \item the subsystem $A$ is a qubit;
    \item instead of minimizing over all possible measurements on $A$, we restrict the measurement of the qubit register to be in the computational basis.
\end{itemize}
The expression for quantum discord in this setting is given by (see \cref{sec:discord})
\begin{equation}
    D(\varrho_{AB})_{\{\ket 0 \bra 0, \ket 1,\bra 1\}} = \mathcal I(B:A) - \mathcal J(B:A)_{\{\ket 0 \bra 0, \ket 1,\bra 1\}}.
\end{equation}
For this definition of quantum discord, the set of classical-quantum states (states with zero quantum discord \cite{Piani2008}), denoted by $\mathcal{CQ}$, is given by \cite[Proposition 3]{Ollivier2001}
\begin{equation}
    \varsigma \in \mathcal{CQ} \iff \varsigma = \ket 0 \bra 0 \otimes \varrho_0 + \ket 1 \bra{1} \otimes \varrho_1,
\end{equation}
where $\varrho_0$ and $\varrho_1$ are positive semi-definite matrices in subsystem $B$ (not necessarily trace one).

For this set of classical-quantum states, we define the geometric quantum discord using the Hilbert--Schmidt distance:
\begin{equation}\label{eq:gqd_restricted_supp}
    D_G^{HS} (\varrho_{AB}) = \min_{\varsigma \in \mathcal {CQ}} \Tr[(\rho_{AB} - \varsigma)^2].
\end{equation}
Then, we obtain the following result:
\begin{lem}\label{applem:closestCQstate}
    Given a bipartite state with a qubit subsystem $A$, and a density matrix given by
    \begin{equation}
        \varrho_{AB} = \begin{bmatrix}
            \varrho_0 && \varrho_{01} \\
            \varrho_{01}^\dagger && \varrho_{1}
        \end{bmatrix},
    \end{equation}
    where $\varrho_{0}, \varrho_{01},\varrho_{01}^\dagger,\varrho_1$ are matrices associated to $\ket 0 \bra 0 \otimes \ket{\bm i}\bra{\bm j}, \ket 0 \bra 1 \otimes \ket{\bm i}\bra{\bm j}, \ket 1 \bra 0 \otimes \ket{\bm i}\bra{\bm j}$ and $\ket 1 \bra 1 \otimes \ket{\bm i}\bra{\bm j}$ respectively, where $\ket{\bm i}$ represents the orthonormal basis of subsystem $B$, its geometric quantum discord, defined in Eq.~\ref{eq:gqd_restricted_supp}, is given by
    \begin{equation}\label{appeq:gqd_exp}
        D_G^{HS} (\varrho_{AB}) = 2||\varrho_{01}||_F^2.
    \end{equation}
\end{lem}

\begin{proof}
    First, we prove that the state $\varsigma \in \mathcal{CQ}$ closest in Hilbert--Schmidt distance to $\varrho_{AB}$ is given by
    \begin{equation}
        \varsigma = \begin{bmatrix}
            \varrho_0 && 0 \\
            0 && \varrho_1
        \end{bmatrix}.\end{equation}
        We prove this by contradiction. Suppose the state $\tilde \varsigma \in \mathcal{CQ}$ closest in Hilbert--Schmidt distance to $\varrho_{AB}$ is given by
        \begin{equation}
            \tilde \varsigma =\begin{bmatrix}
                \tilde \varsigma_0 && 0 \\
                0 &&  \tilde \varsigma_1
            \end{bmatrix}.
        \end{equation}
        This implies
    \begin{equation}\label{eq:gd_restr_contradiction}
        \Tr[(\varrho_{AB} - \varsigma)^2] - \Tr[(\varrho_{AB}- \tilde\varsigma)^2] \geq 0.
    \end{equation}
    Now,
    \begin{equation}
       \Tr[\varrho_{AB} - \varsigma)^2] = \Tr[\varrho_{AB}^2] + \Tr[\varsigma^2] - 2\Tr[\varrho_{AB}\varsigma] = \Tr[\varrho_{AB}^2] + \Tr[\varsigma^2] - 2\Tr[\varsigma^2] = \Tr[\varrho_{AB}^2] - \Tr[\varsigma^2].
    \end{equation}
    Since
    \begin{equation}
        \Tr[\varrho_{AB}\varsigma] = \Tr[\varrho_0^2 + \varrho_1^2] = \Tr[\varsigma^2].
    \end{equation}
    Similarly,
    since
    \begin{equation}
        \Tr[\varrho_{AB}\tilde\varsigma] = \Tr[\varrho_0\tilde\varsigma_0 + \varrho_1\tilde \varsigma_1] = \Tr[\varsigma\tilde\varsigma].
    \end{equation}
    Therefore
    \begin{equation}
        \Tr[(\varrho_{AB} - \tilde \varsigma)^2] = \Tr[\varrho_{AB}^2] + \Tr[\tilde \varsigma^2] - 2\Tr[\varsigma\tilde \varsigma],
    \end{equation}
    And finally
    \begin{equation}
        \Tr[\varrho_{AB} - \varsigma)^2] - \Tr[(\varrho_{AB}- \tilde\varsigma)^2] = -\Tr[\varsigma^2] - \Tr[\tilde \varsigma^2] + 2\Tr[\varsigma \tilde \varsigma] = - \Tr[(\varsigma - \tilde \varsigma)^2] = -||\varsigma - \tilde \varsigma||_F^2,
    \end{equation}
    so the Eq.~\ref{eq:gd_restr_contradiction} gives
    \begin{equation}
        ||\varsigma - \tilde \varsigma||_F^2 \leq 0,
    \end{equation}
    with an equality if $\tilde \varsigma = \varsigma$. Since the Frobenius norm cannot be negative, this shows that the covariance matrix closest in Hilbert--Schmidt distance to $\varrho_{AB}$ is $\varsigma$. This gives the geometric quantum discord as
    \begin{equation}
        D_G^{HS}(\varrho_{AB}) = \Tr[(\varrho_{AB} - \varsigma)^2] = \Tr[\varrho_{\mathrm{rem}}^2],
    \end{equation}
    where
    \begin{equation}
        \varrho_{\mathrm{rem}} \coloneqq \begin{bmatrix}
            0 && \varrho_{01} \\
            \varrho_{01}^\dagger && 0
        \end{bmatrix}.
    \end{equation}
    Therefore,
    \begin{equation}
        D_G^{HS}(\varrho_{AB}) = \Tr[\varrho_{01}^\dagger \varrho_{01}  + \varrho_{01} \varrho_{01}^\dagger] = 2 \Tr[\varrho_{01}^\dagger \varrho_{01}] = 2||\varrho_{01}||_F^2.
    \end{equation}
\end{proof}
\noindent Our next result is an upper bound on geometric quantum discord:
\begin{lem}
    Given a bipartite state $\varrho_{AB}$ with the subsystem $A$ being a qubit and the density matrix given by
    \begin{equation}
        \varrho_{AB} = \begin{bmatrix}
            \varrho_0 && \varrho_{01} \\
            \varrho_{01}^\dagger && \varrho_{1}
        \end{bmatrix},
    \end{equation}
    its maximal geometric quantum discord is given by
    \begin{equation}
        D_G^{HS}(\varrho_{AB}) = 2 ||\varrho_{01}||_F^2 \leq  \frac12.
    \end{equation}
\end{lem}
\begin{proof}
    Since $\varrho_{AB}$ is a density matrix,
    \begin{equation}
        \varrho_{AB} \succeq 0.
    \end{equation}
    Therefore, using the property of semi-definite matrices, $\forall i,j \in \{1,\dots,m\}$, the principal submatrix formed by eliminating all rows and columns except the $i$th and the $m+j$th one satisfies
    \begin{equation}
        \begin{bmatrix}
            (\varrho_0)_{ii} && (\varrho_{01})_{ij} \\
            (\varrho_{01})_{ij}^* && (\varrho_{11})_{jj}
        \end{bmatrix} \succeq 0.
    \end{equation}
    This implies
    \begin{eqnarray}
        |\varrho_{01}|_{ij}^2 &\leq& (\varrho_0)_{ii} (\varrho_{11})_{jj},
    \end{eqnarray}
    and therefore
    \begin{equation}
        2||\varrho_{01}||_F^2 = 2 \sum_{i,j=1}^m |\varrho_{01}|_{i,j}^2 \leq 2 \sum_{i,j=1}^m  (\varrho_{0})_{ii} (\varrho_{1})_{jj} = 2\left(\sum_{i=1}^m (\varrho_0)_{ii}\right)\left(\sum_{j=1}^m (\varrho_1)_{jj}\right).
    \end{equation}
    Denote
    \begin{equation}
        \sum_{i=1}^m (\varrho_0)_{ii} = a.
    \end{equation}
    Then since $\varrho_{AB}$ has trace 1, 
    \begin{equation}
        \sum_{j=1}^m (\varrho_1)_{jj} = 1-a.
    \end{equation}
    This gives
    \begin{equation}
        \left(\sum_{i=1}^m (\varrho_0)_{ii}\right)\left(\sum_{j=1}^m (\varrho_1)_{jj}\right) = a(1-a).
    \end{equation}
    Now, the maximum value of $a(1-a)$ is equal to $1/4$ for $a=1/2$. Putting all of this together,
    \begin{equation}
        D_G^{HS}(\varrho_{AB}) =2||\varrho_{01}||_F^2 \leq 2 \left(\sum_{i=1}^m (\varrho_0)_{ii}\right)\left(\sum_{j=1}^m (\varrho_1)_{jj}\right) = 2a(1-a) \leq \frac12 .
    \end{equation}
\end{proof}

\section{Properties of symplectic coherence and position-momentum correlations}
\noindent To recap, the symplectic coherence of a quantum state $\rho$ is defined as:
\begin{defi} [Symplectic coherence] Given a quantum state $\rho$ with covariance matrix
\begin{equation}
    V^\rho = \begin{bmatrix}
        V_x^\rho && V_{xp}^\rho \\ (V_{xp}^\rho)^T && V_p^\rho
    \end{bmatrix},
\end{equation}
the symplectic coherence of $\rho$, denoted by $\mathfrak c_\rho$, is defined as:
    \begin{equation}
  \mathfrak c_\rho = || V_{xp}^\rho ||_F^2.
\end{equation}
\end{defi}
\noindent In this section, we prove various properties of symplectic coherence and position-momentum correlations, namely the operational interpretation of symplectic coherence (section \ref{app:sc_exp}), its monotonicity and faithfulness (section \ref{app:monotonicity_sc}), the class of relevant non-free operations for symplectic coherence (section \ref{app:sc_non_free}), a no-go result for monotinicity of faithful measures of position-momentum correlations (section \ref{app:active_increase_sc}), relation of position-momentum correlations to quantum discord (section \ref{app:zero_sc_zero_discord}), maximal symplectic coherence value under energy constraints and the structure of the state achieving this maximal value (section \ref{app:max_sc}), behaviour of symplectic coherence under perturbations of the state (section \ref{app:SC_perturbations}), how it changes under photon loss (section \ref{app:sc_under_loss}) and finally how position-momentum correlations enhance average entanglement of pure Gaussian states (section \ref{app:sc_typical_entangle}.

\subsection{Operational interpretation of symplectic coherence}\label{app:sc_exp}
\noindent We restate Theorem \ref{theo:sc_expr}:
\begin{theo}[Operational interpretation of symplectic coherence]\label{apptheo:sc_exp_operational}
    Given a quantum state $\rho$ with covariance matrix
    \begin{equation}
        V^\rho = \begin{bmatrix}
            V_x && V_{xp} \\
            V_{xp}^T && V_p
        \end{bmatrix},
    \end{equation}
    its symplectic coherence is equal to the minimal Hilbert--Schmidt distance from the set of states with the set of free states $\mathcal C$, namely
    \begin{equation}\label{appeq:sc_exp_operational}
       \mathfrak c_\rho = \min_{\sigma \in \mathcal C} \frac12 ||V^\rho - V^\sigma||_F^2 = \min_{\sigma \in \mathcal C} \frac12 \Tr[(V^\rho - V^\sigma)^2].
    \end{equation}
\end{theo}
\noindent Before proving the Theorem, we prove two technical lemmas which are required to complete the proof:
\begin{lem}
    Given the covariance matrix of a quantum state
    \begin{equation}
        V = \begin{bmatrix}
            V_x && V_{xp} \\
            V_{xp}^T && V_p
        \end{bmatrix},
    \end{equation}
    the matrix obtained by removing the position-momentum correlations
    \begin{equation}
        \tilde V = \begin{bmatrix}
            V_x && 0\\
            0 && V_p
        \end{bmatrix}
    \end{equation}
    is a valid covariance matrix of a quantum state.
\end{lem}
\begin{proof}
We prove the Lemma first for covariance matrices of pure Gaussian states, and then extend it to covariance matrices of mixed Gaussian states or non-Gaussian states.
    From \cite[Theorem 1]{albert1969}, a given hermitian matrix
    \begin{equation}
        \begin{bmatrix}
            A && B \\
            B^\dagger && D
        \end{bmatrix}
    \end{equation}
    with a positive definite matrix $A$ is positive semi-definite if and only if 
    \begin{eqnarray}\label{eq:detpos}
        D - B^\dagger A^{-1} B &\succeq& 0.
    \end{eqnarray}
    Since $V^\rho$ is the covariance matrix of a pure Gaussian state,
    \begin{eqnarray}
    V_x &\succ& 0, \nonumber \\
        V + i\Omega &\succeq& 0,
    \end{eqnarray}
    where
    \begin{equation}
        \Omega = \begin{bmatrix}
            \bm 0_m && \I_m \\ - \I_m && \bm 0_m
        \end{bmatrix}.
    \end{equation}
    Therefore, with Eq.~\ref{eq:detpos},
    \begin{eqnarray}
        V_p - (V_{xp}^T - i \I)V_x^{-1}(V_{xp} + i \I) &\succeq& 0, \nonumber \\
        V_p - V_x^{-1} - V_{xp}^T V_x^{-1} V_{xp} + i (V_x^{-1} V_{xp} - V_{xp}^T V_x^{-1}) &\succeq& 0.
    \end{eqnarray}
    Now, for pure Gaussian states \cite[Eq.~9]{Serafini2007entanglement},
    \begin{equation}
        V_{xp} V_x = V_x V_{xp}^T.
    \end{equation}
    Multiplying from the left and from the right $V_x^{-1}$ on both sides of the equality ($V_x$ is positive definite and hence invertible), we get
    \begin{equation}
        V_x^{-1} V_{xp} = V_{xp}^T V_x^{-1}.
    \end{equation}
    Also, $\forall z \in \C^m$, we have
    \begin{equation}
        z^\dagger V_{xp}^T V_x^{-1} V_{xp} z = (V_{xp} z)^\dagger V_x^{-1} V_{xp} z > 0
    \end{equation}
    since $V_x^{-1}$ is positive definite. Combining all of this with the previous results gives
    \begin{equation}
        V_p - V_x^{-1} \succeq V_{xp}^T V_x^{-1} V_{xp} \succ 0.
    \end{equation}
    Therefore, for
    \begin{equation}
        \tilde V = \begin{bmatrix}
            V_x && 0\\
            0 && V_p
        \end{bmatrix},
    \end{equation}
    we have
    \begin{eqnarray}
    \tilde V &=& \tilde V^T, \nonumber \\
        V_x &\succ& 0, \nonumber \\
        V_p &\succ& 0, \nonumber \\
        V_p - V_x^{-1} &\succeq& 0. \nonumber \\
    \end{eqnarray}
    $V_x \succ 0$ and $V_p \succ 0$ implies $\tilde V \succ 0$, diagonalized by the orthogonal matrix
    \begin{equation}
        \begin{bmatrix} O_1 && 0 \\ 0 && O_2 \end{bmatrix},
    \end{equation}
    where $O_1$ and $O_2$ are orthogonal matrices diagonalizing $V_x$ and $V_p$ respectively. Further for
    \begin{equation}
        \tilde V + i\Omega = \begin{bmatrix}
            V_x && i \I_m \\ -\I_m && V_p
        \end{bmatrix}
    \end{equation}
    with positive definite $V_x$, we have that 
    \begin{equation}
        V_p - (-i\I_m)V_x^{-1}(i\I_m) = V_p - V_x^{-1} \succ 0.
    \end{equation}
    The last two results imply
    \begin{eqnarray}
        \tilde V &\succ& 0 \nonumber \\
        \tilde V + i \Omega &\succeq& 0.
    \end{eqnarray}
    So $\tilde V$ represents a valid covariance matrix of a quantum state when $V^\rho$ is the covariance matrix of a pure Gaussian state.

    \medskip

    Now suppose $V^\rho$ is the covariance matrix of a mixed Gaussian quantum state (or non-Gaussian quantum state), then we know that $\tilde V$ satisfies
    \begin{eqnarray}
        \tilde V &=& \tilde V^T \nonumber \\
        \tilde V &\succ& 0.
    \end{eqnarray}
    To prove that $\tilde V + i \Omega \succeq 0$, we note that the covariance matrix of any mixed Gaussian state can be written as the convex mixture of two covariance matrices of pure Gaussian states \cite[Theorem 3]{parthasarathy2011}:
    \begin{equation}
        V^\rho = \frac12 V^{G_1} + \frac12 V^{G_2}.
    \end{equation}
    This implies
    \begin{equation}
        \tilde V = \frac{1}{2} \tilde V^{G_1} + \frac12 \tilde V^{G_2},
    \end{equation}
    and hence
    \begin{equation}
        \tilde V + i \Omega = \frac12 (\tilde V^{G_1} + i \Omega) + \frac 12 (\tilde V^{G_2} + i \Omega) \succeq 0,
    \end{equation}
    since $\tilde V^{G_1} + i \Omega \succeq 0$  and $\tilde V^{G_2} + i \Omega \succeq 0$ from the previous proof in the case of pure Gaussian states.
\end{proof}
\noindent Note that since taking the global transpose of the state $\rho$ flips the sign of the momentum operator and leaves the position operator unchanged, $\tilde V$ can equivalenty be seen as the covariance matrix of $(\rho + \rho^T)/2$, proving its validity as a covariance matrix \footnote{The authors thank Ludovico Lami for pointing this out.}. One more mathematical result is needed before we prove Theorem \ref{theo:sc_expr}:

\begin{lem} \label{lem:close_cv}
    Given a covariance matrix
    \begin{equation}
        V = \begin{bmatrix}
            V_x && V_{xp} \\
            V_{xp}^T && V_p 
        \end{bmatrix}.
    \end{equation}
    The covariance matrix $V^\sigma$ with $\sigma \in \mathcal C$ which is closest (in Hilbert--Schmidt distance) to $V$ is given by
    \begin{equation}
        V^\sigma = \begin{bmatrix}
            V_x && 0 \\
            0 && V_p 
        \end{bmatrix}.
    \end{equation}
\end{lem}

\begin{proof}
    The proof is similar to that of \cref{applem:closestCQstate}. We prove the result by contradiction. Suppose there is a covariance matrix of a quantum state with zero position-momentum correlations, i.e.,
    \begin{equation}
        \tilde V = \begin{bmatrix}
            \tilde V_x && 0 \\
            0 && \tilde V_p
        \end{bmatrix},
    \end{equation}
    which is closer in Hilbert--Schmidt distance to $V$. This implies
    \begin{equation}\label{eq:lem_4_1}
        \Tr[(V - V^\sigma)^2] - \Tr[(V- \tilde V)^2] \geq 0.
    \end{equation}
    Now,
    \begin{equation}
        \Tr[(V - V^\sigma)^2] = \Tr[V^2] + \Tr[(V^\sigma)^2] - 2\Tr[VV^\sigma] = \Tr[V^2] + \Tr[(V^\sigma)^2] - 2\Tr[(V^\sigma)^2] = \Tr[V^2] - \Tr[(V^\sigma)^2].
    \end{equation}
    Since
    \begin{equation}
        \Tr[VV^\sigma] = \Tr[V_x^2 + V_p^2] = \Tr[(V^\sigma)^2].
    \end{equation}
    Similarly,
    since
    \begin{equation}
        \Tr[V\tilde V] = \Tr[V_x \tilde V_x + V_p \tilde V_p] = \Tr[V^\sigma \tilde V].
    \end{equation}
    Therefore
    \begin{equation}
        \Tr[(V - \tilde V)^2] = \Tr[V^2] + \Tr[\tilde V^2] - 2\Tr[V^\sigma\tilde V],
    \end{equation}
    And finally
    \begin{equation}
        \Tr[(V - V^\sigma)^2] - \Tr[(V- \tilde V)^2] = -\Tr[(V^\sigma)^2] - \Tr[(\tilde V)^2] + 2\Tr[V^\sigma \tilde V] = - \Tr[(V^\sigma - \tilde V)^2] = -||V^\sigma - \tilde V||_F^2,
    \end{equation}
    so the earlier inequality (Eq.~\ref{eq:lem_4_1}) gives
    \begin{equation}
        ||V^\sigma - \tilde V||_F^2 \leq 0
    \end{equation}
    With an equality if $\tilde V = V^\sigma$. Since a Frobenius norm is always non-negative and equal to zero only for the zero matrix, this proves that the covariance matrix closest in Hilbert--Schmidt distance to $V$ is $V^\sigma$.
\end{proof}
\noindent With these results, the proof of Theorem \ref{theo:sc_expr} becomes straightforward.

\begin{proof}[Proof of Theorem \ref{theo:sc_expr}]

\noindent From Lemma \ref{lem:close_cv},
\begin{equation}
    \min_{\sigma \in \mathcal C} \frac12 ||V^\rho - V^\sigma||_F^2 = \min_{\sigma \in \mathcal C} \frac12 \Tr[(V^\rho - V^\sigma)^2] =  \frac12 \Tr[V_\mathrm{rem}^2],
\end{equation}
where
\begin{equation}
    V_\mathrm{rem} = \begin{bmatrix}
        0 && V_{xp} \\
        V_{xp}^T && 0
    \end{bmatrix}.
\end{equation}
Since,
\begin{equation}
    V_\mathrm{rem}^2 = \begin{bmatrix}
        V_{xp} V_{xp}^T && 0 \\
        0  && V_{xp}^T V_{xp}
    \end{bmatrix}.
\end{equation}
This gives
\begin{equation}
    \min_{\sigma \in \mathcal C} \frac12 ||V^\rho - V^\sigma||_F^2 = \min_{\sigma \in \mathcal C} \frac12 \Tr[(V^\rho - V^\sigma)^2] = \frac12 \times 2 \Tr[V_{xp}V_{xp}^T] = ||V_{xp}||_F^2,
\end{equation}
which is the exactly the expression of symplectic coherence (Defintion \ref{defi:sc_expr}).
\end{proof}

\subsection{Faithfulness and monotonicity of symplectic coherence}\label{app:monotonicity_sc}
\noindent We restate Theorem \ref{theo:monotonicity_sc}:
  \begin{theo}[Faithfulness and monotonicity of symplectic coherence]
        Given a quantum state $\rho$, its symplectic coherence $\mathfrak c_\rho = 0$ if and only if $\rho \in \mathcal C$ (Definition \ref{defi:free_states}). Further, the symplectic coherence is non-increasing under the following operations:
        \begin{itemize}
            \item Block-diagonal orthogonal unitary gates $\hat O$.
            \item Displacement unitary gate.
            \item Tensor product with free states.
            \item Partial traces.
            \item Classical mixing of zero first moment states.
        \end{itemize}
    \end{theo}
\begin{proof}
\textbf{Faithfulness of symplectic coherence:} The faithfulness of symplectic coherence simply follows from the property of Frobenius norm that for a given $m \times m$ matrix $A$:
\begin{equation}\label{eq:faithful_frobenius}
    ||A||_F \iff A = \bm 0_m,
\end{equation}
where $\bm 0_m$ represents the $m \times m$ zero matrix. Therefore if $\mathfrak c_\rho =0$ for an $m$-mode state $\rho$, this implies that its covariance matrix $V^\rho$ is such that 
\begin{equation}
    ||V_{xp}^\rho||_F^2 = 0.
\end{equation}
Eq.~\ref{eq:faithful_frobenius} then implies $V_{xp}^\rho = \bm 0_m$ and hence $\rho \in \mathcal C$.

\medskip

\noindent \textbf{Block-diagonal orthogonal unitary gates:} Given the covariance matrix of the state $\rho$:
\begin{equation}
    V^\rho = \begin{bmatrix}
        V_x^\rho && V_{xp}^\rho \\
        (V_{xp}^\rho)^T && V_p^\rho
    \end{bmatrix},
\end{equation}
under the effect of a block-diagonal orthogonal unitary gate $\hat O$,
\begin{equation}
    V^\rho \mapsto \begin{bmatrix}O && 0 \\ 0 && O\end{bmatrix} \begin{bmatrix}
        V_x^\rho && V_{xp}^\rho \\
        (V_{xp}^\rho)^T && V_p^\rho
    \end{bmatrix} \begin{bmatrix}O^T&& 0 \\ 0 && O^T\end{bmatrix} = \begin{bmatrix}
        O V_x^\rho O^T && O V_{xp}^\rho O^T \\ O (V_{xp}^\rho)^T O^T && O V_p^\rho O^T
    \end{bmatrix}.
\end{equation}
Therefore, 
\begin{equation}
    \mathfrak c_{\hat O \rho \hat O^\dagger} = ||O V_{xp}^\rho O^T||_F^2 = ||V_{xp}^\rho||_F^2 = \mathfrak c_\rho,
\end{equation}
since the Frobenius norm is invariant under orthogonal matrices.

We note that the symplectic matrix of $\hat O$ maps to a local unitary through the mapping $\mathcal M$ (Definition \ref{defi:Covariance_to_density}). Indeed, the symplectic matrix of $\hat O$ is given by
\begin{equation}\label{app_eq:block_diagonal_orthogonal}
    S_O = \begin{bmatrix}
        O && 0 \\ 
        0 && O
    \end{bmatrix},
\end{equation}
and it acts on the covariance matrix $V^\rho$ as
\begin{equation}
    V^\rho\mapsto S_O V^\rho S_O^T.
\end{equation}
Through the mapping $\mathcal M$, the action of $\hat O$ can be seen as the action of a unitary 
\begin{equation}
    \hat U_{\mathrm{loc}} = \begin{bmatrix}
        O && 0 \\ 
        0 && O
    \end{bmatrix}
\end{equation}
on the virtual state $\rho = \frac{1}{\Tr[V^\rho]}V^\rho$, since 
\begin{equation}
    \hat U_{\mathrm{loc}} = \I_2 \otimes \hat O,
\end{equation}
where $\I_2$ is the identity operation on the Hilbert space of the qubit, whereas $\hat O$ is a local unitary gate on the qudit. Local unitary gates do not increase discord, similar to how block-diagonal orthogonal unitary gates do not change symplectic coherence. This further strengthens the conceptual bridge between symplectic coherence and quantum discord discussed in the main text.

\medskip

\noindent \textbf{Displacement gates:} The symplectic matrix corresponding to the displacement gate is the identity  $\I$. Therefore,
\begin{equation}
    V^{\hat D \rho \hat D^\dagger} = V^\rho,
\end{equation}
and hence
\begin{equation}
    \mathfrak c_{\hat D \rho \hat D^\dagger} = \mathfrak c_\rho.
\end{equation}
\textbf{Tensor product with free states:} Given an $m$-mode state $\rho$ and an $n$-mode free state $\sigma_B$,
\begin{equation}
    V_{xp}^{\rho \otimes \sigma_B} = \begin{bmatrix}
        V_{xp}^\rho && 0 \\ 0 && 0
    \end{bmatrix},
\end{equation}
and therefore,
\begin{equation}
    \mathfrak c_{\rho \otimes \sigma_B} = \sum_{i,j = 1}^{m+n} (V_{xp}^{\rho \otimes \sigma _B})_{ij}^2 = \sum_{i,j=1}^m (V_{xp}^\rho)_{ij}^2 = \mathfrak c_\rho.
\end{equation}

\medskip

\noindent \textbf{Partial traces:} Given a bipartite state $\rho_{AB}$, where $A$ is an $m$-mode system and $B$ is an $n$-mode system,
\begin{equation}
    V^{\rho_{AB}} = \begin{bmatrix}
        V_{xp}^{A} && V_{xp}^{AB} \\ V_{xp}^{BA} && V_{xp}^{B}
    \end{bmatrix},
\end{equation}
we have that
\begin{equation}
    V_{xp}^{\Tr_B[\rho_{AB}]} = V_{xp}^A.
\end{equation}
Therefore,
\begin{equation}
    \mathfrak c (\rho_{AB})= \sum_{i,j=1}^m ( V_{xp}^A)_{ij}^2 + \sum_{i=1}^m \sum_{j=1}^n (V_{xp}^{AB})_{ij}^2 + \sum_{i=1}^n \sum_{j=1}^m (V_{xp}^{BA})_{ij}^2 + \sum_{i,j=1}^{n} ( V_{xp}^B)_{ij}^2 \geq \sum_{i,j=1}^{m} ( V_{xp}^A)_{ij}^2 = \mathfrak{c}_{\Tr_B[\rho_{AB}]}.
\end{equation}

\medskip

\noindent \textbf{Classical mixing for zero first moments:} If we are only mixing states with zero first moments, i.e.
\begin{equation}
    \rho = \sum_{i=1}^l p_i \rho_i,
\end{equation}
where $\rho_i$ has zero first moments of quadratures. Then
\begin{equation}
    V^\rho = \sum_{i=1}^l p_i V^{\rho_i},
\end{equation}
and hence
\begin{equation}
    ||V_{xp}^\rho||_F \leq \sum_{i=1}^l p_i ||V_{xp}^{\rho_i}||_F \leq \max_i ||V^{\rho_i}||_F,
\end{equation}
and therefore,
\begin{equation}
    \mathfrak c_\rho = ||V_{xp}^\rho||_F^2 \leq \max_i \mathfrak c_{\rho_i}.
\end{equation}
\end{proof}


\subsection{Some important non-free operations for symplectic coherence}\label{app:sc_non_free}

\noindent In this section, we highlight some experimentally relevant operations which are not free.

\medskip 

\noindent \textbf{Heterodyne measurement:} Post-selected heterodyne measurements are not free for symplectic coherence. To prove this, we give an example where symplectic coherenct increases after measurement. Consider the state
\begin{equation}
    \rho = \frac12 \ket 0 \bra 0 \otimes \ket{ \xi}\bra{\xi} + \frac12 \ket{-r}\bra{-r} \otimes \ket {0} \bra{0},
\end{equation}
where $\ket{ \xi}\bra{\xi}$ is a rotated squeezed vacuum state with $\xi = r e^{i\theta}$ with $r > 0$, whereas $\ket{-r}$ is a single-mode squeezed state, squeezed along the momentum direction. The off-block-diagonal matrix of its covariance matrix is given by
\begin{equation}
    V_{xp}^\rho = \begin{bmatrix}
        0 && 0 \\ 0 && \frac12 \sinh (2r) \sin (2\theta)
    \end{bmatrix},
\end{equation}
and its symplectic coherence is given by
\begin{equation}
    \mathfrak c_\rho = \frac14 \sinh^2(2r) \sin^2 (2\theta).
\end{equation}
Now, if we post select on the vacuum state $\ket{0}\bra{0}$ on the first mode, corresponding to a $0$ outcome for heterodyne measurement, the post measurement state is 
\begin{equation}
    \sigma = \frac{\bra{0}\rho\ket{0}}{\Tr[\bra{0}\rho\ket{0}]}  =\frac{\cosh(r)}{1+\cosh(r)} \ket{\xi}\bra{\xi} + \frac{1}{1+\cosh(r)} \ket 0 \bra 0.
\end{equation}
This follows from the fact that $|\langle0|-r\rangle|^2 = 1/\cosh(r).$ Its off-diagonal covariance matrix element is
\begin{equation}
    V_{xp}^\sigma = 
        \frac{\cosh(r)}{1+\cosh(r)} \sinh (2r) \sin (2\theta),
\end{equation}
and therefore its symplectic coherence is given by
\begin{equation}
    \mathfrak{c}_\sigma =\frac{\cosh^2(r)}{(1+\cosh(r))^2} \sinh^2(2r) \sin^2 (2\theta).
\end{equation}
Since $\frac{\cosh^2(r)}{(1+\cosh(r))^2} > \frac14$ for all $r \neq 0 $. Therefore, for all $\theta \neq n \times \pi/2$, where $n \in \N$, $\mathfrak c$ increases after the measurement.

\medskip

\noindent \textbf{Homodyne measurements:} Post-selected homodyne measurements are not free for symplectic coherence. Consider again the state
\begin{equation}
     \rho = \frac12 \ket 0 \bra 0 \otimes \ket{ \xi}\bra{\xi} + \frac12 \ket{-r}\bra{-r} \otimes \ket {0} \bra{0},
\end{equation}
with symplectic coherence is given by
\begin{equation}
    \mathfrak c_\rho = \frac{1}4 \sinh^2(2r) \sin^2 (2\theta).
\end{equation}
Denoting position eigenstates as $\ket{q}_x$, we have
\begin{equation}
    |\langle-r|q\rangle_x|^2 = \frac{e^{-r}}{\sqrt{2\pi}}\exp\left(-\frac{q^2}{2} e^{-2r}\right).
\end{equation}
Therefore, if we post select on the homodyne measurement $\ket {0}\bra{0}_x$, the post-selected state is
\begin{equation}
    \sigma = \frac{\bra{0}\rho\ket{0}_x}{\Tr[\bra{0}\rho\ket{0}_x]}  = \frac1{1+e^{-r}}\ket{\xi}\bra{\xi} + \frac{e^{-r}}{1+e^{-r}} \ket{0}\bra{0}.
\end{equation}
Therefore its off-diagonal covariance matrix element is given by
\begin{equation}
    V_{xp}^\sigma = \frac{1}{1+e^{-r}} \sinh(2r)\sin(2\theta),
\end{equation}
and its symplectic coherence is given by
\begin{equation}
    \mathfrak c_\sigma = \frac{1}{(1+e^{-r})^2} \sinh^2(2r)\sin^2(2\theta).
\end{equation}
Since $\frac{1}{(1+e^{-r})^2} > \frac14$ for all $r > 0$. Therefore, for all $\theta \neq 2\pi$, symplectic coherence increases after the measurement.

\medskip

\noindent\textbf{Classical mixing for non-zero first moments:} To see this, consider the classical mixing
\begin{equation}
    \rho = \frac12 \rho_1 + \frac12 \rho_2.
\end{equation}
For such $\rho$,
\begin{eqnarray}
    V_{x_i,p_j}^\rho &=& \frac12\Tr[\{\hat x_i \hat p_j\} \rho] - \Tr[\hat x_i \rho] \Tr[\hat p_j \rho] \nonumber \\
    &=& \frac12 \left(\frac{\Tr[\{\hat x_i \hat p_j\} \rho_1]}{2} - \Tr[\hat x_i \rho_1] \Tr[\hat p_j \rho_1]\right) + \left(\frac{\Tr[\{\hat x_i \hat p_j\} \rho_2]}{2} - \Tr[\hat x_i \rho_2] \Tr[\hat p_j \rho_2]\right) \nonumber \\
    & & + \frac{\Tr[\hat x_i \rho_1] \Tr[\hat p_j \rho_1]}{4} + \frac{\Tr[\hat x_i \rho_2] \Tr[\hat p_j \rho_2]}{4}  - \frac{\Tr[\hat x_i \rho_1] \Tr[\hat p_j \rho_2]}{4} - \frac{\Tr[\hat x_i \rho_2] \Tr[\hat p_j \rho_1]}{4} \nonumber\\
    &=& \frac12 V_{x_i,p_j}^{\rho_1} + \frac12 V_{x_i,p_j}^{\rho_2} + \frac14 (\langle x_i\rangle_1 - \langle x_i \rangle_2)(\langle p_j\rangle_1 - \langle p_j \rangle_2).
\end{eqnarray}
Therefore,
\begin{equation}
    V_{xp}^\rho = \frac12 V_{xp}^{\rho_1} + \frac12 V_{xp}^{\rho_2} + \frac14 ((\bm d_x)_1 - (\bm d_x)_2) ((\bm d_p)_1 - (\bm d_p)_2)^T,
\end{equation}
where $\bm d_x^1$ ($\bm d_p^1$) is the first moment vector of position (momentum) quadratures in $\rho_1$, and similarly for $\rho_2$. Therefore, the covariance matrix is not simply a convex combination of covariance matrices of the individual states and therefore we cannot guarantee that the symplectic coherence is non-increasing under classical mixing.

As an example, consider the convex combination of the single-mode vacuum state and the coherent state $\ket \alpha$ with first moment of position and momentum quadrature being $\alpha_R$ and $\alpha_I$ respectively, with $\alpha_R, \alpha_I \neq 0$. Now, $\mathfrak c_{\ket 0 \bra 0} = 0$ and $\mathfrak c_{\ket \alpha \bra \alpha} = 0$, but for
\begin{equation}
    \rho = \frac12 \ket 0 \bra 0 + \frac12 \ket \alpha \bra \alpha,
\end{equation}
we have
\begin{equation}
    V_{xp}^\rho = \frac14 \alpha_R \alpha_I,
\end{equation}
and hence
\begin{equation}
    \mathfrak c_\rho = \frac{1}{16} \alpha_R^2 \alpha_I^2.
\end{equation}
This provides a concrete example of increasing symplectic coherence by classical mixing.
\subsection{A no-go result for the monotonicity of faithful measures of position-momentum correlations}\label{app:active_increase_sc}
\noindent We restate Theorem \ref{theo:active_increase_sc}:
\begin{theo}
    [Position-momentum correlation monotonicity no-go]\label{apptheo:active_increase_sc}
    There is no faithful measure of position-momentum correlations which is non-increasing under all block-diagonal Gaussian symplectic operations.
\end{theo}
\begin{proof}
    The symplectic matrix corresponding to a block-diagonal symplectic Gaussian $\hat{G}_{BD}$ is given by
    \begin{equation}
        S_{BD} = \begin{bmatrix}
            A && 0 \\ 0 && (A^T)^{-1}
        \end{bmatrix},
    \end{equation}
where $A$ is any real and invertible matrix. Its action on a covariance matrix $V^\rho$ is given by
    \begin{equation}
    V^\rho \mapsto \begin{bmatrix}A && 0 \\ 0 && (A^T)^{-1}\end{bmatrix} \begin{bmatrix}
        V_x^\rho && V_{xp}^\rho \\
        (V_{xp}^\rho)^T && V_p^\rho
    \end{bmatrix} \begin{bmatrix}A^T&& 0 \\ 0 && A^{-1}\end{bmatrix} = \begin{bmatrix}
        A V_x^\rho A^T && A V_{xp}^\rho A^{-1} \\ (A^T)^{-1} (V_{xp}^\rho)^T A^T && (A^T)^{-1} V_p^\rho A^{-1}
    \end{bmatrix}.
\end{equation}
Therefore, any faithful monotone (let us denote it by $||.||_{\mathcal{FM}}$) of position-momentum correlations which is non-increasing under all block-diagonal Gaussian symplectic operations should satisfy
\begin{eqnarray}
    ||M||_{\mathcal{FM}} &=& 0, \text{ if and only if } M = \bm 0_m, \nonumber \\
    ||AMA^{-1}||_{\mathcal{FM}} &\leq& ||M||_{\mathcal{FM}},
\end{eqnarray}
for any real matrix $M$ and real, invertible, non-orthogonal matrix $A$. This implies
\begin{equation}
    ||M||_{\mathcal{FM}} = ||A^{-1}(AMA^{-1})A||_{\mathcal{FM}} \leq ||AMA^{-1}||_{\mathcal{FM}}.
\end{equation}
The only way both inequalities are satisfied is when
\begin{equation}
    ||AMA^{-1}||_{\mathcal{FM}} = ||M||_{\mathcal{FM}},
\end{equation}
for any real matrix $M$ and real invertible matrix $A$.
This implies that $M\mapsto||M||_{\mathcal{FM}}$ is a function of the eigenvalues of $M$ only \cite{Grossmann2021}. However, in that case, the monotone $||M||_\mathcal{FM}$ will give the amount of position-momentum correlations for a quantum state with off block-diagonal matrix
\begin{equation}\label{appeq:finite_PM_covar}
    V_{xp} = \begin{bmatrix}
        0 && 1 \\ 0 && 0
    \end{bmatrix} 
\end{equation}
and for a quantum state with off block-diagonal matrix 
\begin{equation}
      V_{xp} = \begin{bmatrix}
        0 && 0 \\ 0 && 0
    \end{bmatrix}  
\end{equation}
to be the same and equal to zero, even though the quantum state with $V_{xp}$ given by \cref{appeq:finite_PM_covar} is not a free state (it has non-zero position-momentum correlations). Therefore $||M||_\mathcal{FM}$ will not be faithful and we can conclude that we cannot build a faithful monotone which is non-increasing under all block-diagonal Gaussian operations.
\end{proof}


\subsection{Symplectic coherence and quantum discord}\label{app:zero_sc_zero_discord}
\noindent To recap, we define a mapping $\mathcal M$ that allows us to map covariance matrix of a quantum state as density matrix of a virtual $2m$ dimensional qubit-qudit state:
\begin{defi}[Covariance matrix to density matrix mapping \cite{Barthe2025}]\label{appdefi:Covariance_to_density} 
    Given a covariance matrix
    \begin{equation}
        V = \begin{bmatrix}
            V_x && V_{xp} \\
            V_{xp}^T && V_p
        \end{bmatrix},
    \end{equation}
    we define a mapping $\mathcal M: V \mapsto \varrho$ which maps $V$ onto the density matrix $\varrho$ in the $2m$-dimensional qubit-qudit Hilbert space $\mathcal H_2 \otimes \mathcal H_m$ where the state of the qubit $\ket{0}$ or $\ket{1}$ denotes the position or momentum quadratures and the state of the qudit is the mode number we are dealing with. Including the normalization of the obtained virtual state, this mapping gives
    \begin{equation}\label{appeq:covariance_to_density}
        \mathcal{M}\left(V = \begin{bmatrix}
            V_x && V_{xp} \\
            V_{xp}^T && V_p
        \end{bmatrix}\right) = \varrho = \frac{1}{\Tr[V]} \begin{bmatrix}
            V_x && V_{xp} \\
            V_{xp}^T && V_p
        \end{bmatrix},
    \end{equation}
    such that $V_x$ denotes the block for $\ket{0}\bra{0} \otimes \ket{\bm i}\bra{\bm j}$, $V_{xp}$ for $\ket{0}\bra{1} \otimes \ket{\bm i}\bra{\bm j}$, etc., where $\ket{\bm i}$ are the qudit computational basis states (Figure \ref{fig:covar_to_density}).
\end{defi}
\noindent Then we have the following result:
\begin{theo}[Position-momentum correlations and quantum discord]\label{apptheo:zero_sc_zero_discord}
    Under the mapping $\mathcal{M}$ given in \cref{appdefi:Covariance_to_density}, a covariance matrix $V^\rho$ with $\rho \in \mathcal C$ maps to a classical-quantum virtual qubit-qudit state $\varrho = \mathcal M(V^\rho)$ with zero quantum discord with respect to computational basis measurements of the qubit register. Furthermore, the symplectic coherence of $\rho$ can be related to geometric quantum discord (\cref{eq:gqd}),with Hilbert-Schmidt distance as the distance measure, with respect to computational basis measurements in the qubit subsystem of $\varrho$ as
\begin{equation}
    \mathfrak c_\rho = \frac{\Tr[V^\rho]^2}2 D_G^{HS}(\varrho) = 8\left(\Tr[\rho \hat E] - \bm d^T \bm d\right)^2 D_G^{HS}(\varrho).
\end{equation}
\end{theo}
\begin{proof}
    If $\rho \in \mathcal C$,
    \begin{equation}
        V_{xp}^\rho = (V_{xp}^\rho)^T = 0.
    \end{equation}
    Then, $\varrho = \mathcal{M} (V^\rho)$ can be written as
    \begin{eqnarray}\label{eq:zero_discord}
        \rho &=& \frac{1}{\Tr[V^\rho]} \ket{0}\bra{0}\otimes \sum_{\bm {ij}} (V_x^\rho)_{\bm {ij}}\ket{\bm i}\bra{\bm j} + \frac{1}{\Tr[V^\rho]} \ket{1}\bra{1}\otimes \sum_{\bm {ij}} (V_p^\rho)_{\bm {ij}}\ket{\bm i}\bra{\bm j} \nonumber \\
        &=& p_x \ket 0 \bra 0 \otimes \varrho_x + p_p \ket 1 \bra 1 \otimes \varrho_p,
        \end{eqnarray}
    where $\ket {\bm i}$ denotes the computational basis of the qudit register, and where we have defined
    \begin{eqnarray}
        &&p_x \coloneqq \frac{\Tr[V_x^\rho]}{\Tr[V^\rho] }, \hspace{5mm} p_p \coloneqq \frac{\Tr[V_p^\rho]}{\Tr[V^\rho] }, \nonumber \\
        && \varrho_x = \frac{1}{\Tr[V_x^\rho]} \sum_{\bm {ij}} (V_x^\rho)_{\bm {ij}}\ket{\bm i}\bra{\bm j}, \nonumber \\
        &&\varrho_p \coloneqq \frac{1}{\Tr[V_p^\rho]} \sum_{\bm {ij}} (V_p^\rho)_{\bm {ij}}\ket{\bm i}\bra{\bm j}.
    \end{eqnarray}
    Eq.~\ref{eq:zero_discord} is exactly the condition for zero quantum discord with respect to computational basis measurements of the qubit register \cite[Section I]{Streltsov2017}.

    To prove the second part, recall that the geometric quantum discord of a qubit-qudit quantum state $\varrho$ with Hilbert-Schmidt distance as the distance measure with respect to computational basis measurements in the qubit can be expressed in terms of the off-block diagonal matrix of 
\begin{equation}
    \varrho = \begin{bmatrix}
        \varrho_{00} && \varrho_{01} \\ \varrho_{01}^T && \varrho_{11}
    \end{bmatrix}
\end{equation}
as (\cref{appeq:gqd_exp})
\begin{equation}
    D_G^{HS}(\varrho) = 2||\varrho_{01}||_1^2.
\end{equation}
Furthermore, using the mapping $\mathcal{M}$, for $\varrho = \mathcal M(V^\rho)$, $\varrho_{01} = \frac{V_{xp}}{\Tr[V^\rho]}$. Therefore, we finally get that
\begin{equation}
    D_G^{HS}(\varrho) = 2\frac{||V_{xp}||_F^2}{\Tr[V^\rho]^2} = 2 \frac{\mathfrak c_\rho}{\Tr[V^\rho]^2}.
\end{equation}
Finally, remembering that
\begin{equation}
    \Tr[\rho \hat E] = \frac{\Tr[V^\rho] + \bm d^T \bm d}{4}
\end{equation}
we get that
\begin{equation}
    \mathfrak c_\rho = \frac{\Tr[V^\rho]^2}2 D_G^{HS}(\varrho) = 8\left(\Tr[\rho \hat E] - \bm d^T \bm d\right)^2 D_G^{HS}(\varrho).
\end{equation}
    \end{proof}

    Note that we may also prove a converse result: let us look at the class of qubit-qudit states forming the image of the mapping $\mathcal M$, the class of real symmetric matrices $\varrho$ 
    \begin{equation}
        \varrho = \begin{bmatrix}
            \varrho_0 && \varrho_{01} \\
            \varrho_{01}^T && \varrho_1
        \end{bmatrix}
    \end{equation}
    satisfying
    \begin{eqnarray}
        \varrho &\succ& 0, \nonumber \\
        \varrho + \frac{i\Omega}{c} &\succeq& 0, \text{ for some constant } c \geq 2m. \nonumber \\
        \varrho_0 &\succ& 0, \nonumber \\
        \varrho_1 &\succ& 0.
    \end{eqnarray}
    Then, zero discord with respect to computational basis measurements of the qubit would mean that $\rho = p_0 \ket 0 \bra0 \otimes \sum_{\bm i,\bm j = 1}^m (\varrho_0)_{i,j} \ket{\bm i}\bra{\bm j} + p_1 \ket 1 \bra1 \otimes \sum_{\bm i,\bm j = 1}^m (\varrho_1)_{i,j} \ket{\bm i}\bra{\bm j} $ and hence can be written in the form
    \begin{equation}
        \varrho = \begin{bmatrix}
            p_0 \varrho_0 && 0 \\
            0 && p_1 \varrho_1
        \end{bmatrix}.
    \end{equation}
    Given $c$ such that $\varrho + \frac{i\Omega}{c} \succeq 0$, we use  the inverse mapping $\mathcal M^{-1}$ to define a covariance matrix $V = c\varrho$ (However, note that instead of $c$, multiplying $\varrho$ by any other constant $d \geq c$ also gives a valid covariance matrix). This will map to a covariance matrix with zero symplectic coherence, since $V_{xp} = V_{xp}^T=0$.

\subsection{Maximal symplectic coherence of quantum states}\label{app:max_sc}
\noindent We restate Theorem \ref{theo:max_sc}:
\begin{theo}[Maximal symplectic coherence of quantum states]
         Given an $m$-mode quantum state $\rho$ with covariance matrix $V^\rho$ such that $\Tr[V^\rho] = E$, the symplectic coherence of $\rho$ is upper bounded as
    \begin{eqnarray}\label{appeq:max_sc}
        \mathfrak{c}(\rho) &\leq& \frac{(E-2m)^2}{4} + (E-2m) \nonumber \\ &\equiv& \mathfrak c_\mathrm{max}(E,m).
    \end{eqnarray}
     Moreover, the set of pure Gaussian states of maximal symplectic coherence is given by $\mathcal{S}^{\max}(E)$ (Definition \ref{defi:max_SC_states}), while the covariance matrix of set of all states with maximal symplectic coherence is contained in the convex hull of covariance matrices of states in $\mathcal{S}^{\max}(E)$.
\end{theo}
\noindent We separate the proof in three parts:
\begin{itemize}
    \item Maximal symplectic coherence of pure Gaussian states (section \ref{appsubsec:max_sc_pure}).
      \item Structure of pure Gaussian states with maximal symplectic coherence (section \ref{appsubsec:max_structure}).
    \item Maximal symplectic coherence of mixed Gaussian states or non-Gaussian states and structure of covariance matrices of maximally symplectic coherent states (section \ref{appsubsec:max_sc_mixed}).
\end{itemize}
\subsubsection{Maximal symplectic coherence of pure Gaussian states}\label{appsubsec:max_sc_pure}

\noindent In this section, we prove the following result:

\begin{lem}[Maximal symplectic coherence of pure Gaussian states]\label{lem:max_sc_pure}
    Given an $m$-mode pure Gaussian state $\rho$ with covariance matrix $V^\rho$ such that $\Tr[V^\rho] = E$, the symplectic coherence of $\rho$ is upper bounded as
    \begin{equation}
        \mathfrak{c}(\rho) \leq \frac{(E-2m)^2}{4} + (E-2m)\equiv \mathfrak c_\mathrm{max}(E,m).
    \end{equation}
\end{lem}
\noindent Before proving Lemma \ref{lem:max_sc_pure}, we would like to prove a few technical results that would be useful in its proof.
\begin{lem}\label{lem:max_SC_technical_1}
    Given a real, positive definite matrix A and a real-symmetric matrix M,
    \begin{equation}\label{appeq:lemma_9_eq}
        ||A^{1/2}MA^{-1/2}||_F^2 \geq ||M||_F^2.
    \end{equation}
\end{lem}
\begin{proof}
Since A is positive definite, we can define
\begin{equation}
    A^{1/2} = O D^{1/2} O^T
\end{equation}
for some orthogonal matrix $O$ and a diagonal matrix $D$ with diagonal entries $d_i$ such that $d_i > 0 \forall i \in \{1,\dots,m\}$. Similarly,
\begin{equation}
    A^{-1/2} = O D^{-1/2} O^T.
\end{equation}
Therefore,
\begin{equation}
    ||A^{1/2}MA^{-1/2}||_F^2 = ||O D^{1/2} O^T M O D^{-1/2} O^T||_F^2 = ||D^{1/2} \tilde M D^{-1/2}||_F^2,
\end{equation}
where we have used the fact that Frobenius norm is invariant under orthogonal matrix, and where we have defined
\begin{equation}
    \tilde M \coloneqq O^T M O.
\end{equation}
Now,
\begin{eqnarray}
    ||D^{1/2} \tilde M D^{-1/2}||_F^2 = \sum_{i,j=1}^m (D^{1/2} \tilde M D^{-1/2})_{i,j}^2 = \sum_{i,j=1}^m \frac{d_i}{d_j} \tilde M_{i,j}^2 = \frac12 \sum_{i,j=1}^m \left(\frac{d_i}{d_j} + \frac{d_j}{d_i}\right) \tilde M_{i,j}^2.
\end{eqnarray}
The last equality follows from the fact that $\tilde M$ is a symmetric matrix. Using the fact that $x + 1/x \geq 2, \forall x > 0$ with equality if and only if $x = 1$, choosing $x = \frac{d_i}{d_j}$, we have, $\forall i,j$,
\begin{equation}
    ||D^{1/2} \tilde M D^{-1/2}||_F^2 \geq \frac{1}{2} \times 2 \sum_{i,j=1}^m \tilde M_{i,j}^2 = ||\tilde M||_F^2 = ||O^T M O ||_F^2 = ||M||_F^2.
\end{equation}
\end{proof}
\begin{lem}\label{lem:max_SC_technical_2}
    Given a $m \times m$ real positive definite matrix $A$, $C = \mathrm{diag}(\cos \theta_1,\dots,\cos \theta_m)$ and $S = \mathrm{diag}(\sin \theta_1,\dots,\sin \theta_m)$,
    \begin{equation}
        ||CAS - CA^{-1}S||_F^2 \geq ||CAS - SA^{-1}C||_F^2.
    \end{equation}
\end{lem}
\begin{proof}
We expand:
    \begin{equation}
        ||CAS - CA^{-1}S||_F^2 - ||CAS - SA^{-1}C||_F^2 = 2\Tr[ACSA^{-1}CS] - \Tr[AC^2A^{-1}S^2] - \Tr[AS^2A^{-1}C^2].
    \end{equation}
    We have
    \begin{equation}
        \Tr[AS^2A^{-1}C^2] = \Tr[(C^2 A^{-1} S^2 A)^T] =\Tr[C^2 A^{-1} S^2 A] = \Tr[AC^2A^{-1}S^2]
    \end{equation}
    since the trace of a matrix and its transpose is the same. Therefore,
    \begin{equation}
        ||CAS - CA^{-1}S||_F^2 - ||CAS - SA^{-1}C||_F^2 = 2(\Tr[ACSA^{-1}CS] - \Tr[AC^2A^{-1}S^2]).
    \end{equation}
    And
    \begin{eqnarray}
        \Tr[ACSA^{-1}CS] - \Tr[AC^2A^{-1}S^2] &=& \Tr[ACSA^{-1}CS] + \Tr[AC^2A^{-1}C^2] - \Tr[C^2] \nonumber \\ &=& ||A^{1/2}CSA^{-1/2}||_F^2 + ||A^{1/2}C^2 A^{-1/2}||_F^2 -\Tr[C^2] \nonumber \\
        &\geq& ||CS||_F^2 + ||C^2||_F^2 - \Tr[C^2] \nonumber \\
        &=& \Tr[C^2S^2] + \Tr[C^4] - \Tr[C^2] \nonumber \\
        &=&0,
    \end{eqnarray}
    where the third line follows from Lemma \ref{lem:max_SC_technical_1} and where we used $CS=SC$ and $C^2+S^2=\I_m$.
\end{proof}
\begin{lem}\label{lem:max_SC_technical_3}
    Given a $m \times m$ positive definite matrix $A$ with eigenvalues greater than equal to one, $C = \mathrm{diag}(\cos \theta_1,\dots,\cos \theta_m)$ and $S = \mathrm{diag}(\sin \theta_1,\dots,\sin \theta_m)$,
    \begin{equation}
        ||CAS - CA^{-1}S||_F^2 \leq \frac14 ||(A - A^{-1})||_F^2.
    \end{equation}
\end{lem}
\begin{proof}
    \begin{eqnarray}
        ||CAS - CA^{-1}S||_F^2 &\coloneqq& ||CMS||_F^2 = \Tr[C^2 M S^2 M] = \Tr[EF],
    \end{eqnarray}
where we have defined
\begin{eqnarray} 
    M &\coloneqq& A - A^{-1}, \nonumber \\
    E &\coloneqq& M^{1/2} C^2 M^{1/2}, \nonumber \\
    F &\coloneqq& M^{1/2}S^2M^{1/2},
\end{eqnarray}
Note that since $A$ with eigenvalues greater than equal to one, is positive definite, $M$ is positive semi-definite and hence $M^{1/2}$ exists. This gives
\begin{equation}
    E + F = M^{1/2}(C^2 + S^2)M^{1/2} = M.
\end{equation}
Therefore,
\begin{eqnarray}\label{eq:lem_9_1}
    \Tr[(E+F)^2] &=& \Tr[M^2] = \Tr[E^2 + F^2] + 2\Tr[EF] \nonumber \\
    2\Tr[EF] &=& \Tr[M^2] - \Tr[E^2 + F^2].
\end{eqnarray}
Now, $E-F$ is a symmetric matrix, therefore
\begin{equation}\label{appeq:diff_CS_conjugated_M}
    \Tr[(E-F)^2] = \Tr[(E-F)^T (E-F)] = ||E-F||_F^2 \geq 0,
\end{equation}
which implies
\begin{eqnarray}
    \Tr[(E-F)^2] &=& \Tr[E^2 + F^2] - 2\Tr[EF] \geq 0 \nonumber \\
    2\Tr[EF] &\leq& \Tr[E^2 + F^2] \nonumber \\
    -2\Tr[EF] &\geq& - \Tr[E^2 + F^2].
\end{eqnarray}
Putting this in Eq.~\ref{eq:lem_9_1},
\begin{eqnarray}
    2\Tr[EF] &\leq& \Tr[M^2] -2\Tr[EF] \nonumber\\
    \Tr[EF] &\leq& \frac14 \Tr[M^2] \nonumber \\
    ||CMS||_F^2 &\leq& \frac{1}{4} ||M||_F^2 \nonumber \\
    ||C(A-A^{-1})S||_F^2 &\leq& \frac14 ||A-A^{-1}||_F^2.
\end{eqnarray}
\end{proof}
\noindent We are now ready to prove Lemma \ref{lem:max_sc_pure}.

\begin{proof}[Proof of Lemma \ref{lem:max_sc_pure}] As mentioned in the preliminaries (section \ref{sec:prelims}), the covariance matrix of a pure $m$-mode Gaussian state can be written as
\begin{equation}
    V^G = S_U Z S_U^T,
\end{equation}
where 
\begin{equation}
    Z = \begin{bmatrix}
        D && 0 \\
        0 && D^{-1}
    \end{bmatrix},
\end{equation}
where $D$ is an $m \times m$ diagonal matrix with entries $d_i$ such that $d_i \geq 1 \forall i \in \{1,\dots,m\}$, whereas $S_U$ is an $2m \times 2m$ symplectic orthogonal matrix isomorphic to an $m\times m$ unitary matrix. By the Cartan (KAK) decomposition, the unitary group can be decomposed as \cite[Section V]{edelman2022cartan}
\begin{equation}
    U(m) = O(m) \circ e^{i\vec\theta_m} \circ  O(m)
\end{equation}
And since the symplectic matrix of a passive linear unitary gate $S_U$ is isomorphic to a unitary matrix, therefore $S_U$ can be decomposed as
\begin{equation}
    S_U = S_{O_1} S_{\vec \theta} S_{O_2},
\end{equation}
where $S_{O_1}$ and $S_{O_2}$ are (possibly different) symplectic matrices corresponding to block-diagonal orthogonal unitary gates ($S_{O_1}$ and $S_{O_2}$ are isomorphic to $m\times m$ orthogonal matrices), and $S_{\vec \theta}$ is the phase shift gate given by
\begin{equation}
    S_{\vec \theta} = \begin{bmatrix}
        C && S \\
        -S && C
    \end{bmatrix},
\end{equation}
with $C = \mathrm{diag}(\cos \theta_1,\dots, \cos \theta_m)$ and $S = \mathrm{diag}(\sin \theta_1,\dots, \sin \theta_m)$. Since block-diagonal orthogonal matrices are free operations for symplectic coherence, we remove $S_{O_1}$, and we need to bound the pure Gaussian state, lets call it $\ket G$ corresponding to the covariance matrix $S_{\vec \theta} S_{O_2} Z S_{O_2}^T S_{\vec \theta}^T$.
Now
\begin{equation}
    S_{\vec \theta} S_{O_2} Z S_{O_2}^T S_{\vec \theta}^T = \begin{bmatrix}
        CAC + SA^{-1}S && -CAS + S A^{-1} C \\
        - SAC + CA^{-1}S && SAS + CA^{-1}C
    \end{bmatrix},
\end{equation}
where $A = O D O^T$ and $A^{-1} = O D^{-1} O^T$, $O$ is a $m\times m$ orthgonal matrix defined such that
\begin{equation}
    S_{O_2} = \begin{bmatrix}
        O && 0 \\
        0 && O
    \end{bmatrix}.
\end{equation}
Therefore,
\begin{equation}
    \mathfrak{c}_{\ket G} = ||-CAS+ SA^{-1}C||_F^2 = ||CAS - S A^{-1} C||_F^2.
\end{equation}
From Lemmas \ref{lem:max_SC_technical_2} and \ref{lem:max_SC_technical_3},
\begin{equation}
    ||CAS - S A^{-1} C||_F^2 \leq ||CAS - CA^{-1}S||_F^2 \leq \frac14 ||A - A^{-1}||_F^2.
\end{equation}
Now,
\begin{eqnarray}\label{appeq:max_sc_pure_final_step}
    \frac14 ||A - A^{-1}||_F^2 &=& \frac{1}{4}||D-D^{-1}||_F^2 = \frac14 \sum_{i=1}^m \left(d_i - \frac{1}{d_i}\right)^2 \nonumber \\
    &=& \frac{1}{4} \sum_{i=1}^m \left(\left(d_i + \frac{1}{d_i}\right)^2 - 4\right) \\
    &=& \frac{1}{4} \sum_{i=1}^m \left(d_i + \frac{1}{d_i} - 2\right)\left(d_i + \frac{1}{d_i} - 2 + 4\right) \nonumber \\
    &=& \frac{1}{4} \sum_{i=1}^m \left(d_i + \frac{1}{d_i} - 2\right)^2 + \sum_{i=1}^m \left(d_i + \frac{1}{d_i} - 2\right) \nonumber \\
    &\leq& \frac14 \left(\sum_{i=1}^m d_i + \frac{1}{d_i} - 2\right)^2 + \sum_{i=1}^m \left(d_i + \frac{1}{d_i} - 2\right),
\end{eqnarray}
where in the last step we used $d_i + \frac{1}{d_i} - 2\ge0$.
Now, given that,
\begin{equation}
    \Tr[S_{\vec \theta} S_{O_2} Z S_{O_2}^T S_{\vec \theta}^T] = \Tr[Z] = \sum_{i=1}^m \left(d_i +  \frac{1}{d_i}\right) = E
\end{equation}
since $S_{O_2}$ and $S_{\vec \theta}$ are orthogonal matrices. Therefore,
\begin{equation}
    \frac14 ||A - A^{-1}||_F^2 \leq \frac14 (E-2m)^2 + (E-2m),
\end{equation}
and thus
\begin{equation}
    \mathfrak{c}_{\ket G} \leq \mathfrak c_{\max} = \frac14 (E-2m)^2 + (E-2m).
\end{equation}
\end{proof}
\subsubsection{Structure of states with maximal symplectic coherence}\label{appsubsec:max_structure}
\noindent We first recap the set of pure Gaussian states $\mathcal{S}^{\max}(E)$:
\begin{defi}\label{appdefi:max_SC_states}
    A pure Gaussian state $\ket\psi_G$ belongs to the set $\mathcal S^{\max}(E)$ if and only if it can be written in the form
    \begin{equation}
        \ket{\psi}_G = \hat O_2\hat R(\vec \theta)\hat{O}_1\ket{r}\otimes\ket{0}^{m-1},
    \end{equation}
    up to arbitrary displacements at the end. $\ket r$ is the squeezed state with the squeezing parameter such that
    \begin{equation}\label{appeq:focus_squeeze}
        e^{2r} + e^{-2r} = E - 2(m-1),
    \end{equation}
    the phase shift angles $\vec \theta$ and the block-diagonal orthogonal unitary gate $\hat O_1$ with the associated symplectic matrix 
    \begin{equation}
        \begin{bmatrix}
            O && 0 \\
            0 && O
        \end{bmatrix}
\end{equation}
are defined such that
\begin{equation}
    2 |A_{1j}| = \delta_{1j}
\end{equation}
for $A = O^T C O,O^T S O,O^T(CS)O$, where $C = \mathrm{diag}(\cos \theta_1,\dots,\cos \theta_m)$ and $S = \mathrm{diag}(\sin \theta_1,\dots, \sin \theta_m)$. The block-diagonal orthogonal unitary gate $\hat O_2$ can be arbitrary.
\end{defi}
\noindent Then we have the following Lemma:
\begin{lem}[Compression Lemma]
For the set of quantum states $\rho$ with $\Tr[V^\rho] = E$, the set of pure Gaussian states with maximal symplectic coherence is given by $\mathcal{S}^{\max} (E)$ .
\end{lem}
\begin{proof}
    This result is obtained by analyzing the inequalities (Eq.~\ref{appeq:max_sc_pure_final_step}, Eq.~\ref{appeq:diff_CS_conjugated_M}, Eq.~\ref{appeq:lemma_9_eq}) used in proving \cref{lem:max_sc_pure} and determining the necessary and sufficient conditions for which they become an equality. Firstly, the inequality in Eq.~\ref{appeq:max_sc_pure_final_step}
    \begin{equation}\label{appeq:diagonal_focusing}
        \frac{1}{4} \sum_{i=1}^m \left(d_i + \frac{1}{d_i} - 2\right)^2 \leq \frac14 \left(\sum_{i=1}^m d_i + \frac{1}{d_i} - 2\right)^2.
    \end{equation}
    For $d_i \geq 1$, this is an equality if and only if all $d_i + 1/d_i - 2 = 0$ for all $i$ except one (for convenience, let it be non-zero for the first one). This would imply
    \begin{eqnarray}
        d_i + \frac{1}{d_i} - 2 &=& 0 \nonumber \\
         d_i + \frac{1}{d_i} &=&2 \nonumber \\
         d_i&=& 1, \forall i \in \{2,\dots,m\}.
    \end{eqnarray}
    This corresponds to starting with vacuum states in modes $2$ to $m$. Now,
    \begin{eqnarray}
        \sum_{i=1}^m d_i + \frac{1}{d_i} &=& E \nonumber \\
        d_1 + \frac{1}{d_1} &=& E - 2(m-1).
    \end{eqnarray}
    Since $d_1$ relates to the squeezing parameter as 
    \begin{equation}
        d_1 = e^{2r},
    \end{equation}
    with $r\geq0$, the squeezing parameter on the first mode should be such that
    \begin{equation}
        e^{2r} + e^{-2r} = E - 2(m-1).
    \end{equation}
    Another inequality that we use in Lemma \ref{lem:max_SC_technical_3} is (Eq.~\ref{appeq:diff_CS_conjugated_M})
    \begin{equation}
        ||M^{1/2} C^2 M^{1/2} - M^{1/2} S^2 M^{1/2}||_F^2 \geq 0.
    \end{equation}
    for a positive semi-definite matrix $M = A - A^{-1}$. By the property of the Frobenius norm, this becomes an equality if and only if $M^{1/2} C^2 M^{1/2} - M^{1/2} S^2 M^{1/2}$ is the zero matrix, i.e.,
    \begin{equation}\label{appeq:c^2=s^2}
        M^{1/2} C^2 M^{1/2} = M^{1/2} S^2 M^{1/2}.
    \end{equation}
    Given that $A = O D O^T$ for an $m \times m$ orthogonal matrix $O$ and $D = \mathrm{diag}(d_1,d_2,\dots,d_m)$ and following our discussion after Eq.~\ref{appeq:diagonal_focusing} , we know that for maximal symplectic coherence, we need $d_i = 0, \forall i \neq 1$ and $d_1 \geq 1$. We get
    \begin{equation}
        M = O B O^T,
    \end{equation}
    where $B = \mathrm{diag}(b_1,0,\dots,0)$ with $b_1 = \sqrt{d_1} - 1/\sqrt{d_1} \geq 0$. Eq.~\ref{appeq:c^2=s^2} can therefore be written as
    \begin{eqnarray}
        O B O^T C^2 O B O^T &=& O B O^T S^2 O B O^T \nonumber \\
        B O^T C^2 O B&=& B O^T C^2 O B,
    \end{eqnarray}
    where the second line follows from multiplying both sides from the left by $O^T$ and from the right by $O$. Element-wise analysis of the equality implies
    \begin{equation}\label{appeq:c^2=s^2_2}
        (B O^T C^2 O B)_{ij} = (B O^T S^2 O B)_{ij}.
    \end{equation}
    $\forall i,j$. Now,
    \begin{equation}
        (B O^T C^2 O B)_{ij} = \sum_{k,l=1}^m B_{ik} (O^T C^2 O)_{kl} B_{lj} = \sum_{k,l=1}^m b_1^2\delta_{1i}\delta_{1k} \delta_{1l} \delta_{1j} (O^T C^2 O)_{kl} = b_1^2\delta_{1i}\delta_{1j} (O^T C^2 O)_{11}.
    \end{equation}
    Similarly,
    \begin{equation}
        (B O^T S^2 O B)_{ij}  = b_1^2 \delta_{1i}\delta_{1j} (O^T S^2 O)_{11}.
    \end{equation}
    Therefore, for $i$ and $j$ are both not equal to one, Eq.~\ref{appeq:c^2=s^2_2} is already satisfied since both left and right hand side are zero, whereas for $i=j=1$, we require
    \begin{equation}\label{appeq:c^2=s^2_3}
        (O^T C^2 O)_{11} = (O^T S^2 O)_{11}.
    \end{equation}
    We will get back to this equality later. Finally, we want to  see when the inequality Eq.~\ref{appeq:lemma_9_eq} with $M=CS$,
    \begin{equation}
        ||A^{1/2}CSA^{-1/2}||_F^2 = ||D^{1/2} O^T CS O D^{-1/2}||_F^2 \geq ||O^T CS O||_F^2 = ||CS||_F^2
    \end{equation}
    becomes an equality.
    \begin{eqnarray}
        ||D^{1/2} O^T CS O D^{-1/2}||_F^2 &=& \frac12 \sum_{i,j=1}^m \left(\frac{d_i}{d_j} + \frac{d_j}{d_i}\right) (O^T CS O)_{ij}^2 \nonumber \\
        &=& (O^T CS O)_{11}^2 + \frac12\sum_{i,j\neq1}\left(\frac{d_i}{d_j} + \frac{d_j}{d_i}\right) (O^T CS O)_{ij}^2 \nonumber \\
        &&+ \frac12 \sum_{j\neq 1}\left(\frac{d_1}{d_j} + \frac{d_j}{d_1}\right) (O^T CS O)_{1j}^2 + \frac12 \sum_{i\neq 1}\left(\frac{d_i}{d_1} + \frac{d_1}{d_i}\right) (O^T CS O)_{i1}^2 \nonumber \\
        &=& (O^T CS O)_{11}^2 + \sum_{i,j\neq1}(O^T CS O)_{ij}^2 + \left(d_1 + \frac{1}{d_1}\right)\sum_{j\neq1}(O^T CS O)_{1j}^2,
    \end{eqnarray}
    where the last line follows from the fact that maximal symplectic coherence is possible if and only if we have $d_i = 1 \forall i \in \{2,\dots,m\}$ and the fact that $O^T CS O$ is a symmetric matrix.
    We want this to be equal to 
    \begin{equation}
        ||O^T CS O||_F^2 = \sum_{i,j=1}^m (O^T CS O)_{ij}^2 = (O^T CS O)_{11}^2 + \sum_{i,j\neq1}(O^T CS O)_{ij}^2 + \sum_{i=1,j\neq1}2(O^T CS O)_{1j}^2.
    \end{equation}
    Now since $d_1 > 1$ (unless if $E = 2m$ in which case it is trivially the vacuum state), $d_1 + 1/d_1 > 2$, therefore the inequality becomes an equality if and only $(O^T CS O)_{1j} = 0$ for all $j \neq 1$. Therefore, the condition on $O$ and hence $\hat{O}_1$ can be given as
    \begin{equation}\label{appeq:zeros_CS}
        (O^T CS O)_{1j} = 0,
    \end{equation}
    $\forall j\neq1$. Note that a similar analysis for $||A^{1/2}C^2A^{-1/2}||_F^2$ would require
    \begin{equation}\label{appeq:zeros_C^2}
        (O^T C^2 O)_{1j} = 0,
    \end{equation}
    $\forall j \neq 1$. We rewrite Eq.~\ref{appeq:c^2=s^2_3}
    \begin{equation}
        (O^T C^2 O)_{11} = (O^T S^2 O)_{11}.
    \end{equation}
    We also know that
    \begin{equation}
        (O^T (C^2 + S^2) O)_{ij} = (O^T C^2 O)_{ij} + (O^T S^2 O)_{ij} = (O^T O)_{ij} =\delta_{ij}.
    \end{equation}
    Using Eq.~\ref{appeq:c^2=s^2_3} and Eq.~\ref{appeq:zeros_C^2}, this implies
    \begin{eqnarray}\label{appeq:C^2_S^2_condition}
        (O^T C^2 O)_{1j} &=& (O^T S^2 O)_{1j} = \begin{cases}
            \frac12, \text{ if } j =1, \\
            0, \text{ if } j \neq1.
        \end{cases}
    \end{eqnarray}
   Now,
    \begin{equation}
        (O^T C^4 O)_{11} = \sum_{j=1}^m (O^T C^2 O)_{1j} (O^T C^2 O)_{j1} = \sum_{j=1}^m (O^T C^2 O)_{1j}^2 = \frac14,
    \end{equation}
    where the last equality follows from Eq.~\ref{appeq:C^2_S^2_condition}. Similarly, 
    \begin{equation}
        (O^T S^4 O)_{11} = \frac14.
    \end{equation}
    Therefore,
    \begin{equation}
        (O^T (C^2 + S^2)^2 O)_{11} = (O^T C^4 O)_{11} + (O^T S^4 O)_{11} + 2 (O^T C^2 S^2 O)_{11} = 1.
    \end{equation}
    This gives
    \begin{equation}
        (O^T C^2 S^2 O)_{11} = \frac14.
    \end{equation}
    Also,
    \begin{equation}
        (O^T C^2 S^2 O)_{11} = \sum_{j=1}^m (O^T CS O)_{1j} (O^T CS O)_{j1} = \sum_{j=1}^m (O^T CS O)_{1j}^2 = (O^T CS O)_{11}^2.
    \end{equation}
    Where the last equality follows from Eq.~\ref{appeq:zeros_CS}. This gives
    \begin{equation}
        (O^T CS O)_{11} = \pm \frac12.
    \end{equation}
    Combining this with Eq.~\ref{appeq:zeros_CS},
    \begin{eqnarray}\label{appeq:CS_condition}
        (O^T CS O)_{1j} &=& \begin{cases}
            \pm \frac12, \text{ if } j =1, \\
            0, \text{ if } j \neq1.
        \end{cases}
        \end{eqnarray}
    Note that the conditions on $(O^TC^2 O)_{1j}$, $(O^TS^2 O)_{1j}$ (Eq.~\ref{appeq:C^2_S^2_condition}) and on $(O^TCS O)_{1j}$ (Eq.~\ref{appeq:CS_condition}) also holds when you swap the indices $1$ and $j$ since these are all symmetric matrices.
\end{proof} 

\subsubsection{Maximal symplectic coherence of mixed Gaussian states}\label{appsubsec:max_sc_mixed}

\noindent In this section, we prove the following result:

\begin{theo}[Maximal symplectic coherence of mixed Gaussian states]\label{theo:max_sc_mixed}
    Given the covariance matrix $V^\sigma$ of a (possibly) mixed or non-Gaussian $m$-mode state $\sigma$ with the property that $\Tr[V^\sigma] = E$, it holds that
    \begin{equation}
        \mathfrak c_\sigma \leq c_\mathrm{max} = \frac{(E-2m)^2}{4} + (E-2m)
    \end{equation}
    with equality if and only if
    \begin{equation}
        V^{\sigma} = \frac12 V^{G_1} + \frac12  V^{G_2}
    \end{equation}
    such that $V^{G_1}$ and $V^{G_2}$ are covariance matrices corresponding to pure Gaussian states $\ket{G_1}$ and $\ket{G_2}$ respectively, such that $\ket{G_1},\ket{G_2} \in \mathcal S^{\max}(E)$ (Definition \ref{appdefi:max_SC_states}) and such that $V_{xp}^{G_1} = V_{xp}^{G_2}$.
\end{theo}
\begin{proof}
    From \cite[Theorem 3]{parthasarathy2011}, any valid covariance matrix can be written as an equally weighted mixture of covariance matrices of two pure Gaussian states, i.e.,
    \begin{equation}\label{appeq:mixed_covariance_decomp}
        V^\sigma = \frac12 V^{G_1} + \frac{1}{2} V^{G_2}. 
    \end{equation}
    Note that this also implies
    \begin{equation}\label{appeq:mixed_energy_condition}
        \Tr[V^\sigma] = E = \frac{1}{2} \Tr[V^{G_1}] +\frac12 \Tr[V^{G_2}] = \frac12 E_1 + \frac12 E_2.
    \end{equation}
    Therefore,
    \begin{eqnarray}\label{appeq:mixed_inequality_1}
        ||V_{xp}^\sigma||_F &=& \left|\left|\frac12 V_{xp}^{G_1} + \frac12 V_{xp}^{G_2}\right|\right|_F \leq \frac12 ||V_{xp}^{G_1}||_F + \frac12 ||V_{xp}^{G_2}||_F \nonumber\\ &\leq& \frac12 \sqrt{\frac{(E_1 - 2m)^2}{4} + (E_1-2m)} + \frac12 \sqrt{\frac{(E_2 - 2m)^2}{4} + (E_2-2m)}. \nonumber \\
    \end{eqnarray}
    The first inequality used is the triangle inequality whereas the second inequality follows from Lemma \ref{lem:max_sc_pure}, i.e., $||V^{G_i}||_F \leq \sqrt{\frac{(E_i - 2m)^2}{4} + (E_i-2m)}$ for $\Tr[V^{G_i}] = E_i$. Call 
\begin{equation}
    x_1 \coloneqq E_1 - 2m, \hspace{5mm} x_2 \coloneqq E_2 - 2m,
\end{equation}
such that 
\begin{equation}
    \frac12 x_1 + \frac12 x_2 = E-2m,
\end{equation}
we have
\begin{eqnarray}
    ||V_{xp}^\rho||_F &\leq& \frac12 \sqrt{\frac{x_1^2}{4} + x_1} + \frac12 \sqrt{\frac{x_2^2}{4} + x_2} \equiv \frac12 f(x_1) + \frac12 f(x_2).
\end{eqnarray}
Now,
\begin{equation}\label{appeq:mixed_f}
    f(x) \coloneqq \sqrt{\frac{x^2}{4} + x}
\end{equation}
is a strictly concave function of $x$, since 
\begin{equation}
    f''(x) = - \frac{2}{(x^2 + 4x)^{3/2}} < 0,
\end{equation}
for $x \in (0,\infty)$. Therefore, using Jensen's inequality for the concave function $f(x)$, we get
\begin{eqnarray}
    ||V_{xp}^\rho||_F &\leq& \frac12 f(x_1) + \frac12 f(x_2) \leq  f\left(\frac12 x_1 + \frac12 x_2\right) \nonumber \\
    &=& \sqrt{\frac{(\frac12 x_1 + \frac12 x_2)^2}{4} + \frac12 x_1 + \frac12 x_2} = \sqrt{\frac{(E-2m)^2}{4} + (E-2m)}.
\end{eqnarray}
Finally,
\begin{equation}
    \mathfrak c_\sigma =  ||V_{xp}^\rho||_F^2 \leq \frac{(E-2m)^2}{4} + (E-2m).
\end{equation}
To see when this inequality is saturated, note that
\begin{eqnarray}
    ||A + B||_F^2 = \sum_{i,j=1}^m (a_{ij} + b_{ij})^2 = \sum_{i,j=1}^m a_{ij}^2 + \sum_{i,j=1}^m b_{ij}^2 + 2 \sum_{i,j=1}^m a_{ij}b_{ij}.
\end{eqnarray}
Now, by Cauchy-Schwartz inequality,
\begin{equation}
    \sum_{i,j=1}^m a_{ij}b_{ij} \leq \sqrt{\sum_{i,j=1}^m a_{ij}^2} \sqrt{\sum_{i,j=1}^m b_{ij}^2}
\end{equation}
with equality if and only if $b_{ij} = c a_{ij}$ for some non-negative constant $c$  $\forall i,j$ \cite{axler2015linear}. This would give 
\begin{equation}
    B = cA,
\end{equation}
in which case
\begin{equation}
    ||A + B||_F^2 = \sum_{i,j=1}^m a_{ij}^2 + \sum_{i,j=1}^m b_{ij}^2 + 2 \sqrt{\sum_{i,j=1}^m a_{ij}^2} \sqrt{\sum_{i,j=1}^m b_{ij}^2}= ||A||_F^2 + ||B||_F^2 + 2 ||A||_F ||B||_F = (||A||_F + ||B||_F)^2.
\end{equation}
and 
\begin{equation}
    ||A+B||_F = ||A||_F + ||B||_F.
\end{equation}
Applying this analysis to Eq.~\ref{appeq:mixed_inequality_1},
\begin{equation}
    \left|\left|\frac12 V_{xp}^{G_1} + \frac12 V_{xp}^{G_2}\right|\right|_F = \frac12 ||V_{xp}^{G_1}||_F + \frac12 ||V_{xp}^{G_2}||_F
\end{equation}
if and only if 
\begin{equation}\label{appeq:mixed_equality_one}
    ||V_{xp}^{G_2}||_F = c ||V_{xp}^{G_1}||_F.
\end{equation}
We will fix the constant $c$ later. Now, since $f(x)$ (Eq.~\ref{appeq:mixed_f}) is strictly concave, by Jensen's inequality,
\begin{equation}
    \frac12 f(x_1) + \frac12f(x_2) = f\left(\frac12 x_1 + \frac12 x_2\right)
\end{equation}
if and only if 
\begin{eqnarray}
    x_1 &=& x_2, \nonumber \\
    E_1 - 2m&=& E_2 -2m, \nonumber\\
    E_1&=&E_2.
\end{eqnarray}
This, combined with Eq.~\ref{appeq:mixed_energy_condition} implies
\begin{equation}
    E_1 = E_2 = E.
\end{equation}
Finally, we want both $V_{xp}^{G_1}$ and $V_{xp}^{G_2}$ to attain maximal values (Eq.~\ref{appeq:mixed_inequality_1}), that is
\begin{eqnarray}
    ||V_{xp}^{G_1}||_F &=& \sqrt{\frac{(E_1 - 2m)^2}{4} + (E_1 - 2m)} = \sqrt{\frac{(E - 2m)^2}{4} + (E - 2m)}, \nonumber \\
    ||V_{xp}^{G_2}||_F &=& \sqrt{\frac{(E_2 - 2m)^2}{4} + (E_2 - 2m)} = \sqrt{\frac{(E - 2m)^2}{4} + (E - 2m)}.
\end{eqnarray}
Since Eq.~\ref{appeq:mixed_equality_one} implies that
\begin{equation}
    ||V_{xp}^{G_2}||_F = c||V_{xp}^{G_1}||_F,
\end{equation}
both $V_{xp}^{G_1}$ and $V_{xp}^{G_2}$ can only attain maximal values simultaneously if $c=1$, i.e.
\begin{equation}
    V_{xp}^{G_1} = V_{xp}^{G_2}.
\end{equation}
To conclude, the symplectic coherence of a mixed state $\sigma$ such that the trace of its covariance matrix $V^\sigma$ is $\Tr[V^\sigma] = E$ is maximal if and only if the decompositon of its covariance matrix into a convex mixture of covariance matrices of pure Gaussian states $V^{G_1}$ and $V^{G_2}$ (Eq.~\ref{appeq:mixed_covariance_decomp}) is such that
\begin{eqnarray}
\Tr[V^{G_1}]&=&\Tr[V^{G_2}] = E,\nonumber \\
    V_{xp}^{G_1} &=& V_{xp}^{G_2}, \nonumber \\
    ||V_{xp}^{G_1}||_F^2 &=& ||V_{xp}^{G_2}||_F^2 = \frac{(E-2m)^2}{4} + (E-2m).
\end{eqnarray}
Therefore, for $\mathfrak c_\sigma$ to attain maximal possible value, the covariance matrix of $\sigma$ should be decomposable as a equally weighted convex hull of covariance matrices of two pure Gaussian states $\ket{G_1}$ and $\ket{G_2}$ such that $\ket{G_1},\ket{G_2}\in \mathcal{S}^{\max}(E)$, with the additional condition 
\begin{equation}
    V_{xp}^{G_1} = V_{xp}^{G_2}.
\end{equation}
\end{proof}
Note that if $\ket{G_1} = \ket{G_2}$, this condition is trivially true. Hence all pure Gaussian states in $\mathcal{S}^{\max}$ have maximal symplectic coherence.
\subsection{Behavior of symplectic coherence under perturbations of the state}\label{app:SC_perturbations}
\noindent We restate a more explicit version of Theorem \ref{theo:SC_perturbations} in the main text:
\begin{theo}[Symplectic coherence under perturbations]
  Given two quantum states $\rho$ and $\sigma$ with symplectic coherence $\mathfrak c_\rho$ and $\mathfrak c_\sigma$, and second moments of energies satisfying $\Tr[\hat E^2 \rho],\Tr[\hat E^2 \sigma] \leq E^2$, and $||\rho - \sigma||_1 \leq \epsilon/m$, then 
    \begin{eqnarray}
  |\mathfrak c_\rho - \mathfrak c_\sigma| &\leq& 800E^2 \epsilon + 40\sqrt 2 E \max(\sqrt{\mathfrak c_\rho},\sqrt{\mathfrak c_\sigma}) \sqrt{\epsilon} \nonumber \\
  &\leq & 800E^2 \epsilon + 40\sqrt 2 E \sqrt{(2E-m)^2 + 2(2E-m)} \sqrt{\epsilon}.
    \end{eqnarray}
\end{theo}
\begin{proof}
We first give an upper bound on Hilbert--Schmidt distance between the covariance matrices of two states in terms of their trace distance:
\begin{lem}\label{lem:HS_distance_bound}
    Given two quantum states $\rho$ and $\sigma$ with covariance matrices $V^\rho$ and $V^\sigma$, and such that $\Tr[\hat E^2 \rho],\Tr[\hat E^2 \sigma] \leq E^2$, we have
    \begin{equation}
        ||V^\rho - V^\sigma||_F \leq \sqrt{m}||V^\rho - V^\sigma||_\infty \leq 40E \sqrt{m||\rho - \sigma||_1}.
    \end{equation}
\end{lem}
\begin{proof}
    The right hand side follows from \cite[Theorem 62]{annamele2025symplecticrank}, while the left hand side follows from the standard matrix norm inequality $||A||_F \leq \sqrt{m}||A||_\infty$ for any $m \times m$ matrix A.
\end{proof}
\noindent Therefore, if $||\rho - \sigma||_1 \leq \epsilon/m$, then
\begin{equation}\label{appeq:perturbation_frobenius}
    ||V^\rho - V^\sigma||_F \leq 40 E \sqrt{\epsilon}.
\end{equation}
Now if $\tilde V^\rho$ is the covariance matrix with zero position-momentum correlations closest in Hilbert--Schmidt distance to $V^\rho$, then (Theorem \ref{appeq:sc_exp_operational})
\begin{equation}
    \mathfrak c_\rho = \frac12 ||V^\rho - \tilde V^\rho||_F^2.
\end{equation}
Similarly,
\begin{equation}
    \mathfrak c_\sigma = \frac12 ||V^\sigma - \tilde V^\sigma||_F^2.
\end{equation}
Since $\tilde V^\rho$ is the covariance matrix with zero position-momentum correlations closest to $V^\rho$, we have
\begin{eqnarray}
   \mathfrak c_\rho=\frac12||V^\rho - \tilde V^\rho||_F^2 &\leq& \frac12||V^\rho - \tilde V^\sigma||_F^2 \nonumber \\
    &\leq& \frac12 (||V_\rho - V_\sigma||_F + ||V_\sigma - \tilde V_\sigma||_F)^2 \\ \nonumber
    &=& \frac{1}{2} ||V_\rho - V_\sigma||_F^2 + \frac{1}{2} ||V_\sigma - \tilde V_\sigma||_F^2 + ||V_\rho - V_\sigma||_F ||V_\sigma - \tilde V_\sigma||_F  \nonumber \\
    &=& \mathfrak c_\sigma + \frac12 ||V_\rho - V_\sigma||_F^2 + \sqrt{2} \sqrt{\mathfrak c_\sigma}||V_\rho - V_\sigma||_F,
\end{eqnarray}
so that, using Eq.~\ref{appeq:perturbation_frobenius},
\begin{equation}
        \mathfrak c_\rho - \mathfrak c_\sigma \leq 800E^2 \epsilon + 40\sqrt 2 E \max(\sqrt{\mathfrak c_\rho},\sqrt{\mathfrak c_\sigma}) \sqrt{\epsilon}.
\end{equation}
Similarly,
\begin{eqnarray}
    \mathfrak c_\sigma - \mathfrak c_\rho &\leq& 800E^2 \epsilon + 40\sqrt 2 E \max(\sqrt{\mathfrak c_\rho},\sqrt{\mathfrak c_\sigma}) \sqrt{\epsilon}.
\end{eqnarray}
Combining these two results,
\begin{equation}
    |\mathfrak c_\rho - \mathfrak c_\sigma| \leq 800E^2 \epsilon + 40\sqrt 2 E \max(\sqrt{\mathfrak c_\rho},\sqrt{\mathfrak c_\sigma}) \sqrt{\epsilon}.
\end{equation}
Now, since
\begin{equation}
    \Tr[\rho \hat E]^2 \leq \Tr[\rho \hat E^2] \leq E^2,
\end{equation}
and
\begin{equation}
    \Tr[\hat E \rho] = \frac{\Tr[V^\rho]}{4} + \frac{\bm d \bm d^T}{4} \leq E,
\end{equation}
we have
\begin{equation}
    \Tr[V^\rho] \leq 4E,
\end{equation}
since $\bm d \bm d^T \geq 0$. Similarly,
\begin{equation}
    \Tr[V^\sigma] \leq 4E.
\end{equation}
Therefore, from Theorem \ref{theo:max_sc},
\begin{equation}
    \mathfrak c_\rho, \mathfrak c_\sigma \leq \mathfrak c_{\mathrm{max}}(4E,m),
\end{equation}
where $\mathfrak c_{\mathrm{max}}(4E,m) = \frac{(4E-2m)^2}{4} + (4E-2m)$ is given by Eq.~\ref{appeq:max_sc}. This gives
\begin{equation}
    |\mathfrak c_\rho - \mathfrak c_\sigma| \leq 800E^2 \epsilon + 40\sqrt 2 E \sqrt{(2E-m)^2 + 2(2E-m)} \sqrt{\epsilon}
\end{equation}
\end{proof}

\subsection{Symplectic coherence under photon loss}\label{app:sc_under_loss}
\noindent In this sction, we show the following result on the evolution of the covariance matrix of a quantum state under photon loss:
\begin{lem}\label{lem:loss_energy_sc}
   Under pure photon loss $\Lambda_\eta(.)$ the covariance matrix of a quantum state $V^\rho$ evolves as
   \begin{equation}
       V^{\Lambda_\eta(\rho)} = \eta V^\rho + (1-\eta)\I_{2m}.
   \end{equation}
   Therefore, the symplectic coherence of a quantum state $\mathfrak c_\rho$ transforms as 
   \begin{equation}
       \mathfrak c_{\Lambda_\eta(\rho)} = \eta^2 \mathfrak c_\rho.
   \end{equation}
\end{lem}
\begin{proof}
    Denoting the covariance matrix of a bipartite quantum state $\sigma_{AB}$ as
    \begin{equation}
        \begin{bmatrix}
            V_x^{\sigma_{AB}} && V_{xp}^{\sigma_{AB}} \\
            (V_{xp}^{\sigma_{AB}})^T && V_p^{\sigma_{AB}}
        \end{bmatrix},
    \end{equation}
    with
    \begin{equation}
        V_x^{\sigma_{AB}} = \begin{bmatrix}
            V_x^{A} && V_{x}^{AB} \\
            V_{x}^{AB} && V_{x}^{B}
        \end{bmatrix},
    \end{equation}
    and similarly for $V_p$ and $V_{xp}$. Then, the covariance matrix of $\sigma_A = \Tr_B[\sigma_{AB}]$ is given by
    \begin{equation}
        \begin{bmatrix}
            V_x^{{A}} && V_{xp}^{A} \\
            (V_{xp}^{A})^T && V_p^{A}
        \end{bmatrix}.
        \end{equation}
    Now, the state obtained by the application of a multimode photon loss channel of a state $\rho$ is given by
    \begin{equation}
        \Lambda_\eta(\rho) = \Tr_B[U_{BS}(\rho \otimes \ket 0 \bra 0^{\otimes m}) U_{BS}^\dagger ],
    \end{equation}
    that is, the first mode of the state is coupled to the first environmental mode (subsystem $B$) through a beam-splitter of efficiency $\eta$, the second mode of the state is coupled to the second environmental mode and so on.
    
    \noindent Firstly, to analyze $\Tr[V^{\Lambda_\eta(\rho)}]$, we find $\forall i \in \{1,\dots,m\}$,
    \begin{eqnarray}
        (V_x^{\Lambda_\eta(\rho)})_{i,i} &=& \Tr[U_{BS}^\dagger \hat q_i^2 U_{BS} \rho \otimes \ket 0 \bra 0^{m-1}] - \Tr[U_{BS}^\dagger \hat q_i U_{BS} \rho \otimes \ket 0 \bra 0^{m-1}]^2 \nonumber \\
        &=& \Tr[(\sqrt\eta \hat q_i + \sqrt{1-\eta} \hat q_i^B)^2 \rho \otimes \ket 0 \bra 0^{\otimes m}] - \Tr[(\sqrt\eta \hat q_i + \sqrt{1-\eta} \hat q_i^B)\rho]^2 \nonumber \\
        &=& \eta(\Tr[\hat q_i^2 \rho]- \Tr[\hat q_i \rho]^2) + (1-\eta) \Tr[(\hat q_i^B)^2 \ket 0\bra{0}^{\otimes m}] \nonumber \\
        &=& \eta (V_x^\rho)_{i,i} + (1-\eta).
    \end{eqnarray}
    For $i,j \in \{1,\dots,m\}$ and $i \neq j$,
        \begin{eqnarray}
        (V_{x}^{\Lambda_\eta(\rho)})_{i,j} &=& \frac12\Tr[U_{BS}^\dagger \{\hat q_i,\hat q_j\} U_{BS} \rho \otimes \ket 0 \bra 0^{\otimes m} ] - \Tr[U_{BS}^\dagger \hat q_i U_{BS} \rho \otimes \ket 0 \bra 0^{\otimes m}]  \Tr[U_{BS}^\dagger \hat q_j U_{BS} \rho \otimes \ket 0 \bra 0^{\otimes m}] \nonumber \\
        &=& \frac12\Tr[\{\sqrt \eta\hat  q_i + \sqrt{1-\eta} \hat q_i^B,\sqrt \eta \hat q_j + \sqrt{1-\eta} \hat q_j^B\} \rho \otimes \ket 0 \bra 0^{\otimes m} ] \nonumber \\ && - \Tr[(\sqrt \eta\hat  q_i + \sqrt{1-\eta} \hat q_i^B)\rho \otimes \ket 0 \bra 0^{\otimes m}]  \Tr[(\hat q_j + \sqrt{1-\eta} \hat q_j^B) \rho \otimes \ket 0 \bra 0^{\otimes m}] \nonumber \\
        &=& \eta \left(\frac12 \Tr[\{\hat q_i,\hat q_j\} \rho] - \Tr[\hat q_i \rho] \Tr[\hat q_j \rho]\right) = \eta (V_{x}^{\rho})_{ij}.
    \end{eqnarray}
    A similar result holds for $V_p^{\Lambda_\eta(\rho)}$. Finally, to analyze the symplectic coherence of $\Lambda_\eta(\rho)$, we consider the elements of $V_{xp}^{\Lambda_\eta(\rho)}$:
    \begin{eqnarray}
        (V_{xp}^{\Lambda_\eta(\rho)})_{i,j} &=& \frac12\Tr[U_{BS}^\dagger \{\hat q_i,\hat p_j\} U_{BS} \rho \otimes \ket 0 \bra 0^{\otimes m} ] - \Tr[U_{BS}^\dagger \hat q_i U_{BS} \rho \otimes \ket 0 \bra 0^{\otimes m}]  \Tr[U_{BS}^\dagger \hat p_j U_{BS} \rho \otimes \ket 0 \bra 0^{\otimes m}] \nonumber \\
        &=& \frac12\Tr[\{\sqrt \eta\hat  q_i + \sqrt{1-\eta} \hat q_i^B,\sqrt \eta \hat p_j + \sqrt{1-\eta} \hat p_j^B\} \rho \otimes \ket 0 \bra 0^{\otimes m} ] \nonumber \\ && - \Tr[(\sqrt \eta\hat  q_i + \sqrt{1-\eta} \hat q_i^B)\rho \otimes \ket 0 \bra 0^{\otimes m}]  \Tr[(\hat p_j + \sqrt{1-\eta} \hat p_j^B) \rho \otimes \ket 0 \bra 0^{\otimes m}] \nonumber \\
        &=& \eta \left(\frac12 \Tr[\{\hat q_i,\hat p_j\} \rho] - \Tr[\hat q_i \rho] \Tr[\hat p_j \rho]\right) = \eta (V_{xp}^{\rho})_{ij}.
    \end{eqnarray}
    Combining these results for $V_x^{\Lambda_\eta(\rho)}, V_p^{\Lambda_\eta(\rho)}$ and $V_{xp}^{\Lambda_\eta(\rho)}$, we get
    \begin{equation}
        V^{\rho} = \eta V^\rho + (1-\eta) \I_{2m},
    \end{equation}
    and therefore,
    \begin{equation}
        \mathfrak c_{\Lambda_\eta(\rho)} = \sum_{i,j=1}^m (V_{xp}^{\Lambda_\eta(\rho)})_{i,j}^2 = \eta^2 \sum_{i,j=1}^m (V_{xp}^\rho)_{i,j}^2 = \eta^2 \mathfrak c_\rho.
    \end{equation}
\end{proof}

\subsection{Position-momentum correlations enhances the average entanglement of pure Gaussian states}\label{app:sc_typical_entangle}
\noindent Before stating the result, we give the define micro-canonical ensembles of pure Gaussian states, following \cite[Section 3.1]{Serafini2007}:
\begin{defi}[Micro-canonical ensembles of pure Gaussian states \cite{Serafini2007}]\label{appdefi:micro_Gaussian}
    Given that the covariance matrix of pure Gaussian states can be written as (see \cref{subsec:prelims_CVQI})
    \begin{equation}\label{eq:covariance_pure_supp}
        V = S_U \begin{bmatrix}
            D && 0 \\
            0 && D^{-1}
        \end{bmatrix} S_U^T,
    \end{equation}
    where $S_U$ is the symplectic matrix corresponding to a passive linear unitary matrix, we define the micro-canonical measure for pure Gaussian states by sampling from the Haar measure over the unitary group $U(m)$ and writing the $2m \times 2m$ symplectic orthogonal matrix isomorphic to it. This corresponds to the symplectic matrix of the passive linear unitary matrix $S_U$. The elements of the diagonal matrix are sampled from a distribution $\mathcal D$, corresponding to uniformly sampled points according to the following constraints:
    \begin{eqnarray}\label{eq:energy_const_micro}
        \sum_{i=1}^m d_i + \frac{1}{d_i} &=& E \nonumber \\
        d_i&\geq&1, \forall i \in \{1,\dots,m\}.
    \end{eqnarray}
\end{defi}
\noindent Note that sampling from the distribution $\mathcal D$ can be achieved by sampling points uniformly on the $m$-dimensional unit sphere and making the appropriate change of variables. As explained in \cite[Section 7]{Serafini2007entanglement} For states with zero position-momentum correlations, we write 
\begin{equation}\label{appeq:zero_symp_coherence_ortho}
    S_U = \begin{bmatrix}
        O && 0 \\ 0 && O
    \end{bmatrix},
\end{equation}
with $O$ is an $m \times m$ orthogonal matrix sampled from the haar measure over the orthogonal group (denoted by $\mu_H(O)$). Whereas for states with non-zero position-momentum correlations, we write
\begin{equation}\label{appeq:zero_symp_coherence_unitary}
    S_U = \begin{bmatrix}
        X && Y \\ -Y && X
    \end{bmatrix},
\end{equation}
which is isomorphic to an $m \times m$ unitary matrix $U = X+iY$ sampled from the haar measure over the unitary group (denoted by $\mu_H(U)$). Then we have the following Theorem:
\begin{theo}[Position-momentum correlations enhance entanglement]\label{apptheo:typical_entangle}
    For micro-canonical ensemble of pure Gaussian states of energy $E$ (Definition \ref{appdefi:micro_Gaussian}), the typical entanglement of first mode with respect to the other modes, defined as expectation value over the distribution $\mathcal D$ and $\mu_H(O)$ (for states with zero position-momentum correlations)/ $\mu_H(U)$ (for states with non-zero position-momentum correlations), of $\nu_1^2$, where $\nu_1$ is the reduced covariance matrix of the first mode, is higher for states with non-zero position-momentum correlations than states with zero position-momentum correlations.
\end{theo}
\noindent Before proving the Theorem, we want to give the following result about the elements of Haar-random unitary and orthogonal matrices:
\begin{lem}\label{lem:Haar_measure_prop}
    Given an $m \times m$ orthogonal matrix $O$ sampled from the Haar measure over the orthogonal group, $\mu_H(O)$, we have
    \begin{equation}
    \E_{\mu_H(O)} [O_{1i}^2 O_{1j}^2] = \begin{cases}
        \frac{3}{m(m+2)}, i = j, \\
    \frac{1}{m(m+2)}, i \neq j.
    \end{cases} 
\end{equation}
Similarly, given an $m\times m$ unitary matrix $U = X+iY$ sampled from the Haar measure over the unitary group, $\mu_H(U)$, we have
\begin{eqnarray}
    \E_{\mu_H(U)}[X_{1i}^2Y_{1j}^2] &=& \frac{1}{4m(m+1)}, \nonumber \\
    \E_{\mu_H(U)}[X_{1i}^2X_{1j}^2]= \E_{\mu_H(U)}[Y_{1i}^2Y_{1j}^2]&=& \begin{cases}
        \frac{3}{4m(m+1)}, \text{ if }i = j, \\
        \frac{1}{4m(m+1)}, \text{otherwise},
    \end{cases}  \nonumber \\
    \E_{\mu_H(U)}[X_{1i}Y_{1i}X_{1j}Y_{1j}] &=& \begin{cases}
        \frac{1}{4m(m+1)},\text{ if }i=j, \\
        0, \text{otherwise}.
    \end{cases}
\end{eqnarray}
\end{lem}
\begin{proof}
    We are focusing on the elements of the first row of the orthogonal matrix $O$. Since the orthogonality of $O$ implies
    \begin{equation}
        \sum_{i=1}^m O_{1i}^2 = 1,
    \end{equation}
    therefore, the vector $O_1$ is sampled from the points on the unit sphere of dimension $m$. Since the Haar measure over the orthogonal group is left and right invariant under the multiplication by any orthogonal matrix $V$ \cite{Bourbaki2004}, this implies that the distribution over $O_1 V$ is equal to that of $O_1$. This orthogonal invariance implies that $O_1$ is sampled from the uniform distribution on the $m$-dimensional unit sphere.

    One way to model the uniform distribution on the $m$-dimensional unit sphere is to take $m$ points $(p_1,\dots,p_m)$ which are distributed according to an $m$-variate normal distribution with zero mean and covariance matrix equal to identity. Then, the set of points
    \begin{equation}
        \left(\frac{p_1}{\sqrt{p_1^2 + \dots + p_m^2}},\dots,\frac{p_m}{\sqrt{p_1^2 + \dots + p_m^2}}\right)
    \end{equation}
    are uniformly distributed on the unit sphere \cite{Muller1959}. Further, for $m$-variate normal distribution with zero mean and covariance matrix, the expectation value of
    \begin{equation}
        \Pi_{i=1}^m\left( {p_i}^{k_i}/\left(\sqrt{p_1^2 + \dots + p_m^2}\right)^{k_i}\right),
    \end{equation}
    with $k_i \in \N$, $\forall i \in \{ 1,\dots,m\}$, is given by
    \begin{equation}
        \frac{1}{(2\pi)^{m/2}}\int \Pi_{i=1}^m\left( {p_i}^{k_i}/\left(\sqrt{p_1^2 + \dots + p_m^2}\right)^{k_i}\right)e^{-(p_1^2 + \dots + p_m^2)/2} dp_1\dots dp_m
    \end{equation}
    In the polar co-ordinate representation, this integral is equal to
    \begin{equation}\label{appeq:typical_entangle_normal_expectation}
        \frac{1}{(2\pi)^{m/2}}\int_{0}^\infty r^{m-1}e^{-r^2/2} dr \int_{\mathbb{S}^{m-1}}(\Pi_{i=1}^m \omega_i^{k_i}) \Omega(d\omega),
    \end{equation}
    where $r = \sqrt{p_1^2 + \dots + p_m^2}$, $\omega_i = p_i/\sqrt{p_1^2 + \dots + p_m^2} \forall i \in \{1,\dots,m\}$, $\mathbb{S}^{m-1}$ is unit sphere in $\R^m$ and $\Omega$ is the surface measure on $\mathbb{S}^{m-1}$. From \cite{Folland2001},
    \begin{equation}
        \int_{\mathbb{S}^{m-1}}(\Pi_{i=1}^m \omega_i^{k_i}) \Omega(d\omega) = \begin{cases}\label{appeq:polar_coor_integral}
            0, \text{ if any of the } k_i \text{ is odd} \\
            \frac{2 \Gamma(\beta_1)\dots\Gamma(\beta_m)}{\Gamma(\beta_1 + \dots + \beta_m)}, \text{ if all }k_i\text{ are even and } \beta_i = \frac12 (k_i+1)
        \end{cases}
    \end{equation}
    Then, by taking
    \begin{equation}
        O_{1i} = \frac{p_i}{\sqrt{p_1^2 + \dots + p_m^2}} \forall i
    \end{equation}
    where $p_i$'s are distributed according to an $m$-mode normal distribution with mean $0$ and covariance matrix equal to identity and using \cref{appeq:typical_entangle_normal_expectation} and \cref{appeq:polar_coor_integral}, we obtain
    \begin{equation}
    \E_{\mu_H(O)} [O_{1i}^2 O_{1j}^2] = \begin{cases}
        \frac{3}{m(m+2)}, i = j, \\
    \frac{1}{m(m+2)}, i \neq j.
    \end{cases} 
\end{equation}
\noindent Now for the first row of the Haar-randomly sampled unitary $U = X + iY$, denoted by $U_1$, we have
\begin{equation}
    \sum_{i=1}^m X_{1i}^2 + Y_{1i}^2 = 1.
\end{equation}
Therefore, the distribution of the complex vector $U_1$ is the distribution on the surface of $m$-dimensional complex unit sphere. The unitary invariance \cite{Mele2024Haar} further implies that the points are uniformly distributed on the complex unit sphere. Using the bijective mapping
\begin{equation}
    \begin{bmatrix}
        U_{11} \\ \vdots\\ U_{1m}
    \end{bmatrix} \rightarrow \begin{bmatrix}
        X_{11} \\ \vdots \\ X_{1m} \\ Y_{11} \\ \vdots \\ Y_{1m}
    \end{bmatrix},
\end{equation}
the uniform distribution of $U_1$ on the $m$-dimensional complex unit sphere implies that the distribution of $X_{11},\dots,X_{1m},Y_{11},\dots,Y_{1m}$ as a uniform distribution over points on the $2m$-dimensional real unit sphere. Therefore, we can use the results from the previous discussion on Haar-random orthogonal matrices to obtain the required results on the Haar-random unitary matrix:
\begin{eqnarray}
        \E_{\mu_H(U)}[X_{1i}^2Y_{1j}^2] &=& \frac{1}{2m(2m+2)} = \frac{1}{4m(m+1)}, \nonumber \\
    \E_{\mu_H(U)}[X_{1i}^2X_{1j}^2]= \E_{\mu_H(U)}[Y_{1i}^2Y_{1j}^2]&=& \begin{cases}
       \frac{3}{2m(2m+2)} = \frac{3}{4m(m+1)}, \text{ if }i = j, \\
       \frac{1}{2m(2m+2)} = \frac{1}{4m(m+1)}, \text{otherwise},
    \end{cases}  \nonumber \\
    \E_{\mu_H(U)}[X_{1i}Y_{1i}X_{1j}Y_{1j}] &=& \begin{cases}
       \frac{1}{2m(2m+2)} = \frac{1}{4m(m+1)},\text{ if }i=j, \\
        0, \text{otherwise}.
    \end{cases}
\end{eqnarray}
\end{proof}
\noindent Now we are ready to prove the result.
\begin{proof}
\noindent \textbf{States with zero position-momentum correlations:} From the previous discussion (\cref{appeq:zero_symp_coherence_ortho}), covariance matrices of states with zero position-momentum correlations are of the form
\begin{equation}
    V = \begin{bmatrix}
        O && 0 \\
        0 && O
    \end{bmatrix} \begin{bmatrix}
        D && 0 \\ 0 && D^{-1}
    \end{bmatrix} \begin{bmatrix}
        O^T && 0 \\
        0 && O^T
    \end{bmatrix},
\end{equation}
where $O$ is an $m \times m$ orthogonal matrix sampled Haar-randomly. This guarantees each randomly sampled covariance matrix to be of zero symplectic coherence, with covariance matrix
\begin{equation}
    V = \begin{bmatrix}
        O D O^T && 0 \\
        0 && O D^{-1} O^T
    \end{bmatrix}.
\end{equation}
The symplectic eigenvalue of the first mode is then given by
\begin{eqnarray}
    \nu_{1O}^2 = (O D O^T)_{11} (OD^{-1}O^T)_{11} = \sum_{i,j=1}^m \frac{d_i}{d_j} O_{1i}^2 O_{1j}^2
\end{eqnarray}
Taking the expectation value over Haar measure on the orthogonal group and using the results from Lemma \ref{lem:Haar_measure_prop}:
    \begin{equation}
    \E_{\mu_H(O)} [O_{1i}^2 O_{1j}^2] = \begin{cases}
        \frac{3}{m(m+2)}, i = j, \\
    \frac{1}{m(m+2)}, i \neq j.
    \end{cases} 
\end{equation}
we get
\begin{equation}
  \E_{O \leftarrow \mu_H(O)}[ \nu_{1O}^2] = \frac{3}{m+2} + \frac{1}{m(m+2)} \sum_{i\neq j} \frac{d_i}{d_j} =  \frac{3}{m+2} + \frac{1}{2m(m+2)} \sum_{i\neq j} \frac{d_i}{d_j} + \frac{d_j}{d_i}.
\end{equation}
Finally, taking the average over $D\leftarrow\mathcal D$,
\begin{equation}
    (\nu_{1O}^{\mathrm{exp}})^2 =  \E_{O \leftarrow \mu_H(O),D\leftarrow\mathcal D} [\nu_{1O}^2] = \frac{3}{m+2} + \frac{1}{2m(m+2)} S_1,
\end{equation}
where 
\begin{equation}\label{eq:entangle_S1}
    S_1 \coloneqq \E_{D\leftarrow\mathcal D}\left[ \sum_{i\neq j} \frac{d_i}{d_j} + \frac{d_j}{d_i}\right]
\end{equation}
\textbf{States with non-zero position-momentum correlations:} From our previous discussion (\cref{appeq:zero_symp_coherence_unitary}), covariance matrices of states with non-zero position-momentum correlations can be written as
\begin{equation}
    V = \begin{bmatrix}
        X && Y \\
        -Y && X
    \end{bmatrix}  \begin{bmatrix}
        D && 0 \\ 0 && D^{-1}
    \end{bmatrix} \begin{bmatrix}
        X^T && -Y^T \\
        Y^T && X^T
    \end{bmatrix}.
\end{equation}
The final covariance matrix in this case looks like
\begin{equation}
        V = \begin{bmatrix}
        X D X^T + Y D^{-1}Y^T && -X D Y^T + Y D^{-1} X^T \\
        - Y D X^T + X D^{-1} Y^T && Y D Y^T + X D^{-1}X^T
    \end{bmatrix}.
\end{equation}
Therefore, the symplectic eigenvalue of the first mode is
\begin{eqnarray}
    \nu_{1U}^2 &=& (X D X^T + Y D^{-1}Y^T)_{11} (Y D Y^T + X D^{-1}X^T)_{11} - (-X D Y^T + Y D^{-1} X^T)_{11}^2 \nonumber \\
    &=& \sum_{i,j=1}^md_i d_j X_{1i}^2 Y_{1j}^2 + \frac{d_i}{d_j} X_{1i}^2 X_{1j}^2 + \frac{d_j}{d_i} Y_{1i}^2 Y_{1j}^2 + \frac{1}{d_i d_j} X_{1i}^2 Y_{1j}^2 - \sum_{i,j=1}^m \left(d_i - \frac{1}{d_i}\right) \left(d_j - \frac{1}{d_j}\right) X_{1i}Y_{1i}X_{1j}Y_{1j}.\\
\end{eqnarray}
Taking the expectation value over the Haar measure on the unitary group and using Lemma \ref{lem:Haar_measure_prop}:
\begin{eqnarray}
        \E_{\mu_H(U)}[X_{1i}^2Y_{1j}^2] &=& \frac{1}{2m(2m+2)} = \frac{1}{4m(m+1)}, \nonumber \\
    \E_{\mu_H(U)}[X_{1i}^2X_{1j}^2]= \E_{\mu_H(U)}[Y_{1i}^2Y_{1j}^2]&=& \begin{cases}
       \frac{3}{2m(2m+2)} = \frac{3}{4m(m+1)}, \text{ if }i = j, \\
       \frac{1}{2m(2m+2)} = \frac{1}{4m(m+1)}, \text{otherwise},
    \end{cases}  \nonumber \\
    \E_{\mu_H(U)}[X_{1i}Y_{1i}X_{1j}Y_{1j}] &=& \begin{cases}
       \frac{1}{2m(2m+2)} = \frac{1}{4m(m+1)},\text{ if }i=j, \\
        0, \text{otherwise},
    \end{cases}
\end{eqnarray}
we get
\begin{eqnarray}
    \E_{U \leftarrow \mu_H(U)}[\nu_{1U}^2] &=& \frac{3}{2(m+1)} + \frac{1}{4m(m+1)} \sum_{i\neq j}\frac{d_i}{d_j} + \frac{d_j}{d_i} + \frac{1}{4m(m+1)}  \sum_{i,j=1}^m d_i d_j + \frac{1}{d_id_j} -  \frac{1}{4m(m+1)} \sum_{i=1}^m \left(d_i^2 + \frac{1}{d_i^2} - 2\right) \nonumber \\
    &=& \frac{2}{m+1} + \frac{1}{4m(m+1)} \sum_{i\neq j}\left(\frac{d_i}{d_j} + \frac{d_j}{d_i}\right) + \frac{1}{4m(m+1)} \sum_{i\neq j}d_i d_j + \frac{1}{d_id_j}. 
\end{eqnarray}
Finally, taking the expectation value over the distribution on $D$,
\begin{equation}
    (\nu_{1U}^{\mathrm{exp}})^2 = \E_{U \leftarrow \mu_H(U),D\leftarrow\mathcal D}[\nu_{1U}^2] = \frac{2}{m+1} + \frac{S_1}{4m(m+1)} + \frac{S_2}{4m(m+1)}
\end{equation}
where $S_1$ is the same as before (both states with zero and non-zero position correlations have the same distribution over $D$) and $S_2$ is denoted by
\begin{equation}
    S_2 \coloneqq \E_{D\leftarrow\mathcal D}\left[\sum_{i\neq j}d_i d_j + \frac{1}{d_id_j}\right].
\end{equation}
Now, since $\forall i,j$ and for all instances of $D$,
\begin{equation}
    d_i d_j +\frac{1}{d_i d_j} - \frac{d_i}{d_j} - \frac{d_j}{d_i} = \left(d_i -\frac{1}{d_i}\right) \left(d_j -\frac{1}{d_j}\right) \geq 0,
\end{equation}
since $d_i \geq 1, \forall i$ for all instances of $D$. Therefore when we take expectation values on both sides,
\begin{equation}
    S_2 \geq S_1.
\end{equation}
Similarly, since
\begin{equation}
    d_i d_j + \frac{1}{d_i d_j} \geq 2,
\end{equation}
$\forall i,j$ and for all instances of $D$, since $d_i \geq 1, \forall i$ for all instances of $D$. Therefore, taking the expectation values on both sides, we get
\begin{equation}
    S_2 \ge 2m(m-1).
\end{equation}
Therefore,
\begin{eqnarray}
    (\nu_{1U}^{\mathrm{exp}})^2 - (\nu_{1O}^{\mathrm{exp}})^2 &=& \frac{2}{m+1} - \frac{3}{m+2} + \left(\frac{1}{4m(m+1)} - \frac{1}{2m(m+2)}\right) S_1 + \frac{S_2}{4m(m+1)} \nonumber \\
    &=& - \frac{m-1}{(m+1)(m+2)} - \frac{S_1}{4(m+1)(m+2)}  + \frac{S_2}{4m(m+1)} \nonumber \\
    &=& \frac{1}{m+1} \left(- \frac{m-1}{m+2} + \frac{(m+2)S_2 - m S_1}{4m(m+2)}\right).
\end{eqnarray}
Using
\begin{eqnarray}
    S_2 - S_1 &\geq& 0, \nonumber \\
    S_2 &\geq& 2m(m-1),
\end{eqnarray}
we get
\begin{eqnarray}
    (\nu_{1U}^{\mathrm{exp}})^2 - (\nu_{1O}^{\mathrm{exp}})^2 &\geq& \frac{1}{m+1} \left(-\frac{m-1}{m+2} + \frac{4m(m-1)}{4m(m+2)}\right) = 0,
\end{eqnarray}
and hence,
\begin{equation}
    (\nu_{1U}^{\mathrm{exp}})^2 \geq (\nu_{1O}^{\mathrm{exp}})^2.
\end{equation}
\end{proof}

\section{Symplectic coherence in quantum information processing}

\subsection{Symplectic coherence enhances quantum metrology}\label{app:sc_metro}
\noindent We restate Theorem \ref{theo:sc_metro}:
\begin{theo}[Position-momentum correlations enhance quantum metrology]
     Given a pure single-mode input probe state $\rho$ with covariance matrix $V^\rho$ and zero first moments going through a single-mode displacement operator $\hat D(r)$, the Quantum Fisher information with respect to $r$ satisfies
    \begin{equation}
        \mathcal F_Q \leq 2E + 4V_{xp}^\rho,
    \end{equation}
    where $E = V_x^\rho + V_p^\rho = 4 \Tr[\rho \hat E]$. Further, the inequality becomes an equality for pure input probe states.
\end{theo}
\begin{proof}
    Given an input probe state $\rho$ and a unitary 
    \begin{equation}
        \hat U(r) = e^{ir\hat H}
    \end{equation}
    the Fisher quantum information for a quantum state with respect to $r$ is given by \cite[Eqs.~60,61]{Toth_2014}
    \begin{eqnarray}
        \mathcal F_Q &\leq& 4(\Tr[\rho \hat H^2] - \Tr[\rho \hat H]^2) \nonumber \\
        &\leq& 2\left(\Tr[\rho \hat q^2] - \Tr[\rho \hat q]^2 + \Tr[\rho \hat p^2] - \Tr[\rho \hat p]^2 + 2 (\Tr\left[\rho \{\hat q, \hat p\}/2\right] - \Tr[\rho \hat q] \Tr[\rho \hat p]) \right) \nonumber \\
        &\leq& 2(V_x^\rho + V_p^\rho + 2 V_{xp}^\rho) = 2E + 4V_{xp}^\rho,
    \end{eqnarray}
    with equality for pure input probe states. Following the discussion in Section \ref{sec:sc_metro} of the main text, we can always ensure $V_{xp}^\rho$ is positive, in which case $V_{xp}^\rho = \sqrt{\mathfrak c_\rho}$.
\end{proof}
Following the discussion in Section of the main text, we can always 
\subsection{Symplectic coherence lower bounds optimal success probability of certain channel discrimination tasks}\label{app:sc_lb_td}
\noindent We restate Lemma \ref{lem:sc_lb_td_new}:
\begin{res}[Symplectic coherence lower bounds optimal channel discrimination success probability]\label{applem:discr}
The optimal success probability of successful channel discrimination between a multi-mode photon loss channel $\Lambda_\eta$ and an arbitrary orthogonal Stinespring dilation $\mathcal O_0$ for a given input probe state $\rho$ with zero first moments is lower bounded in terms of its symplectic coherence $\mathfrak c_\rho$:
\begin{equation}
        p_\mathrm{opt} \geq \frac12 + \frac{1}{400}\min\left(1,\frac{ \sqrt{\frac{(1-\eta)^2 (E-2m)^2}{2m} + 2\mathfrak c_\rho(1-\eta)^2}}{E+ 1} \right),
    \end{equation}
where $E = 4\Tr[\rho \hat E] = \Tr[V^\rho]$.
\end{res}
Before proving the result, we would like to present a technical result relating the trace distance between two quantum states to their energy and symplectic coherence.
\begin{lem}
    Given two states $\rho$ and $\sigma$ with zero first moments, covariance matrices $V^\rho, V^\sigma$ with $\Tr[V^\rho] = 4 \Tr[\rho \hat E] = E_1$, $\Tr[V^\sigma] = 4\Tr[\sigma \hat E] = E_2$ and such that $\Tr[\rho \hat E], \Tr[\sigma \hat E] \leq E/4$. If $\rho$ and $\sigma$ are Gaussian states, then we have
    \begin{equation}
       \frac12 ||\rho - \sigma||_1 \geq \frac{1}{200} \min \left(1,\frac{\sqrt{\frac{(E_1 - E_2)^2}{2m} + 2(\sqrt{\mathfrak c_\rho} - \sqrt{\mathfrak c_\sigma})^2}}{E+1}\right).
    \end{equation}
    If $\rho$ and $\sigma$ are any general (Gaussian or non-Gaussian) states, we have
    \begin{equation}
        \frac12 ||\rho - \sigma||_1 \geq \frac{\frac{(E_1 - E_2)^2}{2m} + 2(\sqrt{\mathfrak c_\rho} - \sqrt{\mathfrak c_\sigma})^2}{3200\tilde E^2m},
    \end{equation}
    where $\tilde E^2$ is such that $\Tr[\hat{E^2}\rho] = \Tr[\hat{E^2}\sigma] \leq \tilde E^2$.
\end{lem}
\begin{proof}
    First, we prove
    \begin{equation}
        ||V^\rho -  V^\sigma||_F \geq \sqrt{\frac{(E_1 - E_2)^2}{2m} + 2(\sqrt{\mathfrak c_\rho} - \sqrt{\mathfrak c_\sigma})^2}.
    \end{equation}
    Let
\begin{equation}
        V^{\rho} = \begin{bmatrix}
            V_x^\rho && V_{xp}^\rho \\
            (V_{xp}^{\rho})^T && V_p^\rho
        \end{bmatrix},
    \end{equation}
    and
      \begin{equation}
        V^{\sigma} = \begin{bmatrix}
            V_x^\sigma && V_{xp}^\sigma \\
            (V_{xp}^{\sigma})^T && V_p^\sigma
        \end{bmatrix}.
    \end{equation}
    We have
    \begin{equation}
        V^{\rho} - V^{\sigma} = \begin{bmatrix}
            V_x^\rho - V_{x}^\sigma && V_{xp}^\rho - V_{xp}^\sigma \\
            (V_{xp}^{\rho})^T - (V_{xp}^{\sigma})^T && V_p^\rho - V_p^{\sigma}
        \end{bmatrix},
    \end{equation}
    therefore,
    \begin{eqnarray}
        ||V^{\rho} - V^{\sigma}||_F^2 &=& \Tr[(V^{\rho} - V^{\sigma})^2] \nonumber \\
        &=& \Tr[(V_x^{\rho} - V_x^\sigma)^2] + \Tr[(V_p^{\rho} - V_p^\sigma)^2] + 2\Tr[(V_{xp}^\rho - V_{xp}^\sigma)(V_{xp}^\rho - V_{xp}^\sigma)^T] \nonumber \\
        &=& ||V_x^{\rho} - V_x^\sigma||_F^2 + ||V_p^{\rho} - V_p^\sigma||_F^2 + 2||V_{xp}^\rho - V_{xp}^\sigma||_F^2.
    \end{eqnarray}
    Now,
    \begin{eqnarray}\label{eq:chain}
        ||V_x^{\rho} - V_x^\sigma||_F^2 + ||V_p^{\rho} - V_p^\sigma||_F^2 &\geq& \frac{||V_x^{\rho} - V_x^\sigma||_1^2 + ||V_p^{\rho} - V_p^\sigma||_1^2}{m} \nonumber \\
        &\geq& \frac{\Tr[V_x^\rho - V_x^\sigma]^2 + \Tr[V_p^\rho - V_p^\sigma]^2}{m}  \nonumber \\
        &\geq& \frac{(\Tr[V_x^\rho - V_x^\sigma] + \Tr[V_p^\rho - V_p^\sigma])^2}{2m} =\frac{(E_1 - E_2)^2}{2m}.
    \end{eqnarray}
    Here the first inequality $||A||_F^2 \geq ||A||_1^2/m$ follows from Eq.~\ref{eq:FM_TM}, the second inequality $||A||_1^2 \geq \Tr[A]^2$ follows from
    \begin{equation}
        \big|\Tr[A]\big| = \left|\sum_{i=1}^m \lambda_i(A)\right| \leq \sum_{i=1}^m |\lambda_i(A)| \leq \sum_{i=1}^m s_i(A) = ||A||_1,
    \end{equation}
    and the final inequality $\Tr[V_x^\rho - V_x^\sigma]^2 + \Tr[V_p^\rho - V_p^\sigma]^2 \geq (\Tr[V_x^\rho - V_x^\sigma] + \Tr[V_p^\rho - V_p^\sigma])^2/2$ follows from the Cauchy--Schwarz inequality.
    
    In addition, by the inverse triangle inequality
    \begin{equation}
        ||V_{xp}^\rho - V_{xp}^\sigma||_F \geq \left| ||V_{xp}^\rho||_F - ||V_{xp}^\sigma||_F \right| = \left|\sqrt{\mathfrak{c}_\rho} - \sqrt{\mathfrak{c}_\sigma}\right|,
    \end{equation}
    and thus
    \begin{equation}\label{eq:sc_lb_frob}
        ||V^\rho -  V^\sigma||_F \geq \sqrt{\frac{(E_1 - E_2)^2}{2m} + 2(\sqrt{\mathfrak c_\rho} - \sqrt{\mathfrak c_\sigma})^2}.
    \end{equation}
By \cite[Theorem 11]{mele2024learning}, for Gaussian states with $\Tr[\rho \hat E] \leq E/4$,
\begin{equation}
    \frac12 ||\rho - \sigma||_1 \geq \frac{1}{200}\min \left(1,\frac{||V^\rho - V^\sigma||_F}{E+1}\right).
\end{equation}
From Eq.~\ref{eq:sc_lb_frob}, we get
\begin{equation}
   \frac12 ||\rho - \sigma||_1 = \frac{1}{200} \min \left(1,\frac{\sqrt{\frac{(E_1 - E_2)^2}{2m} + 2(\sqrt{\mathfrak c_\rho} - \sqrt{\mathfrak c_\sigma})^2}}{E+1}\right).
\end{equation}
Further from \cite[Theorem 62]{annamele2025symplecticrank}, for two arbitrary quantum states $\rho$ and $\sigma$, we have
\begin{equation}
    \frac12 ||\rho - \sigma||_1\geq \frac{||V^\rho - V^\sigma||_\infty^2}{3200\tilde E^2},
\end{equation}
provided $\Tr[\hat E^2\rho], \Tr[\hat E^2 \sigma] \leq \tilde E^2$. Given $||V^\rho - V^\sigma||_\infty \geq \frac{||V^\rho - V^\sigma||_F}{\sqrt{m}}$ and Eq.~\ref{eq:sc_lb_frob}, we get
\begin{equation}
   \frac12 ||\rho - \sigma||_1 \geq \frac{\frac{(E_1 - E_2)^2}{2m} + 2(\sqrt{\mathfrak c_\rho} - \sqrt{\mathfrak c_\sigma})^2}{3200 \tilde E^2m}.
\end{equation}
\end{proof}
\noindent Now we are ready to prove the result.
\begin{proof}[Proof of \cref{applem:discr}]
     The optimal success probability of successfully distinguishing two states $\mathcal O_0(\rho)$ and $\Lambda_\eta(\rho)$ is given by the Helstrom bound \cite{Helstrom1969}:
    \begin{equation}
        p_\mathrm{opt} = \frac{1}{2} + \frac14||\mathcal O_0 (\rho) - \Lambda_\eta(\rho)||_1.
    \end{equation}
   Under uniform multimode photon loss, the covariance matrix of a quantum state $V^\rho$ transforms as (see \cref{app:sc_under_loss})
    \begin{equation}
        V^\rho \rightarrow \eta V^\rho + (1-\eta)\I_{2m},
    \end{equation}
    hence
    \begin{equation}
        \Tr[V^{\Lambda_\eta(\rho)}] = \eta E + (1-\eta)2m \leq E.
    \end{equation}
    As a consequence,
    \begin{equation}
        \Tr[\hat E \Lambda_\eta(\rho)] = \frac{\Tr[V^{\Lambda_\eta(\rho)}]}{4} \leq \frac{E}{4}. 
    \end{equation}
    From Lemma \ref{lem:sc_lb_td} for Gaussian states,
    \begin{equation}
        p_\mathrm{opt} \geq \frac12 + \frac{1}{400}\min\left(1,\frac{\sqrt{\frac{(E_{\mathcal{O}_0(\rho)} - E_{\Lambda_\eta(\rho)})^2}{2m} + 2(\sqrt{\mathfrak c_{\mathcal O_0(\rho)}} - \sqrt{\mathfrak c_{\Lambda_\eta(\rho)}})^2}}{E+1}\right).
    \end{equation}
    Now, for block-diagonal orthogonal operations
    \begin{equation}
        E(\mathcal O_0(\rho)) = E, \hspace{5mm}\mathfrak c_{\mathcal O_0(\rho)} = \mathfrak c_\rho.
    \end{equation}
    Further, using $V^\rho \rightarrow \eta V^\rho + (1-\eta)\I_{2m}$,
    \begin{equation}
        E(\Lambda_\eta(\rho)) = \eta E + 2m(1-\eta), \hspace{5mm} \mathfrak c_ {\Lambda_\eta(\rho)} = \eta^2 \mathfrak c_\rho.
    \end{equation}
    Therefore,
    \begin{equation}
        (E(\mathcal O_0(\rho)) -  E(\Lambda_\eta(\rho)))^2 = (1-\eta)^2(E-2m)^2,
    \end{equation}
    and
    \begin{equation}
        \left(\sqrt {c_{\mathcal O_0(\rho)}} - \sqrt {c_{\Lambda_\eta(\rho)}}\right)^2 = \mathfrak c_\rho (1-\eta)^2.
    \end{equation}
    Hence, we obtain
    \begin{equation}
        p_\mathrm{opt} \geq \frac12 + \frac{1}{400}\min\left(1,\frac{ \sqrt{\frac{(1-\eta)^2 (E-2m)^2}{2m} + 2\mathfrak c_\rho(1-\eta)^2}}{E+1} \right).
    \end{equation}
\end{proof}

\subsection{Symplectic coherence underlines the efficiency of certain channel discrimination tasks}\label{app:eff_ortho_discrimination}
\noindent For completeness, we define Orthogonal stinespring dilation here:
\begin{defi}[Orthogonal Stinespring dilation]\label{appdefi:ortho_Stinespring}
    Given a quantum state $\rho$, an orthogonal Stinespring dilation of $\rho$ is defined as
    \begin{equation}
        \Tr_B[\hat D \hat{O} (\rho \otimes \sigma^{(E)}) \hat O^\dagger \hat D^\dagger],
    \end{equation}
    where $\hat O$ are block-diagonal orthogonal gates, $\hat D$ is the displacement operator, and $\sigma \in \mathcal C$.
\end{defi}
\noindent Then we have the following Lemma:
\begin{lem} [Symplectic coherence underlines the efficiency of multi-shot discrimination protocols]
    Given a Gaussian input state $\rho$ with $\Tr[V^\rho] = E$ and symplectic coherence $\mathfrak c_\rho$, two orthogonal Stinespring dilation channels $\mathcal O_1$ and $\mathcal O_2$ satisfying 
$\Tr[\mathcal O_i(\rho) \{\hat q_1, \hat p_1\}/2] = \mu_i$ for $i=1,2$ can be discriminated using the measurement $\hat M = \{\hat q_1, \hat p_1\}/2$ 
with success probability $1 - \delta$ as long as the number of samples satisfies
\begin{equation}
    N \geq N_{\mathrm{thres}},
\end{equation}
where
\begin{equation}
    N_{\mathrm{thres}} \leq 
    \frac{272 \log(2/\delta)}{(\mu_2 - \mu_1)^2} 
    \max \big( f(m,E) + \mu_1^2,\, f(m,E) + \mu_2^2 \big),
\end{equation}
where
\begin{equation}
    f(m,E) \coloneqq 1 + \frac14{(E-2(m-1))^2}.
\end{equation}
Moreover, since for given energy $E$ and symplectic coherence $\mathfrak c_\rho$, $|\mu_1|,\, |\mu_2| \leq \sqrt{\mathfrak c_\rho}$, the optimal upper bound (its lowest possible value) is given by
\begin{equation}
272 \log (2/\delta) 
    \left( \frac{f(m,E)}{4 \mathfrak c_\rho} + \frac14 \right).
\end{equation}
\end{lem}
Before proving the result, we describe the protocols allowing us to discriminate between $\mathcal O_1$ and $\mathcal O_2$ using the output measurement operator $\hat M$ and $N$ copies of the probe state $\rho$:
\noindent To perform the channel discrimination using $N$ samples of the input probe state $\rho$, we define the following protocol.
\begin{protocol}[Channel discrimination protocol]\label{Prot:channel_discrimination}
\hspace{3mm}
\begin{itemize}
    \item For the given input state $\rho$ and orthogonal Stinespring dilations $\mathcal O_1$ and $\mathcal{O}_2$, compute analytically $\Tr[\mathcal O_1(\rho) \{\hat q_1,\hat p_1\}/2] = \mu_1$ and  $\Tr[\mathcal O_2(\rho) \{\hat q_1,\hat p_1\}/2] = \mu_2$. Define $d\coloneqq \mu_2 - \mu_1$ (labelling the channels such that $\mu_2 >  \mu_1$) and $T\coloneqq \frac{\mu_2 + \mu_1}{2} = \mu_1 + d/2 = \mu_2 - d/2$.
    \item For each of the $N$ input copies of $\rho$, make the measurement $\hat M = \{\hat{q}_1,\hat p_1\}/2$ after $\rho$ has passed through the channel to obtain $X_i, \forall i \in \{1,\dots,m\}$.
    \item Calculate $X$ as the median of means of the data set $\{X_1,\dots,X_N\}$ \cite{JERRUM1986,Vershynin_2018}.
    \item If $X > T$, predict channel 2, otherwise predict channel 1.
\end{itemize}
\end{protocol}
\noindent Then, we prove the efficiency of Protocol \ref{Prot:channel_discrimination} through Lemma \ref{lem:eff_ortho_discrimination}.
\begin{proof}[Proof of \cref{lem:eff_ortho_discrimination}]
    The probability of error using Protocol \ref{Prot:channel_discrimination} is given by
    \begin{eqnarray}
        p_{\mathrm{err}} &=& \frac{1}{2} \Pr[X_1 \geq T] + \frac12 \Pr[X_2 \leq T] \nonumber \\
        &=& \frac{1}{2} \Pr[X_1 \geq \mu_1 + d/2] + \frac12 \Pr[X_2 \leq \mu_2 - d/2],
    \end{eqnarray}
    where $X_1$ and $X_2$ are median of means estimator values when the channel implements channel 1 or channel 2, respectively. Now,
    \begin{equation}
        p_\mathrm{err}\leq  \frac12 \Pr[|X_1 - \mu_1| \geq d/2] + \frac12 \Pr[|X_2 - \mu_2| \geq d/2].
    \end{equation}
    and the median of means estimator satisfies \cite{JERRUM1986}, 
    \begin{equation}
        \Pr[|X_1 - \mu_1| \geq d/2] \leq \delta
    \end{equation}
    as long as the number of samples is greater than
    \begin{equation}
        N_1\coloneqq68 \log\left(\frac{2}{\delta}\right)\frac{\mathrm{Var}(\hat M)}{(d/2)^2} = 272 \log\left(\frac{2}{\delta}\right) \frac{\mathrm{Var}_{\mathcal O_1(\rho)}(\hat M)}{d^2}
    \end{equation}
    Now,
    \begin{eqnarray}
        \mathrm{Var}_{\mathcal O_1(\rho)}(\hat M) &=& \Tr[\hat M^2 \mathcal{O}_1(\rho)] - \Tr[M \mathcal{O}_1(\rho)]^2 = \Tr[\hat M^2 \mathcal{O}_1(\rho)] - \mu_1^2.
    \end{eqnarray}
    Now $\Tr[\hat M^2 \mathcal{O}_1(\rho)]$ can be written in the Wigner representation as \cite[Eq.~39]{Polkovnikov2013}
    \begin{eqnarray}
        \Tr[\hat M^2 \mathcal{O}_1(\rho)] &=& \int d^m\bm q d^m \bm p (q_1^2 p_1^2+1)\frac{1}{(2\pi)^m \sqrt{\mathrm{det}(V^{\mathcal O_1(\rho)})}}\exp\left(-[q_1,\dots, q_m ,p_1, \dots ,p_m](V^{\mathcal O_1(\rho)})^{-1}[q_1,\dots, q_m ,p_1, \dots ,p_m]^T\right) \nonumber \\ && \hspace{-5mm} =1 + \int d^m\bm q d^m \bm p (q_1^2)\frac{1}{(2\pi)^m \sqrt{\mathrm{det}(V^{\mathcal O_1(\rho)})}}\exp\left(-[q_1,\dots, q_m ,p_1, \dots ,p_m](V^{\mathcal O_1(\rho)})^{-1}[q_1,\dots, q_m ,p_1, \dots ,p_m]^T\right) \nonumber \\ && \hspace{-2mm} \times  \int d^m\bm q d^m \bm p (p_1^2)\frac{1}{(2\pi)^m \sqrt{\mathrm{det}(V^{\mathcal O_1(\rho)})}}\exp\left(-[q_1,\dots, q_m ,p_1, \dots ,p_m](V^{\mathcal O_1(\rho)})^{-1}[q_1,\dots, q_m ,p_1, \dots ,p_m]^T\right) \nonumber \\ && \hspace{-5mm} +2 \left(\int d^m\bm q d^m \bm p (q_1 p_1)\frac{1}{(2\pi)^m \sqrt{\mathrm{det}(V^{\mathcal O_1(\rho)})}}\exp\left(-[q_1,\dots, q_m ,p_1, \dots ,p_m](V^{\mathcal O_1(\rho)})^{-1}[q_1,\dots, q_m ,p_1, \dots ,p_m]^T\right)\right) \nonumber \\ && \hspace{-5mm} = 1 + \Tr[\hat q_1^2 \rho] \Tr[\hat p_1^2 \rho] + 2\Tr[(\{\hat q_1,\hat p_1\}/2) \rho]^2.
    \end{eqnarray}
    The first line follows from the fact that the Wigner representation of $\hat M^2$ is $\hat q^2 \hat p^2 +1$. The second line follows from Isserli's Theorem or Wick's Theorem and finally, the last line follows from the fact that the Wigner space representation of $qp$ corresponds to the operator $\{\hat q, \hat p\}/2$ and that of $\hat{q}^2$ ($\hat{p}^2$) is given by $q^2$ ($p^2$). All these expressions for Wigner representations of operators polynomial in position and momentum operators can be derived using \cite[Eq.~39]{Polkovnikov2013}. Looking at the reduced covariance matrix of the first mode
    \begin{eqnarray}
        \Tr[\hat q_1^2 \rho] \Tr[\hat p_1^2 \rho] - \Tr[(\{\hat q_1,\hat p_1\}/2) \rho]^2 = (\nu_1^{(1)})^2,
    \end{eqnarray}
    where $\nu_1^{(1)}$ is the symplectic eigenvalue of the reduced covariance matrix of the first mode of $\mathcal O_1(\rho)$. This gives
    \begin{equation}
        \Tr[\hat M^2 \mathcal{O}_1(\rho)] = 1+ (\nu_1^{(1)})^2 + 3\Tr[(\{\hat q_1,\hat p_1\}/2) \rho]^2 =1 + (\nu_1^{(1)})^2 + 3\mu_1^2.
    \end{equation}
    This gives
    \begin{equation}
        \mathrm{Var}_{\mathcal O_1(\rho)}(\hat M) = 1 + (\nu_1^{(1)})^2 + 2\mu_1^2,
    \end{equation}
    and
    \begin{equation}
        N_1 = 272 \log\left(\frac{2}{\delta}\right) \frac{1 + (\nu_1^{(1)})^2 + 2\mu_1^2}{(\mu_2 - \mu_1)^2}.
    \end{equation}
    Similarly, using the median of means estimator $X_2$ we have
    \begin{equation}
         \Pr[|X_2 - \mu_2| \geq d/2] \leq \delta,
    \end{equation}
    as long as number of samples is greater than
    \begin{equation}
        N_2\coloneqq272 \log\left(\frac{2}{\delta}\right) \frac{1 + (\nu_1^{(2)})^2 + 2\mu_2^2}{(\mu_2 - \mu_1)^2}.
    \end{equation}
    Therefore,
    \begin{equation}
        p_\mathrm{err} \leq \frac12 \Pr[|X_1 - \mu_1| \geq d/2] + \frac12 \Pr[|X_2 - \mu_2| \geq d/2] \leq \frac{\delta}2+\frac{\delta}2 = \delta,
    \end{equation}
    as long as the number of samples satisfies
    \begin{equation}
        N \geq N_{\mathrm{thres}} = 272 \log\left(\frac{2}{\delta}\right)\max\left(\frac{1 + (\nu_1^{(1)})^2 + 2\mu_2^2}{(\mu_2 - \mu_1)^2},\frac{1 + (\nu_1^{(2)})^2 + 2\mu_2^2}{(\mu_2 - \mu_1)^2}\right).
    \end{equation}
Further, the reduced covariance matrix of the first mode can be written as
\begin{equation}
   V_{\mathrm{red}}^1 =  \nu_1 \begin{bmatrix}
        \cos \theta && \sin \theta \\ -\sin \theta && \cos \theta
    \end{bmatrix}\begin{bmatrix}
        e^{2r} && 0 \\ 0 && e^{-2r} 
    \end{bmatrix}\begin{bmatrix}
        \cos \theta && - \sin \theta \\  \sin \theta && \cos \theta
    \end{bmatrix},
\end{equation}
with $r\geq 0$. We also have
\begin{equation}
    \Tr[V_{\mathrm{red}}^1] = \nu_1(e^{2r} + e^{-2r}) \leq E- 2(m-1),
\end{equation}
where the last inequality comes from the fact that each of the rest of the individual modes have the trace of their reduced covariance matrix at least two (see Eq.\ \ref{eq:lowerboundtracecov}). Finally, since $e^{2r} + e^{-2r} \geq 2$ for all $r$, this gives 
\begin{equation}
    \nu_1 \leq E/2 \leq (E - 2(m-1))/2,
\end{equation}
and thus
\begin{equation}
    N_{\mathrm{thres}} \leq 272 \log\left(\frac{2}{\delta}\right)\max\left(\frac{1 + (E-2(m-1))^2/4 + 2\mu_1^2}{(\mu_2 - \mu_1)^2},\frac{1 + (E - 2(m-1))^2/4 + 2\mu_2^2}{(\mu_2 - \mu_1)^2}\right).
\end{equation}
We now simplify this upper bound on the number of sufficient samples, given that
\begin{equation}
    \mu_2^2 = (V_{xp}^{\mathcal O_2(\rho)})_{11}^2 \leq \sum_{i,j=1}^m (V_{xp}^{\mathcal O_2(\rho)})_{ij}^2 = || V^{\mathcal O_2(\rho)}||_F^2 \leq \mathfrak c_{\rho},
\end{equation}
and hence $|\mu_2| \leq \sqrt{\mathfrak c_\rho}$ and similarly, $|\mu_2| \leq \sqrt{\mathfrak c_\rho}$. We are left with finding the minimal value of the function
\begin{equation}\label{eq:Thm7_1}
    \frac{f(m,E) + \max(\mu_1^2,\mu_2^2)}{(\mu_2-\mu_1)^2}
\end{equation}
given that $\mu_2 \geq \mu_1$ and $|\mu_2|,|\mu_1|\leq \sqrt{\mathfrak c_\rho}$. Writing
\begin{equation}
    d = \mu_2 - \mu_1,
\end{equation}
Eq.~\ref{eq:Thm7_1} can be written as
\begin{equation}
    \frac{f(m,E) + \max(\mu_1^2,(d+\mu_1))^2)}{d^2}.
\end{equation}
The minimal value of $\max(\mu_1^2,(d+\mu_1))^2)$ is obtained for $\mu_1 = -d/2$ (to see this, note that for $\mu_1 < -d/2$, we have $\mu_1^2 > d^2/4 > (\mu_1+d)^2$, while for $\mu_1 > -d/2$, we have $(\mu_1+d)^2 > d^2/4 > \mu_1^2$) and is equal to 
\begin{equation}
     \frac{f(m,E)}{d^2} + \frac{1}{4}.
\end{equation}
Here, $d \leq 2\sqrt{\mathfrak c_\rho}$ with the upper bound saturated by $\mu_2 = -\mu_1 = \sqrt{\mathfrak c_\rho} $, therefore the upper bound on the sufficient number of samples may be refined as
\begin{equation}
      N_{\mathrm{thres}} \leq272 \log (2/\delta) \left(\frac{1 + (E - 2(m-1))^2/4 }{4\mathfrak c_\rho} + \frac14\right).
\end{equation}
\end{proof}
\subsection{Symplectic coherence for loss channels discrimination}\label{app:eff_loss_discrimination}
\noindent We restate Lemma \ref{lem:eff_loss_discrimination}:
\begin{lem}[Symplectic coherence for loss channels discrimination]
    Given a Gaussian input probe state $\rho$ with $\Tr[V^\rho] = E$, $\Tr[\rho \{\hat q_1, \hat p_1\}/2] = \mu$, and with $\nu$ the symplectic eigenvalue and $E_1$ the trace of the reduced covariance matrix of its first mode, two photon loss channels with transmissivities $\eta_1$ and $\eta_2$ can be distinguised with success probability $1-\delta$ by measuring $\hat M = \{\hat q_1, \hat p_1\}/2$, provided the number of samples of $\rho$ is greater than
\begin{equation}
    N_\mathrm{thres} = \frac{272 \log(2/\delta)}{(\eta_2 - \eta_1)^2} 
    \max\big( g(\nu,\mu,E,\eta_1),\, g(\nu,\mu,E,\eta_2) \big),
\end{equation}
where
\begin{equation}
    g(\nu,\mu,E,\eta) \coloneqq 
    \frac{
        \eta^2 \nu^2 + (1-\eta)^2 + \eta(1-\eta) E + 2\eta^2 \mu^2
    }{
        \mu^2
    }.
\end{equation}
For general $\eta_1$ and $\eta_2$, this is lower bounded by
\begin{equation}
    N_\mathrm{thres} \ge 
    \frac{272 \log(2/\delta)}{(\eta_2 - \eta_1)^2} 
    \max\big( \tilde g(m,E,\eta_1),\, \tilde g(m,E,\eta_2) \big),
\end{equation}
with
\begin{equation}
    \tilde g(m,E,\eta) \coloneqq 
    \frac{
        \eta^2 + (1-\eta)^2
    }{
        \mathfrak c_{\max}(m,E)
    } 
    + 2\eta^2,
\end{equation}
where $\mathfrak c_{\max}(m,E)$ is given by Eq.~\ref{appeq:max_sc}. 
This bound is saturated by a state of maximal symplectic coherence
\begin{equation}
    \ket{MSC} = 
    \hat R(\pi/4) \ket{r} \otimes \ket{0}^{\otimes (m-1)},
\end{equation}
where $\ket{r}$ is a momentum-squeezed state ($r>0$), with $r$ satisfying Eq.~\ref{appeq:focus_squeeze}.
\end{lem}

\begin{proof}
Under the effect of a photon loss channel $\Lambda_\eta(.)$, the covariance matrix transforms as (\cref{app:sc_under_loss})
\begin{equation}
    V^{\Lambda_\eta(\rho)} = \eta V^\rho + (1-\eta)\I.
\end{equation}
Therefore,
\begin{equation}
    \Tr[\Lambda_{\eta}(\rho)\hat M] = \eta \mu,
\end{equation}
and we denote 
\begin{equation}
    d \coloneqq \mu (\eta_2 - \eta_1),
\end{equation}
assuming $\eta_2 > \eta_1$ without loss of generality. We again use protocol \ref{Prot:channel_discrimination} for discrimination. Following the proof steps of Lemma \ref{lem:eff_ortho_discrimination}, the probability of error using Protocol \ref{Prot:channel_discrimination} satisfies
    \begin{equation}
        p_\mathrm{err} \leq \delta,
    \end{equation}
    as long as the number of samples is greater than
    \begin{equation}
        N \geq 272 \log\left(\frac{2}{\delta}\right)\max\left(\frac{1 + (\nu_1^{(1)})^2 + 2\eta_1^2 \mu^2}{\mu^2(\eta_2 - \eta_1)^2},\frac{1 + (\nu_1^{(2)})^2 + 2\eta_2^2\mu^2}{\mu^2(\eta_2 - \eta_1)^2}\right).
    \end{equation}
Now, if the reduced first mode covariance matrix of the first mode of the input probe state $\rho$ is given by
\begin{equation}
    \begin{bmatrix}
        \sigma_x && \mu \\
        \mu && \sigma_p
    \end{bmatrix},
\end{equation}
then the symplectic eigenvalue of the reduced covariance matrix of the first mode of $\Lambda_{\eta_1}(\rho)$ is given by
\begin{equation}
    (\nu_1^{(1)})^2 = (\eta_1 \sigma_x + (1-\eta_1))(\eta_1 \sigma_p + (1-\eta_1)) - \eta_1^2 \mu^2 = \eta_1^2 (\sigma_x \sigma_p - \mu^2) + (1-\eta_1)^2 + \eta_1 (1-\eta_1)(\sigma_x + \sigma_p) = \eta_1^2 \nu^2 + (1-\eta_1)^2 + \eta_1 (1-\eta_1)E_1.
\end{equation}
where we used $V^{\Lambda_\eta(\rho)} = \eta V^\rho + (1-\eta)\I$ and we denoted $E_1\coloneqq \sigma_x + \sigma_p$. With a similar analysis for $\nu_1^{(2)}$, we obtain
\begin{eqnarray}
      N \geq N_{\mathrm{thres}}\coloneqq 272 \log\left(\frac{2}{\delta}\right)\max\bigg(&&\frac{1 + \eta_1^2 \nu^2 + (1-\eta_1)^2 + \eta_1 (1-\eta_1)E_1 + 2\eta_1^2 \mu^2}{\mu^2(\eta_2 - \eta_1)^2}, \nonumber \\&&\frac{1 + \eta_2^2 \nu^2 + (1-\eta_2)^2 + \eta_2 (1-\eta_2)E_1 + 2\eta_2^2\mu^2}{\mu^2(\eta_2 - \eta_1)^2}\bigg).
\end{eqnarray}
Since $\nu \geq 1$ and for fixed $E_1$, both the functions inside the max function are decreasing in $\mu$. Since $\mu^2 \leq E_1^2/4 -1$, we have
\begin{eqnarray}
    N_{\mathrm{thres}}\geq \frac{272}{(\eta_2 - \eta_1)^2} \log\left(\frac{2}{\delta}\right)\max\bigg(\frac{1 + \eta_1^2 \nu^2 + (1-\eta_1)^2 + \eta_1 (1-\eta_1)E_1 }{E_1^2/4 -1} + 2\eta_1^2,&& \nonumber \\
    &&\hspace{-20mm}\frac{1 + \eta_2^2 \nu^2 + (1-\eta_2)^2 + \eta_2 (1-\eta_2)E_1 + 2\eta_2^2\mu^2}{E_1^2/4 -1} + 2\eta_2^2\bigg). \nonumber \\
\end{eqnarray}
Since now both the functions inside the max function are decreasing in $E_1$ we would want to set $E_1$ to its maximal possible value, i.e. $E_1 = E - 2(m-1)$ (the reduced covariance matrices of the rest of the modes should have traces at least $2$ by Eq.\ \ref{eq:lowerboundtracecov}). The state which saturates all these lower bounds $\nu = 1$, $\mu =E_1^2/4-1$ and $E_1 = E - 2(m-1)$ and hence gives the optimal number of sufficient samples is
\begin{equation}
    \ket \psi = \hat R(\pi/4) \ket{r} \otimes \ket{0}^{\otimes (m-1)},
\end{equation}
with the $\ket r$ being a squeezed state (squeezed in the $p$ direction, that is, with $r > 0$), with the squeezing parameter $r$ such that
\begin{equation}
    e^r + e^{-r} = (E-2(m-1)).
\end{equation}
With this input probe state, the number of sufficient samples is given by
\begin{eqnarray}
     \frac{272 \log(2/\delta)}{(\eta_2 - \eta_1)^2} \max \bigg(\frac{\eta_1^2 + (1-\eta_1)^2}{(E-2m)^2/4 + (E-2m)} + 2\eta_1^2 +&&  \frac{\eta_1(1-\eta_1)(E-2m)}{(E-2m)^2/4 + (E-2m)^2}, \nonumber \\ &&\hspace{-30mm}\frac{\eta_2^2 + (1-\eta_2)^2}{(E-2m)^2/4 + (E-2m)} + 2\eta_2^2 +  \frac{\eta_2(1-\eta_2)(E-2m)}{(E-2m)^2/4 + (E-2m)}\bigg).
\end{eqnarray}
\end{proof}
\subsection{Position-momentum correlations determine the distinguishability of certain Gaussian channels}\label{app:rotated_quadrature_sc}
\noindent To recap, we define \textsf{PP-MM} equivalent Gaussian channels:
\begin{defi}[$\textsf{PP-MM}$ equivalence between Gaussian channels]\label{appdefi:PP-MM_equivalence}
    Two $m$-mode Gaussian channels $\mathcal{G}_1$ and $\mathcal{G}_2$ are said to be $\textsf{PP-MM}$ equivalent if given an arbitrary $m$-mode quantum state $\rho$ with covariance matrix $V^\rho$, they change the diagonal block matrices of $V^\rho$ the same, but the off-block diagonal matrices are changed differently.
\end{defi} 

\noindent Then we have the following Lemma:
\begin{lem}[Position-momentum correlations determines distinguishability for $\textsf{PP-MM}$ equivalent channels]\label{apptheo:rotated_quadrature_sc}
     Given an input single-mode probe state $\rho$ with covariance matrix $V^\rho$ and two $\textsf{PP-MM}$ equivalent Gaussian channels $\mathcal G_1$ and $\mathcal G_2$ (Definition \ref{appdefi:PP-MM_equivalence}), the total variation distance between probability distributions generated by rotated quadrature measurement $\hat q_\theta$
    \begin{equation}
        \hat q_\theta = \hat q \cos \theta + \hat p \sin \theta
    \end{equation}
    is upper bounded by
    \begin{equation}
        \min(1,h_\theta(\mathcal G_1(\rho)),h_\theta(\mathcal G_2(\rho))), \nonumber\\
    \end{equation}
    where
    \begin{equation}
        h_{\theta}(\mathcal G(\rho)) \coloneqq \frac{|\sin(2\theta)||\sigma_{xp}^{\mathcal{G}_1(\rho)} - \sigma_{xp}^{\mathcal{G}_2(\rho)}|}{|\sigma_x \cos^2 \theta + \sigma_p \sin^2 \theta + \sigma_{xp}^{\mathcal G(\rho)} \sin 2\theta|},
    \end{equation}
    and $\sigma_{xp}^{\mathcal{G}_1(\rho)}$ is off-diagonal element of the covariance matrix of $\mathcal G_1(\rho)$, and similarly for $\sigma_{xp}^{\mathcal{G}_1(\rho)}$.
\end{lem}
\begin{proof}
    The rotated quadrature measurement $\hat q_\theta$ on $\mathcal G(\rho)$ can be seen as rotating the quantum state $\mathcal G(\rho)$ by $\theta$ and performing a measurement of the position quadrature.

    If the covariance matrix of $\mathcal G(\rho)$ is given by
    \begin{equation}
        V^{\mathcal G(\rho)} = \begin{bmatrix}
            \sigma_x && \sigma_{xp}^{\mathcal{G}(\rho)} \\
            \sigma_{xp}^{\mathcal{G}(\rho)} && \sigma_p
        \end{bmatrix},
    \end{equation}
    then the covariance matrix of $\hat R(\theta) \mathcal G(\rho) \hat R^\dagger (\theta)$ is given by
    \begin{eqnarray}
        &&\begin{bmatrix}
            \cos \theta && \sin \theta \\
            -\sin \theta && \cos \theta
        \end{bmatrix}
        \begin{bmatrix}
            \sigma_x && \sigma_{xp}^{\mathcal{G}(\rho)} \\
            \sigma_{xp}^{\mathcal{G}(\rho)} && \sigma_p
        \end{bmatrix}
           \begin{bmatrix}
            \cos \theta && -\sin \theta \\
            \sin \theta && \cos \theta
        \end{bmatrix} \nonumber \\ &&= \begin{bmatrix}
            \sigma_x \cos^2 \theta + \sigma_p \sin^2 \theta + \sigma_{xp}^{\mathcal G(\rho)} \sin 2\theta && -\sigma_x \sin \theta \cos \theta + \sigma_p \cos \theta \sin \theta + \sigma_{xp} (\cos^2 \theta - \sin^2 \theta) \\ 
            -\sigma_x \sin \theta \cos \theta + \sigma_p \cos \theta \sin \theta + \sigma_{xp} (\cos^2 \theta - \sin^2 \theta) && \sigma_x \sin^2 \theta + \sigma_p \cos^2 \theta - \sigma_{xp}^{\mathcal G(\rho)} \sin 2\theta
        \end{bmatrix}. \nonumber \\
    \end{eqnarray}
    Therefore, the measurement of $\hat q$ on $\hat R(\theta) \mathcal G_1(\rho) \hat R^\dagger (\theta)$ samples a random variable from a normal distribution with zero mean and variance 
    \begin{equation}
        \sigma_1 = \sigma_x \cos^2 \theta + \sigma_p \sin^2 \theta + \sigma_{xp}^{\mathcal G_1(\rho)} \sin 2\theta,
    \end{equation}
    and similarly, the measurement of $\hat q$ on $\hat R(\theta) \mathcal G_2(\rho) \hat R^\dagger (\theta)$  samples a random variable from a normal distribution with zero mean and variance
    \begin{equation}
        \sigma_2 = \sigma_x \cos^2 \theta + \sigma_p \sin^2 \theta + \sigma_{xp}^{\mathcal G_2(\rho)} \sin 2\theta.
    \end{equation}
    From \cite[Theorem 1.1]{devroye2023}, the total variation distance between two single-variate normal distributions $\mathcal N_1$ and $\mathcal N_2$ of variances $\sigma_1$ and $\sigma_2$ satisfies
    \begin{equation}
        \mathrm{TVD}(\mathcal N_1,\mathcal N_2) \leq \frac32 \min(\min\left(1,|\sigma_1/\sigma_2 - 1|\right),\min\left(1,|\sigma_2/\sigma_1 - 1|\right)).
    \end{equation}
    Now,
    \begin{equation}
        \frac{\sigma_1}{\sigma_2} - 1 = \frac{\sin 2\theta (\sigma_{xp}^{\mathcal G_1(\rho)} - \sigma_{xp}^{\mathcal G_2(\rho)})}{ \sigma_x \cos^2 \theta + \sigma_p \sin^2 \theta + \sigma_{xp}^{\mathcal G_2(\rho)} \sin 2\theta},
    \end{equation}
    and similarly for $\sigma_2/\sigma_1 - 1$. Therefore, the total variation distance generated by rotated quadrature measurements $\hat q_\theta$ on $\mathcal G_1(\rho)$ and $\mathcal G_2(\rho)$ is given by
    \begin{equation}
        \mathrm{TVD}_{\hat q_\theta} \leq \min(\min(1,h_\theta(\mathcal G_1(\rho))),\min(1,h_\theta(\mathcal G_2(\rho))), \nonumber\\
    \end{equation}
    where
    \begin{equation}
        h_{\theta}(\mathcal G(\rho)) \coloneqq \frac{|\sin(2\theta)||\sigma_{xp}^{\mathcal{G}_1(\rho)} - \sigma_{xp}^{\mathcal{G}_2(\rho)}|}{|\sigma_x \cos^2 \theta + \sigma_p \sin^2 \theta + \sigma_{xp}^{\mathcal G(\rho)} \sin 2\theta|}.
    \end{equation}
\end{proof}

\end{document}